\documentclass[10pt]{article}

\usepackage[top=1.3in, bottom=1.3in, left=1.5in, right=1.5in]{geometry}

\usepackage{amsmath, amssymb, amsfonts, amsthm, amscd, mathtools, graphicx, mathtext,color, wrapfig, verbatim, slashed}
\usepackage{yfonts}

\usepackage{mathrsfs}
\usepackage{pifont}

\usepackage{graphicx}
\usepackage[font={small,sl}]{caption}
\usepackage[all]{xy}

\usepackage{tikz}
\usetikzlibrary{matrix,arrows}

\usepackage{lipsum}

\let\OLDthebibliography\thebibliography
\renewcommand\thebibliography[1]{
  \OLDthebibliography{#1}
  \setlength{\parskip}{0pt}
  \setlength{\itemsep}{0pt plus 0.0ex}
}
\makeatletter

\@addtoreset{equation}{section}
\makeatother

\newcommand{\beq}{\begin{equation}}
\newcommand{\eeq}{\end{equation}}

\newcommand{\bF}{\ensuremath{\mathbb{F}}}

\newcommand{\bP}{\ensuremath{\mathbb{P}}}

\newcommand{\rd}{/\!\!/\!\!/\!\!/}
\newcommand{\rdd}{/\!\!/}

\newcommand{\scF}{\ensuremath{\mathcal{F}}}

\newcommand{\scQ}{\ensuremath{\mathcal{Q}}}

\newcommand{\scX}{\ensuremath{\mathcal{X}}}

\newcommand{\fg}{\mathfrak{g}}
\newcommand{\fh}{\mathfrak{h}}
\newcommand{\fgh}{\widehat{\mathfrak{g}}}
\newcommand{\fC}{\mathfrak{C}}
\newcommand{\fP}{\mathfrak{P}}

\newcommand{\fR}{\mathfrak{R}}
\newtheorem{theorem}{Theorem}[section]

\newtheorem{conjecture}[theorem]{Conjecture}

\DeclareMathOperator{\Stab}{Stab}

\DeclareMathOperator{\Def}{Def}
\DeclareMathOperator{\Hom}{Hom}

\DeclareMathOperator{\Aut}{Aut}
\DeclareMathOperator{\Obs}{Obs}
\DeclareMathOperator{\Lie}{Lie}

\DeclareMathOperator{\Ell}{Ell}

\DeclareMathOperator{\pt}{pt}

\DeclareMathOperator{\Pic}{Pic}
\DeclareMathOperator{\tr}{tr}

\DeclareMathOperator{\ev}{ev}
\DeclareMathOperator{\Vx}{\mathbf{V}}
\DeclareMathOperator{\tVx}{\widetilde{\Vx}}
\DeclareMathOperator{\QM}{\mathsf{QM}}
\DeclareMathOperator{\End}{End}

\DeclareMathOperator{\Proj}{Proj}

\newcommand{\Ct}{\mathbb{C}^\times}

\newcommand{\bbV}{\mathbb{V}}
\newcommand{\bbW}{\mathbb{W}}

\newcommand{\aroof}{\widehat{\mathsf{a}}}

\newcommand{\Hd}{{H}^{\raisebox{0.5mm}{$\scriptscriptstyle \bullet$}}}

\newcommand{\Cp}{\mathsf{Cap}}
\newcommand{\Glue}{\mathsf{Glue}}
\newcommand{\Fusion}{\mathsf{J}}
\newcommand{\Vertex}{\mathsf{Vertex}}
\newcommand{\tO}{\widehat{\mathscr{O}}}
\newcommand{\vir}{\textup{vir}}

\newcommand{\cF}{\mathscr{F}}
\newcommand{\tX}{\widetilde{X}} 
\newcommand{\tcF}{\widetilde{\cF}}  
\newcommand{\bfC}{\mathbf{C}}

\newcommand{\cK}{\mathscr{K}}
\newcommand{\cL}{\mathscr{L}}
\newcommand{\cO}{\mathscr{O}}
\newcommand{\cT}{\mathscr{T}}

\newcommand{\cV}{\mathscr{V}}
\newcommand{\cM}{\mathscr{M}}
\newcommand{\cU}{\mathscr{U}}
\newcommand{\cG}{\mathscr{G}}
\newcommand{\cW}{\mathscr{W}}

\newcommand{\Q}{\mathbb{Q}}
\newcommand{\C}{\mathbb{C}}
\newcommand{\Z}{\mathbb{Z}}

\newcommand{\bA}{\mathsf{A}}
\newcommand{\bT}{\mathsf{T}}
\newcommand{\bS}{\mathsf{S}}

\newcommand{\be}{\mathbf{e}}
\newcommand{\bd}{\mathbf{d}}
\newcommand{\bla}{\boldsymbol{\lambda}}

\newcommand{\cHom}{\mathscr{H}\!\!\mathit{om}}
\newcommand{\cEnd}{\mathscr{E}\!\mathit{nd}}

\newcommand{\nc}{\newcommand}
\nc{\mc}{\mathcal}
\nc{\wh}{\widehat}
\nc{\ghat}{\wh\fg}
\nc{\on}{\operatorname}
\def\neg{\negthinspace}
\newcommand{\Lfgh}{\wh{^L\fg}}
\nc{\n}{{\mathfrak n}}
\nc{\pone}{{\mathbb C}{\mathbb P}^1}
\nc{\la}{\lambda}
\nc{\al}{\alpha}

\DeclareMathOperator{\bVx}{\overline{\Vx}}

\newtheorem{Proposition}{Proposition} 
\newtheorem{Lemma}{Lemma} 
\newtheorem{Corollary}{Corollary} 
\newtheorem{Theorem}{Theorem}

\theoremstyle{definition}
\newtheorem{remark}[theorem]{Remark}
\def\tr{{\rm tr}\,}

\begin{document}
\baselineskip=28pt  % a la harvmac
\baselineskip 0.7cm

\begin{titlepage}

%% Set the number of the title with 0

% change the footnote symbol
\renewcommand{\thefootnote}{\fnsymbol{footnote}}

\begin{center}

\vspace*{10mm}
\vskip 1in
{\LARGE \bf
Quantum $q$-Langlands Correspondence
}
%\vskip 0.5cm
\vspace*{10mm}

{\large
Mina Aganagic$^{1,2}$, Edward Frenkel$^{2}$, Andrei Okounkov$^{3}$
}
\\
\medskip

\vskip 0.5cm

{\it
$^1$Center for Theoretical Physics, University of California, Berkeley\\
$^2$Department of Mathematics, University of California, Berkeley\\
$^3$Department of Mathematics, Columbia University 
}

\vskip 1cm

{\em Dedicated to Ernest Vinberg on the occasion of his 80th birthday}

\end{center}

\vskip 1cm

%-----------------------------------------
\centerline{{\bf Abstract}}
\medskip

We conjecture, and prove for all simply-laced Lie algebras, an
identification between the spaces of $q$-deformed conformal blocks for
the deformed ${\mc W}$-algebra ${\mc W}_{q,t}(\fg)$ and quantum affine
algebra of $\wh{^L\fg}$, where $^L\fg$ is the Langlands dual Lie
algebra to $\fg$. We argue that this identification may be viewed as a
manifestation of a $q$-deformation of the quantum Langlands
correspondence. Our proof relies on expressing the $q$-deformed
conformal blocks for both algebras in terms of the quantum K-theory of
the Nakajima quiver varieties. The physical origin of the isomorphism
between them lies in the 6d little string theory. The quantum
Langlands correspondence emerges in the limit in which the 6d little
string theory becomes the 6d conformal field theory with (2,0)
supersymmetry.

\noindent\end{titlepage}
\setcounter{page}{1} % don't number title page

\setcounter{tocdepth}{2}
\tableofcontents

%\setcounter{section}{-1}
%%%%%%%%%
\section{Introduction}\label{sec:zero}

\subsection{Overview} 

% Langlands Program was launched by Robert Langlands 50 years ago with
% the goal of relating certain questions in number theory and harmonic
% analysis. It was subsequently generalized to the geometric setting of
% a complex projective algebraic curve ${\mc C}$.

In the 50 years of its existence, the Langlands program and 
the Langlands philosophy have grown to encompass many 
objects of central importance to both mathematics and 
mathematical physics.

In particular, the geometric Langlands correspondence 
starts with a complex projective algebraic curve ${\mc C}$
with the goal, as it is usually understood today, to 
prove an equivalence between certain categories 
associated to a pair $G$, $^LG$ of Langlands dual connected reductive
  complex Lie groups. These are certain categories 
of sheaves
  (of ${\mc D}$-modules and ${\mc O}$-modules, respectively) on the
  moduli stack $\on{Bun}_{^LG}$ of $^LG$-bundles on ${\mc C}$ and
  the moduli stack $\on{Loc}_G$ of flat $G$-bundles on ${\mc C}$.\footnote{The existence of such an equivalence,
    which may be viewed as a categorical non-abelian Fourier
    transform, was originally proposed by Beilinson and Drinfeld;
    later, a precise conjecture was formulated in \cite{AG}. We note
    that some of our notation bucks the usual conventions. In
    particular, the roles of $G$ and $^LG$ are exchanged.}  Kapustin
  and Witten have shown \cite{KW} that this equivalence is closely
  related to $S$-duality of maximally supersymmetric 4d gauge theories
  with gauge groups being the compact forms of $G$ and $^LG$.

 Beilinson and Drinfeld have constructed in \cite{BD} an important part of the geometric Langlands correspondence using the isomorphism
\cite{FF} between the center of the (chiral) affine Kac--Moody
 algebra $\wh{^L\fg}$ at the critical level $^Lk=-^Lh^\vee$ and the
 classical ${\mc W}$-algebra ${\cal W}_\infty({\fg})$. Their
 construction is closely connected to the 2d conformal field theory and
the theory of chiral (or vertex) algebras (see \cite{Frenkel} for a survey; and also  \cite{Wittengrass} in which an
 analogy between 2d CFT and the theory of automorphic representations
 was first observed and investigated).

Since the level of $\wh{^L\fg}$ may be deformed away from the critical
value, and at the same time ${\cal W}_\infty({\fg})$ may be
deformed to the quantum ${\mc W}$-algebra ${\cal W}_\beta({\fg})$, one
is naturally led to look for a quantum deformation of the geometric
Langlands correspondence.
  
Many interesting structures have emerged in the studies under
the umbrella of ``quantum geometric Langlands'' (from the point of
view of 2d CFT \cite{F, Frenkel1, FFS, Stoyanovsky, JT, gaitsW,
  gaitsQ, Sch}; in
the framework of 4d gauge theory \cite{KW,Kapustin,FGT}; and, in the
abelian case, as a deformation of the Fourier--Mukai transform
\cite{PR}).

\subsubsection{}

For us, the main feature of the quantum geometric
  Langlands correspondence is an isomorphism between
the spaces of conformal blocks of certain representations of two
chiral algebras: \beq\label{geomL} \widehat{^L\fg}_{{}^Lk} \qquad
\longleftrightarrow \qquad {\cal W}_{\beta}({\fg}), \eeq the affine
Kac--Moody algebra of $^L{\fg}$ at level ${}^Lk$ and the ${\cal
  W}$-algebra ${\cal W}_{\beta}({\fg})$. The algebra ${\cal
  W}_{\beta}({\fg})$ is obtained by the quantum Drinfeld--Sokolov
reduction \cite{FF:ds,BO,FF} of the affine algebra $\ghat$ at level
$k$, where $\beta = m(k+h^\vee)$ in the notation of
\cite{FR}.\footnote{Thus, what we denote here by ${\cal
    W}_{\beta}(\fg)$ is ${\cal W}_{k}(\fg)$ of \cite{FB,Frenkel},
  where $\beta = m(k+h^\vee)$. In our present notation, the classical
  ${\cal W}$-algebra associated to $\fg$ is ${\cal
    W}_\infty(\fg)$. See Section \ref{glc} for more details.}

We will establish this isomorphism and prove a stronger result in the
case of simply-laced $\fg$ and genus zero curve ${\mc C}$: an identification of
conformal blocks of the two algebras if the parameters are generic and
related by the formula \beq\label{zero1} \beta - m =
\frac{1}{^L(k+h^\vee)}.  \eeq

The relation between the corresponding chiral algebras may be viewed
as a strong/weak coupling transformation. Indeed, if we define $\tau =
\beta/m$ and $^L\tau = - {^L(k+h^\vee)}$, then \eqref{zero1} says that
\beq\label{firstl} \tau -1 = -1/(m ^L\tau), \eeq and so $^L\tau$ near
zero corresponds to large values of $\tau$. The parameters $\tau$ and
$^L\tau$ are related to the complexified coupling constants of the two
$S$-dual 4d Yang-Mills theories. Note the shift $\tau \mapsto \tau-1$,
as compared to the ${\cal W}$-algebra duality formula of \cite{FF}
(see Section \ref{glc} for more details). This is a shift of the theta
angle from 4d gauge theory perspective (see Section \ref{sec:nine}).

Here by identification of the spaces of conformal blocks we mean a
canonical isomorphism between them. However, this canonical
isomorphism arises only after we introduce one more parameter and
perform one more deformation.

\subsubsection{}

We consider a {\em two-parameter deformation} of the geometric
Langlands correspondence: the $q$-defor\-mation together with the
deformation away from the critical level. This turns out to be a
productive point of view.

Namely, we replace the above chiral algebras with their deformed
counterparts: the quantum affine algebra $U_{{\hbar}}(\wh{^L\fg})$, 
which is an
$\hbar$-deformation of the universal enveloping algebra of
$\wh{^L\fg}$ introduced by Drinfeld and Jimbo \cite{Drinfeld,Jimbo},
and the deformed ${\mc W}$-algebra ${\cal
  W}_{q,t}({\fg})$ introduced in \cite{FR} (see also
\cite{SKAO,FF:qw,AKOS} for $\fg=sl_n$), which is a deformation of ${\mc
  W}_\beta(\fg)$. We will refer to both of these as
``$q$-deformations'', both for brevity and because $q$ will appear as
a step in difference equations that are of principal importance to
us. (In our notation, the quantum affine algebra
  $U_{{\hbar}}(\wh{^L\fg})$
  becomes the enveloping algebra of $\wh{^L\fg}$ in the limit $\hbar
  \rightarrow 1$; this agrees with the notation
  used in \cite{OK}. For a fixed non-critical value of $^Lk$, this
  limit is the same as the limit $q\rightarrow 1$.) 

We focus on the case that the curve ${\mc C}$ is an infinite cylinder,
$$
{\cal C}\cong \Ct \cong \textup{infinite cylinder}. 
$$
It should be noted that integrable deformations away from the
conformal point are unlikely to exist unless ${\cal C}$ is flat. The torus case should follow from the case
of the cylinder, by imposing periodic identifications.\footnote{To get
  the deformed conformal blocks on a torus ${\cal C} ={\mathbb
    C^\times}/p^{\mathbb Z}$, one would study with blocks on ${\cal C}
  ={\mathbb C}^{\times}$, but with insertions that are invariant under
  the action of $p^{\mathbb Z}$.} The case when ${\cal C}$ is a plane can be obtained from ours, by taking the radius of the cylinder to infinity.

We conjecture (and prove in the simply-laced case) a correspondence
between $q$-deformed conformal blocks of the quantum affine algebra
$U_{\hbar}(\wh{^L\fg})$ and the deformed ${\cal W}$-algebra
\beq\label{qLanglands}
U_{\hbar}(\wh{^L\fg}) \qquad  \longleftrightarrow \qquad {\cal W}_{q,t}({\fg}),
\eeq
where the parameters 
\beq\label{zero}
q=\hbar^{-^L(k+h^{\vee})}, \qquad t=q^{\beta},
\eeq
are generic and related by the formula 
\beq\label{first}
 t = q^m/ \hbar
\eeq
which yields \eqref{firstl}. 

It is this identification of the deformed conformal blocks that we
refer to as a ``quantum $q$-Langlands correspondence'' in the title of
the present paper.

The physical setting for the correspondence is a six-dimensional
string theory, called the ``$(2,0)$ little string theory''.  The little
string theory \cite{SeibergN, Mm} is a one-parameter deformation of the ubiquitous 6d
$(2,0)$ superconformal theory (see e.g. \cite{Wittenl1}). The
deformation corresponds to giving strings a non-zero characteristic size, and
``converts'' the relevant chiral algebras, such as $\ghat$ and ${\mc
  W}_\beta(\fg)$, into the corresponding deformed algebras.   
  
\subsubsection{}

Some preliminary remarks about deformed conformal
  blocks are in order. In the case of an affine Kac--Moody algebra and a cylinder ${\mc C}$,
the space of conformal blocks is isomorphic to the space of solutions
of the Kniznik-Zamolodchikov (KZ) equations, which behave well as the insertion points are taken to infinity. The space of $q$-deformed
conformal blocks for quantum affine algebras can be similarly defined, following
\cite{FIR}, as the space of solutions of the quantum
Kniznik-Zamolodchikov (qKZ) equations. In either case, there is a particular fundamental solution of the equations which comes from sewing chiral vertex operators. This solution is given by \eqref{electric} in the case of   
deformed conformal blocks of $U_{\hbar}(\Lfgh)$.

To the best of our knowledge, the definition of the
  space of deformed conformal blocks for the deformed ${\mc
    W}$-algebra ${\cal W}_{q,t}({\fg})$ was not available in the
  literature until now. The blocks formally equal correlation
  functions of free field vertex operators of the deformed ${\mc
    W}_{q,t}({\fg})$ algebra in \eqref{magnetic}, constructed in
  \cite{FR}, however the definition is not complete. One has yet to
  specify the space of allowed contours of integration for screening
  charges. Further, the analogues of the qKZ equations were
  previously unknown for the deformed ${\mc W}$-algebras ${\mc
    W}_{q,t}({\fg})$, as far as we know.

One of the results in this paper is a definition of the space of deformed ${\mc W}_{q,t}({\fg})$ algebra conformal blocks, and a characterization of the difference equations they satisfy. The key new insight is the geometric interpretation of these objects in terms of (quantum) $K$-theory of a Nakajima quiver variety $X$ \cite{OK}, whose quiver diagram is based on the Dynkin diagram of ${\fg}$.

\subsection{Statement of the correspondence} 
\label{s_statement}

Let $x$ be a coordinate on ${\cal C}\cong \Ct$. Fix a finite
collection of distinct points on ${\cal C}$, with coordinates $a_i$. 
We propose, and prove in the simply-laced case, 
a correspondence between the following two types of $q$-conformal blocks on ${\cal C}$.

\subsubsection{} 

On the \emph{electric} side, we consider the quantum affine algebra
$U_{\hbar}(\Lfgh)$ blocks \cite{FIR}
\beq\label{electric}
\langle \lambda'| \;\prod_{i} { \Phi}_{{}^L\!\rho_i} (a_{i})  \;|\lambda \rangle
\eeq
where ${\Phi}_{{}^L\!\rho}(x)$ is a chiral vertex operator corresponding to a
finite-dimensional $U_{\hbar}({{\Lfgh}})$-module $^L\!\rho$. 
The state $|\lambda\rangle
$ is the highest weight vector in a level $^Lk$ Verma module. 
Its weight $\lambda \in {^L\fh}^*$ is an element of the dual of the
Cartan subalgebra for $^L{\fg}$. This is illustrated in 
Figure \ref{f_cylinder}. 

\begin{figure}[!hbtp]
  \centering
   \includegraphics[scale=0.5]{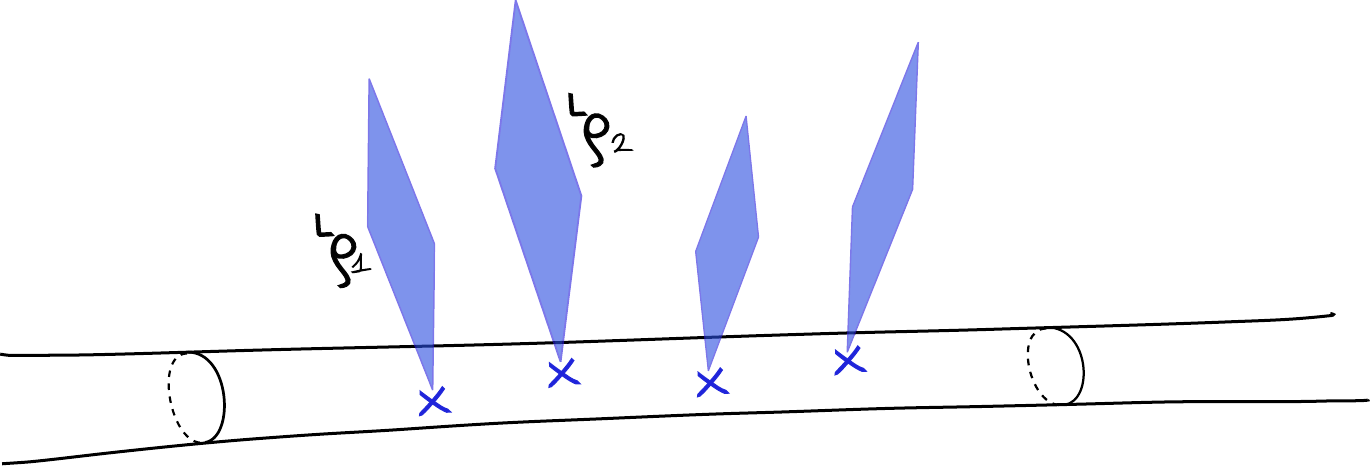}
 \caption{The cylinder ${\cal C}$ with the insertions of 
vertex operators corresponding to 
finite-dimensional $U_{\hbar}({{\Lfgh}})$-modules $^L\!\rho_i$
at the points $a_i\in {\cal C}$. 
Boundary conditions at infinity are 
the highest weight vectors 
 $\langle \lambda'|$ and $|\lambda \rangle$.}
  \label{f_cylinder}
\end{figure}

It suffices to focus on vertex operators corresponding to the 
fundamental representations because all others may be 
generated from these, by fusion. The highest weight of a 
 fundamental representation is one of the fundamental weights 
$^Lw_a$ of $^L{\fg}$. The conformal block  \eqref{electric}
takes values in a weight subspace of 
$$ 
\otimes_{ i} \, ({^L\!\rho}_i)= \otimes_a \,( ^L\!\rho_a)^{\otimes
  m_a} \,,
$$
namely, it has 
\begin{align}
  \textup{weight} & = \lambda' - \lambda \notag \\
 &= \label{weight}
\sum_a m_a\, ^Lw_a - \sum_a d_a \,^Le_a\,, 
\quad d_a \ge 0 \,. 
\end{align}
In \eqref{weight}, we write the weight as the difference 
of the highest weights and simple positive roots $^Le_a$ of 
$^L{\fg}$. The index $a$ runs here from
$1$ to ${\rm rk}({\fg})$.

\subsubsection{}

On the \emph{magnetic} side, we consider $q$-correlators of the ${\cal
  W}_{q,t}({\fg})$ algebra of the form
\beq\label{magnetic}
\langle \mu'| \;\prod_{i} {V}^{\vee}_ i(a_{i}) \prod_a \Bigl({Q}_a^{\vee}\Bigr)^{d_a} \;|\mu \rangle.
\eeq
$V^{\vee}_a(x)$ and $Q^{\vee}_a$ are the vertex and the screening
charge operators defined by E. Frenkel and N. Reshetikhin in
\cite{FR}. They are labeled by coroots and coweights of ${\fg}$,
respectively.  Recall that Langlands duality maps coweights and
coroots of ${\fg}$ to weights and roots of $^{L}{\fg}$, respectively.
The screening charge operators are defined as integrals of screening
current vertex operators $Q^{\vee}_a = \int dx \; S_a^{\vee}(x)$, so
\eqref{magnetic} is implicitly an integral formula for ${\cal
  W}_{q,t}({\fg})$ algebra blocks.

The coweights of ${\fg}$ labeling $V^{\vee}_a(x)$ are the highest
weights of the fundamental representations of $^L{\fg}$.  The operator
${V}^{\vee}_ i(a_{i})$, inserted at a point on ${\cal C}$ with the
coordinate $a_i$, is associated to the same representation of
$^L{\fg}$ as the corresponding vertex operator in
\eqref{electric}. The state $|\mu\rangle$, labeled by an element $\mu
\in {\fh}$ of the Cartan subalgebra of ${\fg}$, generates an
irreducible Fock representation of the ${\cal W}_{q,t}({\fg})$ algebra
\cite{FR}. The (co)weights $\mu$ and $\mu'$ are
  determined by $\lambda$ and $\lambda'$ (the exact formula depends on
  the chosen normalization).

% Note that Langlands duality identifies ${\fh}$ with $^L{\fh}^*$. For
% review of all necessary aspects of the relation between representation
% theory of ${\fg}$ and $^L{\fg}$ see, for example, the appendix of
% \cite{GuW}.

\subsubsection{}

The key result of the paper is the following theorem:
\begin{Theorem}\label{t:one} Let ${\fg}$ be a simply laced Lie algebra.
The deformed conformal blocks of
  $U_{\hbar}({{\Lfgh}})$ in \eqref{electric} and the deformed
  conformal blocks of ${\mc W}_{q,t}({\fg})$ in \eqref{magnetic} are identified 
   by the  
   \beq\begin{aligned}\label{linear}
\textup{specific covector} \times \textup{ $U_{\hbar}(\Lfgh)$ conformal block} =
 \textup{${\cal W}_{q,t}(\fg)$ algebra block},
\end{aligned}
\eeq
provided that the parameters of the two
  algebras are generic and related by equation \eqref{zero}.
  \end{Theorem}

The covector in \eqref{linear}, as well other ingredients of Theorem 
\ref{t:one} are best explained in geometric terms, namely, in terms
of the (quantum) K-theory of a Nakajima quiver variety $X$, see
below. Specifically, the covector in question corresponds to the
insertion of the identity $\cO_X\in K_\bT(X)$ (more precisely, to no insertion) in a
certain enumerative
problem. In geometric representation theory literature, it is customary
to characterize $\cO_X$ by a certain Whittaker property under the
action of lowering operators, see e.g.\ \cite{MO} for a discussion in 
cohomology. While we did not pursue 
such characterization in the present paper, there is little doubt that
it can be given.

 We will also explain, following the predictions of string theory,
 what is the natural setting for the non-simply laced cases, see 
Section \eqref{s_folding}.  As certain crucial geometric
representation theory ingredients are missing in this case, we propose 
the non-simply laced analog of Theorem \ref{t:one} as a conjecture.

\subsection{{Geometry behind the correspondence}}
%\subsection{Quiver varieties and vertex functions}
\label{s_Nakajima}

The central ingredient of our proof is that for Lie algebras of simply-laced type, when
$$
^L{\fg} ={\fg},
$$
we can realize the $q$-conformal blocks \eqref{electric} and \eqref{magnetic}
as {\it vertex functions} in equivariant 
quantum K-theory of a
certain holomorphic symplectic variety $X$. The variety $X$ is the
Nakajima quiver variety with 
$$
\textup{quiver ${\scQ}$} = 
\textup{Dynkin diagram of ${\fg}$} \,.
$$

\subsubsection{}
A Nakajima quiver variety $X$ is a hyper-K\"ahler  quotient
(or a holomorphic symplectic reduction) 
\beq\label{NV}
X=T^*{{\rm Rep }\,\scQ} \rd G_{\scQ}\,, 
\eeq
where 
\beq\label{repQ}
{{\rm Rep }\,\scQ}=\oplus_{a \rightarrow b} \,{\rm Hom}(V_a, V_b) \oplus_{a} {\rm Hom}(V_a, W_a)
\eeq
and 
\beq\label{gaugeC}
G_{\scQ} =\prod_{a} GL(V_a), \qquad G_{W} = \prod_{a} GL(W_a) \,. 
\eeq
The arrows in \eqref{repQ} are the arrows of the quiver. The
dimensions of the vector spaces $V_a$ and $W_a$ correspond 
as follows 
$$
\dim V_a = d_a\,, \quad \dim W_a = m_a
$$
to the weight space data in \eqref{weight}.

\subsubsection{}

The quotient in \eqref{NV} involves a geometric invariant theory (GIT)
quotient, which depends on a choice of stability conditions.
As a result, vertex functions also depend on a stability
condition. This stability condition makes them analytic in a certain 
region of the K\"ahler moduli space of $X$. The transition matrix between
vertex functions and q-conformal blocks will similarly 
depend on the stability condition. This dependence will be understood in what follows. 

\subsubsection{}

The majority of variables in \eqref{electric} and \eqref{magnetic} 
become \emph{equivariant variables} in their
geometric interpretation. We have 
$$
G_W\times \Ct_\hbar \subset \Aut(X)
$$
where $\Ct_\hbar$ rescales the cotangent directions in \eqref{NV} with 
weight $\hbar^{-1}$. This gives the symplectic form on $X$ weight
$\hbar$ under $\Ct_\hbar$. We fix a maximal torus 
$A\subset G_W$ and denote 
$$
T = A \times \Ct_\hbar \,. 
$$
The coordinates $a_i$ of $A$ are the positions at which the vertex
operators are inserted in \eqref{electric} and \eqref{magnetic}, 
while $\hbar$ is the quantum group deformation parameter in 
$U_{\hbar}(\Lfgh)$. 

A multiplicative group $\Ct_q$ acts on quasimaps
$\bP^1\dasharrow X$ by automorphisms
of the domain $\bP^1$. The coordinate $q\in \Ct_q$ is the
$q$-difference parameter from the title of the paper.

\subsubsection{}

In \cite{Nakajima}, Nakajima identified $K_T(X)$ with a space of 
weight \eqref{weight} 
in a $U_{\hbar}(\Lfgh)$-module. This is an important result in geometric representation
theory which generated a lot of further research. In \cite{MO} the authors suggested a somewhat different approach to constructing geometric actions of quantum groups. One of its advantages is its transparent connection with \emph{quantum} cohomology and K-theory of $X$, see \cite{MO,OK}. 

By quantum cohomology and K-theory we mean enumerative
theories of curves in $X$. The precise flavor of such computations
depends on the exact setup of the enumerative problem, including the
choice of the moduli spaces in question.  Givental and collaborators developed a very general K-theoretic
analog of quantum cohomology using the moduli spaces of stable 
rational maps, see e.g.\ \cite{Giv}. This is not the theory that 
will be used here. The following features of the quantum 
K-theory used here will be important:
\begin{itemize}
\item[---] it deals with quasimaps to a GIT quotient as
  in \cite{CKM}, 
\item[---] the quotient \eqref{NV} is a holomorphic symplectic
  reduction of a cotangent bundle, 
\end{itemize}
see \cite{OK} for an introduction.

\subsubsection{}

The basic object of the theory of \cite{OK} is the vertex function ${\bf V}$. The vertex
  function is an equivariant $K$-theoretic count of quasimaps from ${\mathbb C}$ to $X$
 of all possible degrees.   It is an analog of Givental's ${\bf
  I}$-function. 
The variables $z$ in this generating function
are called \emph{K\"ahler parameters}. They are related to the choice 
of the Fock vacuum  $|\lambda\rangle$ in \eqref{electric} and $|\mu\rangle$ in
\eqref{magnetic}.  
Its definition and basic properties will be reviewed
in Section \ref{s_vertex} below. 

\subsubsection{}

A key geometric property of vertex functions are the 
$q$-difference
equations that they satisfy, as functions of both equivariant and 
K\"ahler variables (see \cite{OK}, Section 8, for an introduction).
In particular, the $q$-difference equations in the variables $a_i$ were 
identified in Section 10 of \cite{OK} with the quantum
Knizhnik-Zamolodchikov (qKZ) equations of I.~Frenkel and N.~Reshetikhin 
\cite{FIR}. In \cite{FIR}, these were introduced as the $q$-difference
equations that determine the $q$-deformations of conformal blocks corresponding to ${\Lfgh}$ in \eqref{electric}. 

More precisely, the fundamental solutions of qKZ are
vertex functions counting maps from
$
{\mathbb C}^{\times}$ to  $X$
together with relative insertions at $0 \in \C$
\cite{OK}. 
The relative insertions may be traded for descendent
insertions \cite{Sm_cap, Bethe}. 

In this introductory discussion, we
 will call quasimaps counts with a relative
insertion at $0\in \C$ the {\it vector} vertex functions. This is to distinguish
them from the normal vertex functions counting quasimaps from $\C$ to $X$.

\subsection{Key points of the proof}
Theorem \ref{t:one} follows from connections between
\eqref{electric}, \eqref{magnetic}, and the vertex functions which, in broad strokes, go as
follows.

\subsubsection{Vector vertex functions vs.\
  $U_{\hbar}(\Lfgh)$-conformal blocks}\label{s_steps_qKZ}

On the electric, that is, 
$U_{\hbar}(\Lfgh)$-algebra side, we have a
  characterization of deformed conformal blocks in \eqref{electric} by
  the quantum Knizhnik-Zamolodchikov equations that they satisfy. 
Vector vertex functions provide a different basis of solutions of 
the same qKZ equation. The difference manifests itself through
difference analytic dependence on the equivariant variables $a_i$ and 
the K\"ahler variables $z$.

As correlations functions of chiral operators,
  conformal blocks are analytic in a region of schematic form 
\begin{equation}
|a_5| \gg |a_1| \gg |a_3| \gg \dots,\label{a_region}
\end{equation}
corresponding to time ordering of operators. This the ordering in
which we sew together the chiral vertex operators on ${\cal C}$ to get the
conformal block, and this basic analyticity is unaffected by 
$q$-deformation.

By contrast, vector vertex functions are born as
  convergent power series in the K\"ahler variables $z$, and they have 
 poles in any region of the form \eqref{a_region}.  The variable $z$ 
in which they are holomorphic enter as parameters in the qKZ
equation, namely as an element of the Cartan torus for ${}\Lfgh$.

The dichotomy
 between the two kind of solutions may be axiomatized as in
 \cite{ese}. We have a flat 
$q$-difference connection on a product of two toric varieties (with 
coordinates $a$ and $z$), which is regular in each group of variables 
separately, but is not regular jointly. Regions of the form
\eqref{a_region} and $z\to 0$ are punctured neighborhoods of fixed
points in the two toric varieties. The solutions analytic there 
are called the $a$- and $z$-solutions respectively. With this 
terminology, we can say that 
 \beq\label{Km} 
\textup{ vector vertex functions} = \textup{$z$-solutions to
  $U_{\hbar}(\Lfgh)$ qKZ equation} \,. 
\eeq 

\subsubsection{Elliptic stable envelopes}
\label{s_steps_ese}

Like any two bases of meromorphic solutions to the same difference 
equations, the vector vertex functions and the $U_{\hbar}(\Lfgh)$-conformal blocks 
are connected by a $q$-periodic transition matrix. This $q$-periodic transition matrix
may be called the \emph{pole
  subtraction matrix}, because it literally cancels unwanted poles 
in one set of variables at the expense of introducing poles in another 
set of variables, see \cite{ese} for a detailed discussion.

This pole subtraction matrix was identified geometrically in \cite{ese} 
as the elliptic cohomology version of stable envelopes of the Nakajima variety $X$. 
In equivariant 
cohomology, stable envelopes were introduced in \cite{MO}. They
are the main geometric input in the construction of quantum group
actions suggested there, see Section 9 of \cite{OK} for an overview.
This notion has a natural lift to equivariant 
K-theory, derived categories of coherent sheaves, and, as shown in 
\cite{ese}, also to the equivariant elliptic 
cohomology. 

In parallel to cohomology and K-theory, elliptic stable envelopes
produce an action of a quantum group, namely an elliptic quantum
group. The analysis of \cite{ese} equates the monodromy of qKZ with the braiding for this elliptic 
quantum group. First steps towards such identification were taken
already in \cite{FIR}, with many subsequent developments, as 
discussed in \cite{ese}.

In the enumerative problem, elliptic stable envelopes are inserted via
the the evaluation map at infinity of $\C^\times$, away from the point
$0$ where the relative conditions have been inserted. Later, when we 
discuss integral representation of the solutions, they will appear as 
elliptic functions multiplying the measure of integration as in 
Section \ref{s_contours}. In either interpretation, they map 
vector vertex functions to $U_{\hbar}(\Lfgh)$-conformal blocks.

\subsubsection{Vertex function and $\mathcal{W}$-algebra
  correlators}\label{s_steps_integrals}

On the magnetic, that is, $\mathcal{W}$-algebra
  side we prove in Theorem \ref{t:W} in Section \ref{sec:quasimap-integral} that the vertex functions ${\bf V}$ of
  $X$, counting quasimaps
  \beq\label{scalar}\C \dasharrow X,
  \eeq
 equal the integrals \eqref{magnetic} for a specific
  choices of contours of integration.\footnote{Formal integral
    solutions of differential or $q$-difference equations use 
only the covariance of $\int dx$ with respect to affine linear 
transformations. For $q$-difference equations, $\int dx$ is 
indistinguishable from $\int g(x) dx$, where $g(x)$ is any 
elliptic function. So, by a choice of a contour of integration we 
really mean a choice of both $g(x)$ and $\gamma$ in $\int_\gamma dx \,
g(x) \dots$, where $\gamma$ has to be constrained by the poles of both
the integrand and of $g(x)$, see the discussion in Section
\ref{s_contours}. 
}
The integral formulas for vertex functions of $X$ arise as follows. 

It is well-known (and reviewed in the Appendix) that 
K-theoretic computations on a GIT-quotient by a
reductive group $G$ may be expressed as $G$-invariants in a 
$G$-equivariant computation on the prequotient. The projection onto 
$G$-invariants may be recast, by Weyl integration formula, as an integral over a suitable
cycle in a maximal torus in $G$.
Generalizing this, it is not difficult to show, see Section 
\ref{s_integr}, that
for K-theoretic computations on the moduli spaces of quasimaps to a 
GIT-quotient, there are similar integral formulas. (In fact, such
integral formulas are used routinely in supersymmetric gauge theory
literature. There, they connect two different ways to compute the supersymmetric 
index of the 3d gauge theory on $\C \times S^1$, starting from either
the its Higgs or the Coulomb branch. Coulomb branches of 3d ${\cal
  N}=4$ gauge theories are studied in \cite{Nak13a, Nak13b, Nak13c}.)

To complete the match, it suffices
to recognize in these formulas the integral formulas of \cite{FR} for
the free field correlators of ${\cal W}_{q,t}({\fg})$.

The same dichotomy arises in the discussion of the 
magnetic conformal blocks. Vertex functions are analytic as $z\to 0$,
while the natural requirement for the ${\cal W}_{q,t}({\fg})$-conformal
blocks is to be analytic in regions of the forms
\eqref{a_region}. Very importantly, the very same elliptic stable
envelopes transform the $z$-series into functions with the right 
analyticity in the $a$-variables. The geometry of the correspondence is
tautologically the same, as the insertion of the elliptic stable
envelope happens at infinity, away from the point $0$ which
distinguishes vertex functions from their vector analogs. 
In integral formulas, stable envelopes appear as 
elliptic functions multiplying the measure of integration.

\subsubsection{The match of conformal blocks}

It must be intuitively clear that 
vertex functions are a special case of the vector vertex
functions, namely the one corresponding to no insertion at $0$. Since
the moduli spaces in questions are not really identical, the correct
technical way to see this it is via the degeneration formula as in Section 
\ref{s_degenr}. In particular, it expresses vertex functions as 
vector vertex functions paired with a specific covector, see formula  
\eqref{CGF}.
Applying elliptic stable envelopes to both sides gives the statement
of Theorem \ref{t:one}.

The above identification is a special case of a more 
general important problem in enumerative geometry --- to match 
relative counts with the so-called descendent counts. By definition,
the insertions in the descendent counts are pulled back via the
evaluation map to the quotient stack, while the evaluation map from
the relative moduli spaces goes to the Nakajima variety $X$. While, by the 
degeneration formula, the two kind of counts formally 
contain the same enumerative information, it is very important to
be able to control this equivalence explicitly. A very powerful result
in this direction has been obtained by Smirnov in 
\cite{Sm_cap}, and we use this result here. An alternative, and more 
convenient for our purposes, result has been obtained by two of the 
present authors in \cite{Bethe} after the present work has been completed.

\subsection{First applications and some further directions}

\subsubsection{Difference equations for ${\cal
    W}_{q,t}({\fg})$-conformal blocks} 

The match of the $q$-conformal blocks can be used to
 transfer valuable information in both directions.

On the one hand, the equation \eqref{linear} implies 
that the ${\cal W}_{q,t}({\fg})$-algebra conformal block solve a 
explicit scalar $q$-difference equation gauge equivalent to the qKZ
equations. The existence of such equations is not clear from the 
first principles of deformed ${\cal W}$-algebras as they exist today
and their further investigation is surely a very interesting direction
of research.

Note, in particular, that the monodromy of these 
difference equations is the same as the monodromy of the qKZ 
equations. The stable envelope analysis of \cite{ese} shows 
abstractly that it is given by the $R$-matrices of the corresponding 
elliptic quantum group, as predicted in \cite{FIR}.

\subsubsection{Integral solutions of qKZ} 

In the different direction, any vertex with
  descendants has an integral representation and the match between 
descendant and relative counts gives an integral solutions to qKZ. 
Finding such solutions has been an area of very active research. 
The formulas of \cite{Bethe} give a uniform general answer 
that specializes to the results of 
\cite{ReshqKZ, Matsuo, Matsuo3,
  Varchenko, TV, TV4, TV5} for
$\fg=\mathfrak{sl}_n$.

\subsubsection{General quivers}

The geometric steps outlined above work for a Nakajima
variety associated to a completely general quiver $\scQ$, which may 
have loops at vertices, parallel edges\footnote{
Note that the meaning of parallel edges is different in Nakajima
theory and in the usual notation for Dynkin diagrams. In Nakajima
theory
$$
\textup{Cartan matrix} = 2 - Q- Q^T \,,
$$
where $Q$ is the adjacency matrix of $\scQ$. In particular, 
the Cartan matrix is always symmetric, and so the Lie algebra is
always simply-laced, in this sense.} 
etc. For any such quiver, 
there is a quantum loop group \cite{MO,OK,OS}, and 
the corresponding qKZ equations, which form a 
part of the quantum difference equations. Our argument gives 
an integral solution to these difference equation in a form that 
may be interpreted as a $\mathcal{W}$-algebra
conformal block. 

A representation-theoretic study of these conformal blocks may be an
interesting direction for further research. Note that 
$\mathcal{W}$-algebras associated to quivers appear in the work of Kimura and Pestun \cite{KimPest} in connection with Nekrasov's theory of 
$qq$-character constraints in quiver gauge theories \cite{Nekrasov:2015wsu,Nekrasov:2016qym,Nekrasov:2016ydq}. 
 
\subsection{Non-simply laced groups and folding} 
\label{s_folding} 

%\subsubsection{}
Let $\fg$ be a finite-dimensional simple Lie algebra which is not 
simply-laced, that is, 
$$
^L{\fg}\neq {\fg} \,. 
$$
The Dynkin diagram of ${\fg}$ 
is a quotient of the Dynkin diagram of a simply-laced Lie algebra
${\fg}_0$ by an abelian group $H$ of diagram automorphisms as tabulated in 
\eqref{folding}. 
 This well-known
procedure is called \emph{folding}.

\subsubsection{}
Let the quiver ${\scQ}_0$ be the Dynkin diagram of ${\fg}_0$ and let
$X_0$ be the corresponding Nakajima quiver variety, as before. We
require the dimension vectors to be invariant under $H$. 
Such data is labeled by representations of $^L\fg$, the Langlands dual
Lie algebra of $\fg$, see Sec. \ref{sec:six}. 

\subsubsection{}
We consider $H$-invariant quasimaps to $X_0$, where $H$ acts
simultaneously on the target \emph{and} the source $\bP^1$ of the
quasimaps.  As usual, the $H$-invariant part of the obstruction theory 
defines a perfect obstruction theory for the
 moduli
space $\QM(X_0)^H$ of $H$-invariant quasimaps. Thus we can define the folded 
vertex functions which we denote ${\bf V}^H$. 

These folded vertex function have an integral formula, just like the 
unfolded ones. By inspection, these match the 
integral formulas for the ${\cal W}_{q,t}({\fg})$ deformed 
conformal blocks. 

\subsubsection{}

We conjecture that the steps from Sections \ref{s_steps_qKZ} and 
\ref{s_steps_ese} generalize. 
This requires the development of elliptic stable envelopes (and, as a 
consequence, K-theoretic stable envelopes) in the folded setting.
If true, this would prove our conjectural correspondence in full 
generality.

% be the vertex function of $H$-equivariant quantum K-theory of $X_0$, defined to count $H$-twisted quasi-maps from ${\mathbb P}^1$ to $X_0$. These are quasi-maps which are invariant under a simultaneous rotation of the ${\mathbb P}^1$ and action of a generator $h \in H$ on $X_0$. The $H$ action on $X_0$ is induced from its action on the quiver $\scQ_0$. 

% We conjecture that $H$-equivariant vertex functions ${\bf V}^H$ compute $q$-correlators of ${\cal W}_{q,t}({\fg})$ in \eqref{magnetic}, on one hand, and solve the qKZ equation of $U_{\hbar}(\Lfgh)$, on the other. If so, the $q$-conformal blocks of ${\bf V}_{\fC}^{H}$, the solutions to the qKZ equation holomorphic in a chamber ${\fC}$ equivariant moduli space (the H-invariant subspace of the $A$-equivariant moduli space of $X_0$), are related to ${\bf V}^H$ by a pole-subtraction matrix, an analogue of \eqref{connection},
% $${\bf V}_{\fC}^H = {\fP}_{\fC}^H \,{\bf V}^H.$$
% The pole subtraction matrix  ${\fP}_{\fC}^H$ should be the $H$-equivariant elliptic stable envelope [...] of $X_0$.

\subsection{String theory origin}

The $q$-conformal blocks of $U_{\hbar}(\Lfgh)$ and ${\cal W}_{q,t}(\fg)$ algebras 
are the partition functions of the six-dimensional ``little'' string
theory with $(2,0)$ supersymmetry. 

Little string theory has a conformal limit, in which it becomes a point particle theory, the 6d $(2,0)$ superconformal field theory. This theory is sometimes denoted as theory $\scX({\fg})$; it has been related to quantum Langlands correspondence in \cite{Wittenl1, NW}, following \cite{KW, Kapustin}. 

The conformal limit of little string turns out to coincide with conformal limit of the algebras, when $q$-deformations go away.

\subsubsection{}

For ${\fg}$ a simply laced Lie algebra, one takes the ${\fg}$-type little string theory on a six-manifold 
\beq\label{6man}
M_6 ={\mathcal C} \times {\mathbb C}\times {\mathbb C}.
\eeq 
Here ${\cal C}$ is the Riemann surface on which the chiral algebras
live. The parameters $q$ and $t^{-1}$ are related to equivariant
rotations of the two complex planes in
\eqref{6man};  $\hbar$ is associated with an $R$-symmetry twist, and \eqref{first} is required to preserve supersymmetry.

The vertex operator insertions in \eqref{electric} and
\eqref{magnetic} correspond to introducing codimension four defects of
the little string theory, supported at points of ${\cal C}$ and the
complex plane in \eqref{6man} rotated by $q$. This is illustrated in
Figure \ref{f_branes}. 

\begin{figure}[!hbtp]
  \centering
   \includegraphics[scale=0.48]{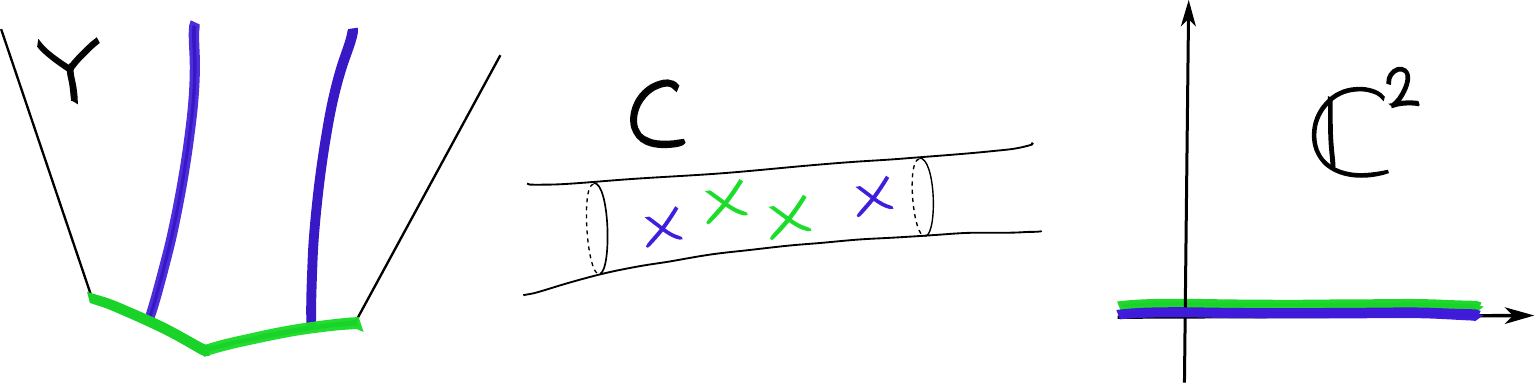}
 \caption{The 10d spacetime of the IIB string
 is the product of an
ADE surface $Y$, the cylinder ${\cal C}$, and $\C^2$. 
The defects we consider are located at $\C\subset \C^2$ that 
is rotated by $q$, at points in ${\cal C}$, and at middle-dimensional
cycles in $Y$. Compact cycles in $Y$, shown in green, give rise to the screening 
operators, while the dual noncompact cycles $H_2(Y, \partial Y, \Z)$
produce vertex operators in fundamental representations.}
  \label{f_branes}
\end{figure}

\subsubsection{}
The ${\fg}$-type little string on \eqref{6man} arises in a limit of IIB string theory on $Y\times M_6$ where $Y$ is an ADE surface of type ${\fg}$. 
The defects of little string theory on $\C$ in $M_6$ lift to D3 branes of IIB string. In $Y$, the D3 branes are supported on 2-cycles whose homology class in $H_2(Y, \partial Y, \Z)$ is identified with the weight in \eqref{weight} using the identification of $H_2(Y, \partial Y, \Z)$ with the weight lattice of $\fg$.

\subsubsection{}

The partition function of the six dimensional little string theory on $M_6$ in \eqref{6man} with the defect D3 branes turns out to localize, due to supersymmetry, to the partition function of the theory on the defects themselves. The theory on defects is \cite{DM} a 3d quiver gauge theory with quiver $\scQ$ whose Higgs branch is the Nakajima variety $X$ in \eqref{NV} (the theory has ${\cal N}=4$ supersymmetry). The 3d gauge theory is supported on 
\beq\label{3man}
{\mathbb C}\times S^1,
\eeq 
where ${\C}$ is identified with the complex plane in $M_6$ supporting the defects. 
%We can neglect all excitations along $Y$ due to the limit we are taking to get from type IIB to the little string theory on $M_6$.
The fact that a defect on $\C$ in $M_6$ supports a three dimensional gauge theory is due to a stringy effect. Given a D3 brane at a point on ${\cal C}$, there are winding modes of strings which begin and end on the brane, and wind around the circle in ${\cal C}\cong{\mathbb C}^{\times}$. These winding modes are mapped to momentum modes on the (T-)dual circle, corresponding to the $S^1$ in \eqref{3man}.

The partition function of the 3d gauge theory are the vertex functions of $X$, computed by quantum K-theory of \cite{OK}. They give either the electric or magnetic blocks, depending on the boundary conditions at infinity in ${\mathbb C}$. 

Many other examples of relations between partition functions of supersymmetric gauge and string theories and ($q-$)conformal blocks  (called BPS/CFT correspondence \cite{BPS}) appeared in physics literature over the years, \cite{IH, NO1, AGT} are a few. One should note that the relation between the 6d $(2,0)$ theory and gauge theories we use here is {\it different} from that in \cite{AGT}.  We use supersymmetry to localize the 6d theory to the theory on its defects -- and observe, following \cite{AH}, that in little string theory, the theory on the defects is a Nakajima quiver gauge theory, for any ${\fg}$, and all possible defects.

\subsubsection{}

To get non-simply laced theories, we start with the little string corresponding to a simply laced Lie algebra $\fg_0$ compactified on $M_6$ in \eqref{6man}, and add an $H$-twist. The $H$-action is represented by a simultaneous rotation around the origin of the ${\mathbb C}$-plane supporting the defects, and permutation of the modes of the theory induced from the action by generator $h$ of $H$ on the Dynkin diagram of ${\fg}_0$.

The theory on the D3 brane defects is a described by starting with an $H$-invariant quiver gauge theory based on ${\scQ}_0$, compactified on $S^1\times {\mathbb C}$. The $H$-twist restricts the fields of the theory on the defects to be those invariant under the simultaneous rotation of ${\mathbb C}$ and $h$-action on the quiver.

\subsubsection{}
All of our discussion so far corresponds to {\it unramified} case of the geometric Langlands correspondence. An important generalization is to include {\it ramifications} at a number of points on ${\cal C}$. 

From the string theory perspective, this is straightforward: ramifications correspond to another class of defects in little string theory on $M_6$ supported at points on ${\cal C}$ and filling ${\mathbb C}\times {\mathbb C}$. These defects were studied in \cite{AH}. They originate from D5 branes supported on 2-cycles in $Y$ in IIB on $Y\times M_6$.
Their effect on the D3 brane gauge theory is to introduce an additional sector, coming from D3-D5 strings, which breaks supersymmetry to
${\cal N}=2$ in 3d. The partition function of this theory on  ${\C}\times S^1$ is a ${\cal W}_{q,t}({\fg})$ algebra conformal
blocks with vertex operators which are $q$-deformations of ${\cal
  W}({\fg})$ algebra primaries. The mathematical implication of this
is a precise statement what the variety $X$ becomes in the ramified
case (the Higgs branch of the 3d ${\cal N}=2$ quiver gauge theory); and a conjecture for ramified quantum $q$-Langlands
correspondence. On the left hand side in \eqref{qLanglands}, one
considers $q$-conformal blocks of $U_{\hbar}(\Lfgh)$ with vertex
operators labeled by the Verma module representations of
$U_{\hbar}(^L{\fg})$ inserted at ramification points; on the right, we get the
${\cal W}_{q,t}({\fg})$ algebra blocks from \cite{AH}.

\subsubsection{}

Little string theory of ${\fg}_0$ on $M_6/H$ is related to both of the 4d Yang-Mills theories with gauge groups based on Lie algebras ${\fg}$ and $^L{\fg}$. $S$-duality relating the 4d gauge theories is a consequence of T-duality in string theory.
One views $M_6/H$ as a $T^2$ fibration over  ${\cal C}\times B$. The two gauge theories arise by T-duality on one or the other cycle of the $T^2$, after one takes the limit in which the characteristic size of the string and the size of the torus go to zero.

In the limit, the partition function of little string theory on the one hand computes conformal blocks of $\Lfgh$ and ${\cal W}_{\beta}(\fg)$ algebras; and on the other it computes the partition functions of the 4d gauge theories based on $^L{\fg}$ and ${\fg}$, respectively. We also derive from this the identification of the parameters of the two 2d CFT's with the parameters $\tau = m(\beta-1)$ and $^L\tau = ^L(k+h^{\vee})$ of the two gauge theories. (See Section \ref{sec:nine}.)

We hope that our work will help provide a unified framework for the
quantum geometric Langlands correspondence relating the 2d conformal
field theory and the 4d gauge theory approaches of
\cite{F,Frenkel1,FFS,Stoyanovsky,JT,gaitsW,gaitsQ,Sch} and
\cite{KW,GuW,Kapustin,FGT}.

\subsection{Plan of the paper}

The paper is organized as follows. In Section 2 we review relevant
aspects of the $U_{\hbar}({\Lfgh})$ and ${\cal W}_{q,t}({\fg})$
algebras.  In Section 3 and 4, we specialize to the case of simply
laced ${\fg}$. In Section 3 we first review relevant aspects of
quantum $K$-theory and of vertex functions. Then, we develop integral
representation of vertex functions and relate them to free field
correlators of the ${\cal W}_{q,t}({\fg})$ algebra in
\eqref{magnetic}.  In Section 4, we review the results of \cite{OK}
relating vertex functions to solutions of qKZ equation corresponding
to $U_{\hbar}({\Lfgh})$, and the role of elliptic stable envelopes of
\cite{ese}.  This completes the proof the quantum Langlands
correspondence for simply laced ${\fg}$. The ${\fg} =A_1$ example is
worked out in detail in Section 5. It should help 
the reader connect  the results of 
present paper to earlier work.  In Section 6 we discuss
various approaches to the quantum geometric Langlands correspondence
and explain why the existence of isomorphisms between conformal blocks
of the affine Kac--Moody algebra $\Lfgh$ and the ${\mc W}$-algebra
${\mc W}_\beta(\fg)$ may be viewed as its manifestation. We relate
this to a conjectural equivalence of two braided tensor categories
associated to $\Lfgh$ and ${\mc W}_\beta(\fg)$. We also discuss the
identification between these conformal blocks using the integral (free
field) representation, and give explicit examples of what our results
in the $q$-deformed case imply in the conformal limit.  In Sections 7
and 8 we explain the relation to physics of three dimensional gauge
theories, and to their string theory embedding. This leads us to the
conjecture for the non-simply laced cases. In Section 9 we explain the
relation of little string theory to 4d gauge theories that were
related to Langlands correspondence in \cite{KW}. The last section is
the appendix which reviews the theory of GIT quotients in K-theory.

\subsection{Acknowledgments}

We benefited greatly from discussions with Tomoyuki Arakawa, Pavel
Etingof, Sergei Gukov, Nathan Haouzi, Victor Kac, Nicolai Reshetikhin,
Shamil Shakirov, Andrey Smirnov, Joerg Teschner and Cumrun Vafa.

MA's research is supported by NSF grant \#1521446, by the Simons
Foundation as a Simons Investigator and by the Berkeley Center for
Theoretical Physics. EF's research was supported by the NSF grant
DMS-1201335. AO thanks the Simons foundation for being financially
supported as a Simons investigator, NSF for supporting enumerative
geometry at Columbia as a part of FRG 1159416, and 
Russian Academic Excellence Project '5-100'. 

\section{$q$-deformed conformal blocks}

\subsection{Electric side} \label{sec:two}

\subsubsection{}

Let ${\cal C}\cong\Ct$ be the Riemann surface from 
Section \ref{s_statement}. For any simple Lie algebra $^L{\fg}$, 
I.~Frenkel and N.~Reshetikhin in \cite{FIR} described the
$\hbar$-deformation of the  $ {\Lfgh}_{^Lk}$ WZW model conformal blocks
on ${\cal C}$ based on the quantum affine algebra
$U_{\hbar}(\Lfgh)$. We briefly recall some of their results here. Throughout this subsection, the normalization of the bilinear form  $(,)_{^L \fg}$ on the Lie algebra is chosen so the longest root has length squared equal to $2$; these are the usual conventions for affine Lie algebras.

\subsubsection{}
The deformed conformal blocks are correlators 
of chiral vertex operators:
\beq\label{WZ}
\Psi(a_1, \ldots, a_{\ell}, \ldots, a_n) = \langle \lambda'| \Phi_{1}(a_1)\ldots  \Phi_{\ell} (a_\ell) \ldots \Phi_{n} (a_n)  |\lambda\rangle.
\eeq
State $|\lambda \rangle$ is a highest weight vector of a Verma
module ${\rho}_{\lambda, ^Lk}$ for $U_{\hbar}(\Lfgh)$ at level $^Lk$. 
These are $\hbar$-deformations of Verma modules of $\Lfgh$. 
%We will denote by $\rho_{\lambda, k}$ the Verma module representation, generated by the highest weight vector $|\lambda\rangle$.
\subsubsection{}
A chiral vertex operator ${\Phi}_{\ell}(a)$ is labeled by a finite dimensional representation ${\rho}_{\ell}$ of $U_{\hbar}( \Lfgh)$ and 
acts as an intertwiner
\beq\label{chiral}
\Phi_{\ell}(a) :  \, \rho_{\lambda_i, ^Lk} \rightarrow \rho_{\lambda_j, k} \otimes \rho_{\ell}(a) \, a^{h(\lambda_i)-h(\lambda_j)},
\eeq
where $\rho_{\lambda_{i,j}, ^Lk}$ are Verma modules of
$U_{\hbar}({\Lfgh})$ and 
\begin{align*}
  \rho_{\ell}(a) =\,\, & \textup{the representation $\rho_\ell $} \\
& \textup{precomposed by the action of $x\in \Ct$} \\
& \textup{by a loop rotation automorphism of $U_{\hbar}({\Lfgh})$}
\end{align*}
is an analog of a evaluation representation for $U_{\hbar}({\Lfgh})$. 

Above,  $h(\lambda)$ is the same factor as for the affine Lie algebra, given by the conformal weight of the state $|\lambda\rangle$.
%$h(\lambda)= C(\lambda)/2(k+h)$. 
The space of intertwiners with this data is 
$$
H_{\lambda_i}^{\lambda_j , \rho_\ell}= {\rm Hom}_{U_{\hbar}(^L{\fg})}(\rho_{\lambda_i}, \rho_{\lambda_j} \otimes {\rho_{\ell}}).
$$ 
where $U_{\hbar}(^L\fg)$ is the finite quantum group, and $\rho$'s are
the corresponding modules --- a direct generalization of the WZW
statement.

\subsubsection{}
The deformed conformal block $\Psi(a)$ takes values in 
a weight subspace  
\begin{equation}
\Psi(a_1, \ldots, a_{\ell}, \ldots ,a_n)  \in (\rho_1\otimes \ldots
\rho_{\ell} \otimes \ldots \otimes
\rho_n)_{\lambda'-\lambda}, \label{weightPsi} 
\end{equation}
the weight $\lambda'-\lambda$. 

\subsubsection{}
As a deformation of the Knizhnik-Zamolodchikov equations \cite{KZ}, 
I.~Frenkel and N.~Reshetikhin obtain a holonomic system of 
$q$-difference equations for the conformal block \eqref{WZ}. 
It is called the  quantum Knizhnik-Zamolodchikov equations
and is the most powerful description of \eqref{WZ}. It has the 
form: 
\beq\label{qKZ}
\begin{aligned}
\Psi(a_1, \ldots q a_\ell, \ldots a_n) =&{\cal R}_{\ell \ell-1}(q a_{\ell}/a_{\ell-1}) \cdots {\cal R}_{\ell 1} (q a_{\ell}/a_{\ell-1})   (\hbar^{\rho})_\ell\\
\times &{\cal R}_{\ell n}(a_{\ell}/a_n) \ldots {\cal R}_{\ell \ell+1}(a_{\ell}/a_{\ell+1})\Psi(a_1, \ldots a_\ell, \ldots a_n) 
\end{aligned}
\eeq
where
$$
q = \hbar^{-^L(k + h^{\vee})},
$$
and 
$${\cal R}_{i j}(x) \subset  {\rm End}({\rho}_i \otimes {\rho}_j)
$$ 
is the $U_{\hbar}(\Lfgh)$ $R$-matrix corresponding to a pair $\rho_i, \rho_j$ finite dimensional $U_{\hbar}(\Lfgh)$ modules.  
Furthermore, $(\hbar^{\mu})_{\ell}$ acts on $\ell$'th component of the tensor, corresponding to representation $\rho_{\ell}$. Its action on a vector $v_w$ of weight $w$ is
$$
\hbar^{\mu}(v_w) =\hbar^{(\mu, w)} v_w
$$
The vector $\rho$ is the Weyl vector, equal to half the sum of positive roots of $^L\fg$.
Once we fix 
a specific ordering of vertex operators  in 
\eqref{WZ} or, equivalently, a region of the form 
$$
|a_5| > |a_2| > |a_7| >  \dots, 
$$
 the qKZ equation determines the $q$-conformal blocks completely. 
The solutions in each region are  labeled by elements of 
$$
\bigoplus_{\lambda_1,\dots,\lambda_{n-1}} 
H_{\lambda_1}^{\lambda_0 , \rho_1} \otimes\ldots \otimes  H_{\lambda_\ell}^{\lambda_{\ell-1} , \rho_\ell} \otimes \ldots  \otimes H_{\lambda_\infty}^{\lambda_{n-1} , \rho_n}
$$
where $\lambda_1,\dots,\lambda_{n-1}$ are the highest weights 
of Verma modules 
in intermediate channels. Note that the dimension of this space 
equals the dimension of \eqref{weightPsi}. 

\subsubsection{}

The notation here differs from that of I.~Frenkel and N.~Reshetikhin in \cite{FIR} by:
\beq\label{identify}
(q)_{\textup{ here}} =(p)_{\textup{FR}}, \qquad  (\hbar)_{\textup{ here}} =(q^{2})_{\textup{FR}},
\eeq
and $(p = q^{-2^L(k+h)})_{\textup{FR}}$.

\subsubsection{}

In the conformal limit, when 
\beq\label{limitwzw}
  \hbar, \;\; q \qquad  \longrightarrow \qquad 1
\eeq
with $^L(k+h)$ fixed,
the qKZ equation reduces to KZ equation 
\beq\label{KZ}
^L(k+h)\, a_{\ell}{\partial \over \partial a_{\ell}} \,\Psi = 
\left(\sum_{j \neq \ell} {r_{\ell j}}(a_\ell/a_j)+ r_{\ell 0} +
  r_{\ell \infty}
\right)\, \Psi.  
\eeq
derived in \cite{KZ}. 
The matrix 
$$
r_{ij}(a_i/a_j) = \frac{r_{ij} a_i+r_{ji} a_j}{a_i-a_j}
$$
with 
$$
r= {1\over 2} \sum_a h_a\otimes h_a + \sum_{\alpha>0}
e_{\alpha}\otimes e_{-\alpha}
$$
in the standard Lie theory notation, 
is the classical $R$-matrix of WZW on
$\mathcal{C}=\Ct$. This is referred to as the trigonometric $R$-matrix, as opposed to the rational one, corresponding to the case ${\cal C} = \C$.

\subsection{Magnetic side}\label{sec:four}

% \subsubsection{}

% There is no general definition of a deformed chiral algebra available.
% For quantum affine algebras, one can define their $q$-conformal blocks as the solutions of the
% quantum Kniznik-Zamolodchikov equations above. For the deformed ${\cal W}$-algebras
% ${\cal W}_{q,t}({\fg})$ such equations are not known. (For $\fg=A_1$, ${\cal W}_\beta(\fg)$ is the Virasoro algebra,
% and the conformal blocks satisfy the BPZ equations; however, their deformations are
% not known.) We will give a definition of the ${\cal W}_{q,t}(\fg)$ $q$-conformal blocks in terms of certain "vertex" functions of a class of Nakajima quiver varieties.

\subsubsection{}

Let ${\cal C}$ and $\fg$ be a before.  
The deformed ${\cal W}_{q,t}(\fg)$ algebra, and certain classes of
vertex and screening operators, were constructed by E. Frenkel and R. Reshetikhin in
\cite{FR} in terms of free fields, as a $q$-deformation of ${\cal W}_{\beta}(\fg)$ algebra, where $t=q^{\beta}$. (See also
\cite{SKAO,FF:qw,AKOS} for $\fg=sl_n$.)
The free field realization implies that the $q$-conformal blocks on ${\cal C}$ of the form 
\beq\label{mcor}
\langle \mu'| V_1^{\vee}(a_1) \ldots V_{n}^{\vee}(a_n) \prod_a(Q_a^\vee)^{d_a} | \mu\rangle
\eeq
have a direct description in terms of certain contour integrals.

\subsubsection{}
Let ${\fg}$ be a simple Lie algebra, and $C_{ab}$ its Cartan matrix, defined as
\beq\label{Cm}
C_{ab} = 2 (e_a, e_b)/(e_a, e_a) = (e_a^{\vee}, e_b),
\eeq
in terms of simple positive roots $e_a$, the coroots $e_{a}^{\vee}$ and the invariant inner product $(,)$ on the Lie algebra.
Let $m$ be the lacing number, the maximum number of arrows connecting a pair of nodes in the Dynkin diagram. Unless $m=1$ and the theory is simply laced, the Cartan matrix is not symmetric. Instead, the symmetric matrix is
\beq\label{Bm}
B_{ab} = m_a C_{ab} = m(e_a, e_b),
\eeq
where we defined 
$$m_a=m {(e_a, e_a)/2}.
$$  
We choose the normalization of the inner product $(,)_{\fg}$ on so that  $m_a=m$ for long roots and  $m_a=1$ for short roots. 
\subsubsection{}
To define the deformed ${\cal W}_{q,t}(\fg)$, one starts \cite{FR} with the $q$-deformed Heisenberg algebra ${\cal H}_{q,t}({\bf g})$ in terms of ``root'' type generators $e_a[k]$, for $k\in {\mathbb Z}$ where $a$ labels the simple positive root of $\fg$. The generators satisfy commutation relations
$$
[e_a[k], e_b[\ell]] = {1\over k} (q^{k\over 2} - q^{-{k\over 2}})(t^{{k\over 2} }-t^{-{k\over 2} })B_{ab}(q^{k} , t^{k} ) \delta_{k, -\ell}.
$$
Here, $B_{ab}(q,t)$ is a $q$-deformation\footnote{The quantum number $n$ is defined as $[n]_q = {q^{n/2} - q^{-n/2} \over q^{1/2} - q^{-1/2}}$. } of \eqref{Bm},
$$
B_{ab}(q,t) = [m_a]_q\;  C_{ab}(q,t),
$$
where $C_{ab}(q,t)$ is the $q$-deformed Cartan matrix 
$$
C_{ab}(q,t)= q^{m_a\over 2} t^{-{1 \over 2}} + q^{-{m_a\over 2}}t^{1\over 2}-[I_{ab}]_q,
$$
and $I_{ab}$ is the classical incidence number, $I_{ab} = 2\delta_{ab} - C_{ab}$.

\subsubsection{}
We get a Fock representation of the Heisenberg algebra, denoted by $\pi_{\mu}$, by starting with the state $|\mu\rangle$, such that
$$
e_a[k]|\mu\rangle =0, k>0, \qquad e_a[0] |\mu\rangle = (\mu, e_a) |\mu\rangle, 
$$
and acting by the algebra generators.

\subsubsection{}
One defines the magnetic and the electric screening currents following
\cite{FR}:
\beq\label{magneticS}
S^{\vee}_a(x) =[\ldots]x^{-e_a[0]/m_a} : \exp\Bigl(\sum_{k\neq 0}{ e_a[k] \over q^{k m_a\over 2} - q^{-{k m_a \over 2}}} x^k\Bigr):
\eeq
and
\beq\label{electricS}
S_a(x) =[\ldots]x^{e_a[0]/\beta} : \exp\Bigl(-\sum_{k\neq 0}{ e_a[k]
  \over t^{k\over 2} - t^{-{k \over 2}}} x^k\Bigr): \,.
\eeq
The terms denoted by $[\ldots]$ above are
operators responsible, in part, for
shifts of momenta $\mu$ in $\pi_{\mu}$ in \eqref{screen}.  
If  ${\fg}$ is simply laced, 
there is a symmetry exchanging $S^{\vee}_a$ and $S_a$ and swapping $q$
and $t$.

The algebra ${\cal W}_{q,t}({\fg})$ is
defined as the associative algebra generated by the (Fourier
coefficients of) operators $T(x)$ which commute with the screening
charges $S^{\vee}_a(x)$ and $S_a(x)$ up to a total difference, e.g.
\begin{equation}
\left[T(x), S^{\vee}_a(x')\right] = {\cal D}_{x',q} f(x, x'), \qquad {\cal
   D}_{x,q}f(x) = { f(x)-f(qx)\over x(1-q)}  \label{comm_q}  \,. 
\end{equation}

\subsubsection{}\label{s_contours}

For the corresponding screening charges
\beq\begin{aligned}\label{screen}
Q_a^{\vee} &= \int  S_a^{\vee}(x)\;: \pi_0\rightarrow \pi_{- e_a \beta /m_a } \cr
 Q_a&= \int  S_a(x)\; :  \pi_0\rightarrow \pi_{e_a } 
\end{aligned}
\eeq
equation \eqref{comm_q} implies 
$$
\left[T(x), \int S^{\vee}_a(x') \right] = 0  \,. 
$$
Here 
$
f(x) \mapsto \int  f(x)
$
is any linear functional such that 
$$
\int \frac{f(x)}{x} = \int \frac{f(qx)}{x}  \,. 
$$
For example, we can take 
$$
\int f(x) = \int_\gamma f(x) \, dx 
$$
for any path $\gamma$ such that 
$$
q\cdot \gamma - \gamma = 0 \in H_1(\Ct \setminus 
\textup{singularities of the integrand})  \,. 
$$
More flexibly, we can take, 
\begin{equation}
\int f(x) = \int_\gamma f(x) \, g(x) \, dx \,, \quad g(qx) = g(x)
\,, \label{intgx}
\end{equation}
with the same assumption on 
the integration cycle $\gamma$.

As we will see below, the insertion of the right
  elliptic function $g(x)$ under the integral as in \eqref{intgx} 
corresponds geometrically to the insertion of an elliptic stable 
envelope in quasimap enumeration. These elliptic stable envelope 
transform the $z$-solutions that appear naturally in the enumerative 
problem into $a$-solutions that correspond to conformal blocks.

For $S_a(x)$, the analysis is the same, except $q$-shifts are replaced
by $t$-shifts. 

\subsubsection{}

The weight type generators $w_a[k]$ are associated with fundamental weights of $\fg$. They are defined by
$$
e_a[k] = \sum_{b} C_{ab}(q^k, t^k) w_b[k]
$$
which satisfy
$$
[e_a[k], w_b[\ell]] = {1\over k} (q^{k m_a\over 2} - q^{-{km_a \over 2}})(t^{{k\over 2} }-t^{-{k\over 2} })\delta_{ab} \delta_{k, -\ell}
$$

\subsubsection{}
Similarly, there are magnetic $V_a^{\vee}$ and electric $V_a$
degenerate vertex operators, associated with fundamental coweights
and weights (as defined in Section 9 of \cite{FR}, with some slight relabeling): 
\beq\label{weightv}
V^\vee_a(x) = x^{w_a[0]/m_a}:  \exp\Bigl( - \;\sum_{k\neq 0} {w_a[k]\over q^{k m_a\over 2} - q^{-{k m_a\over 2}}}  \, x^{k}\Bigr) :
\eeq
and 
\beq\label{weightvb}
V_a(x) = x^{-w_a[0]/\beta}:  \exp\Bigl( \;\sum_{k\neq 0} {w_a[k]\over t^{k\over 2} - t^{-{k \over 2}}}  \, x^{k}\Bigr) :
\eeq
for the electric vertex operators.  The insertion at infinity is determined by charge conservation.

\subsubsection{}

We denote 
\begin{equation}
\varphi_q(s) = \prod_{n=0}^\infty (1-q^n s) \,, \label{defvarphi}
\end{equation}
In the taxonomy of special function, the function \eqref{defvarphi} is
best described as the reciprocal of the $q$-analog of the
$\Gamma$-function. It is also known under many other names. 

The infinite product \eqref{defvarphi} is a half of the odd genus $1$ 
theta-function 
\begin{equation}
\theta_q(s) =  \varphi_q(s) \, \varphi_q(q/s)  \,. 
\label{vth}
\end{equation}
which is vanishing at $s=1$ and normalized to 
be a single-valued function of $s$.  

For simply laced $\fg$ we will often drop the subscripts and then $\varphi(s) = \varphi_q(s)$, etc.

\subsubsection{}\label{s:explicit}

Collecting the definitions above, the ${\mc W}_{q,t}
(\fg)$ correlator in \eqref{mcor} 
\beq\label{onem}
\langle \mu'| V_1^{\vee}(a_1) \ldots V_{n}^{\vee}(a_n) \prod_a(Q_a^\vee)^{d_a} | \mu\rangle
\eeq
is the integral 
\beq\label{Wai}
\int \,d_{\rm Haar} x \,x^{\mu}\; \Phi(x,a).
\eeq
where we defined
$$
 x^{\mu} = \prod_{a, \alpha}  x_{a, \alpha}^{(\mu, e_a)}.
$$
and $d_{\rm Haar} x =\prod_{a, \alpha} dx_{a,\alpha}/x_{a, \alpha}$.
The integrand $\Phi(x,a)$ is a product of three terms, which come from normal (re)ordering of the operators in \eqref{onem}:
The first is comes from the screening currents associated to node $a$:
\beq\label{sctdasl1}
\prod_{\alpha < \alpha'} \langle S^{\vee}_a(x_{\alpha, a})S^{\vee}_a(x_{\alpha',a}) \rangle = \prod_{\alpha \neq \alpha'}{ \varphi_{q_a}(  x_{\alpha ,a}/x_{\alpha',a})\over \varphi_{q_a}(   t\,x_{\alpha, a}/x_{\alpha',a})}  \prod_{\alpha < \alpha'}{ \theta_{q_a}( t x_{\alpha ,a}/x_{\alpha',a})\over \theta_{q_a}(    \,x_{\alpha, a}/x_{\alpha',a})} 
\eeq
Here and below, $\langle\rangle$ stand for expectation values computed in state $|0\rangle$. The second factor comes from  screening charges associated to pair of nodes $a,b$ connected in by a link the dynkin diagram of ${\fg}$, and equals:
\beq\label{sctdasl2}
\prod_{\alpha , \beta} \langle S^{\vee}_a(x_{\alpha, a}) S^{\vee}_b(x_{\beta, b}) \rangle=\prod_{\alpha, \beta} { \varphi_{q_{ab}}( t  v_{ab} \,x_{\alpha,a}/x_{\beta,b})\over \varphi_{q_{ab}}(v_{ab} \,{x_{\alpha ,a} /x_{\beta,b}})}.
\eeq 
 Above $v_a$, $v_{ab}$ are defined as follows:  $v_a = \sqrt{q_a/t}$ and $v_{ab} =\sqrt{q_{ab}/ t}$, where $q_a=q^{m_a}$ and
$q_{ab} = q^{{\rm min}(m_a, m_b)}$. (If either of the nodes $a,b$ corresponds to a short root then $q_{ab} = q$, and, it both of the nodes are long then $q_{ab}=q^m$.) 
The third, and finally factor comes from normal ordering the vertex operator $V_a^{\vee}(a_i)$'s with the screening currents $ S^{\vee}_a$ coming from the same node:
\beq\label{sctdasl3}
\prod_{\alpha , a} \langle S^{\vee}_a(x_{\alpha, a})  V^{\vee}_a(a_i)  \rangle=\prod_{a, \alpha}  { \varphi_{q_a}( t v_a \,{a_i/x_{\alpha,a}})\over \varphi_{q_a}( v_a \,{a_i} /x_{\alpha,a})}.
\eeq
The $\mu$ dependent factor in \eqref{Wai} accounts for the fact that the the incoming stat in  \eqref{onem} is $|\mu\rangle$, and not the trivial vacuum $|0\rangle.$

\subsubsection{}

In writing \eqref{sctdasl1} we assume the argument $x$ of
$\varphi_q(x)$ is less than one, $|x|<1$. Otherwise $\varphi_q(x)$
gets replaced by $1/\varphi_q(q/x) = \varphi_q(x)/ \theta_q(x)$.
This is a feature of deformed chiral algebras, as defined in
\cite{FR:dca}.

\subsection{Conformal limit}

The conformal limit, in which one recovers the ordinary (conformal)
${\mc W}$- algebra ${\cal W}_{\beta}(\fg)$,
$$
{\cal W}_{q,t}({\fg }) \qquad \rightarrow \qquad {\cal W}_{\beta}(\fg)
$$
corresponds to taking 
\beq\label{limitw}
q, \; t=q^{\beta}
 \qquad \rightarrow \qquad 1.
\eeq
keeping ${\beta}$ fixed, as in \cite{FR}. The conformal ${\mc
  W}$-algebra ${\cal W}_{\beta}(\fg)$ with $\beta = m (k+h^{\vee})$ is
obtained from $\widehat{\fg}$ of level $k$ via the quantum
Drinfeld--Sokolov reduction (see \cite{FF,FB} and Section 6 below for
details).

\subsubsection{} The limit \eqref{limitw} requires rescaling of the generators of the algebra. 
The generators of the algebra that stay finite in the limit are $e'_a[k] = e_a[k]/\log(q)$ and $w'_a[k]= w_a[k]/\log(q)$.
In the limit, we get
%The moduli of the hyperkahler metric on $X$ are encoded in periods of three self-dual two-forms ${\omega}_{I, J, K}$  along the 2-cycles $S_a$ generating $H_2(X, {\mathbb Z})$.  When we compactify IIB string theory on $X$, we get two more moduli, coming from periods of IIB two-forms, $B_{NS}$ an $B_{RR}$ on $X$.
\beq\label{clima}
 \langle S^{\vee}_a(x)\, S^{\vee}_b(x') \rangle = (x - x')^{{\beta\over m} (e_a^{\vee} , e_b^{\vee})}, \qquad   \langle S^{\vee}_a(x)V^{\vee}_b(x') \rangle = (x - x')^{-{\beta \over m} (e_a^{\vee} , w_b^{\vee})}
 \eeq
and
\beq\label{climb}
\langle S_a(x)\, S_b(x') \rangle = (x - x')^{{m \over \beta} (e_a , e_b)}
, \qquad \langle S_a(x)V_b(x') \rangle = (x - x')^{-{m\over \beta} (e_a , w_b)} 
\eeq
 where  $e_a^\vee$, $w_a^{\vee}$ are the coroots, and the fundamental
 coweights, respectively\footnote{The fundamental coweights are
   defined by $(e_a, w_b^{\vee})=\delta_{a,b}$.}. The formulas reflect
 the fact that for a pair of Langlands dual Lie algebras, $\fg$ and
 $^L\fg$, there is an isomorphism of the corresponding ${\cal
   W}$-algebras \cite{FF}:
$${\cal W}_{\beta}(\fg)= {\cal W}_{^L\beta}(^L\fg), \qquad \beta\, ^L \beta =m
$$
(see Section 6 for details). One recalls that, under exchanging
${\fg}$ and $^L \fg$, roots and coroots get exchanged, as well as the
weights and coweights, and the inner product gets rescaled,
$(,)_{\fg}= m(,)_{^L \fg}$.

\section{Integral representation of vertex functions}\label{sec:quasimap-integral}

\subsection{Quasimaps to Nakajima varieties}

\subsubsection{} 

Let $X$ be a Nakajima variety as in Section \ref{s_Nakajima}
\begin{align}
  X&=T^*{{\rm Rep }\,\scQ} \rd G_{\scQ} \notag\\
   & = \mu^{-1}(0) \rdd G_{\scQ} \,, \label{Xquo} 
\end{align}
where $\mu$ is the complex moment map for the action of $G_{\scQ}$. 
From \eqref{Xquo}, it is a GIT quotient of an affine algebraic variety 
by an action of a reductive group. 

For such quotients, Ciocan-Fontanine, Kim, and Maulik define in
\cite{CKM} a notion
of quasimap 
\begin{equation}
f: C \dasharrow X\label{fquasi} \,. 
\end{equation}
These are maps to $X$ with certain singularities. Informally, a 
\emph{stable} quasimap $f$ is
allowed to take a GIT-unstable value at finitely many points of
$C$. In what follows, all quasimaps are assumed stable. 

Analogous notions are known in both supersymmetric gauge theory 
and mathematics literature. The precise definition of \cite{CKM} is 
best suited for our goals here. An introductory discussion of 
quasimaps of \cite{CKM} may be found in \cite{OK}.

\subsubsection{}

The vector spaces $V_a$ in the quiver description of $X$ descend to
vector bundles on $X$. These are called \emph{tautological}. 

The data of a quasimap includes vector bundles $\cV_a$ on 
the domain $C$; they coincide with the pullbacks $f^* V_a$ of the 
tautological bundles for regular maps $f: C \to X$ and are part of 
the definition in general. Similarly, we have the trivial bundles
$\cW_a$ on $C$ corresponding to the trivial bundles $W_a$ on $X$. 
We denote by 
\begin{equation}
\label{cM}
\cM = \oplus_{a \rightarrow b} \,\cHom(\cV_a, \cV_b) \oplus_{a}
\cHom(\cV_a, \cW_a)
\end{equation}
the bundle corresponding to \eqref{repQ}.

By definition, a quasimap is collection of bundles $\{\cV_a,\cW_a\}$
together with a stable section 
$$
f \in H^0(C,\cM \oplus \hbar^{-1}\otimes \cM^*) 
$$
satisfying the moment map equation 
$$
\mu(f) = 0 \in H^0(C, \oplus_a \, \cEnd(\cV_a))  \,.
$$
Stability of $f$
means it evaluates to a GIT-stable point
at all but finitely many points of $C$. This data is 
considered up to isomorphisms fixing $C$ pointwise. In other words, 
we consider quasimaps from parametrized domains.

\subsubsection{}

Let $\QM(X)$  be the moduli
space of quasimaps 
from  $C\cong\bP^1$ to $X$. On the domain 
$C$, we fix marked points 
$$
p_1 = 0 \,, \quad p_2 = \infty\,,
$$ 
and denote  
$$
\Ct_q = \Aut(C,p_1,p_2) \,. 
$$
Here the subscript is to distinguish this torus from other tori
present; an element of $\Ct_q$ will be denoted $q$. We take 
$T_{p_1}C$ as the defining (i.e., weight one) 
representation of $\Ct_q$. 

\subsubsection{}

The degree of a quasimap is defined as follows 
$$
\deg f = \left(\deg \cV_1,\deg \cV_2,\dots\right) \in
H_2(X,\Z)_\textup{effective} \,,
$$
see the discussion in Section 7.2 of \cite{OK}. 
This is a locally constant function on $\QM(X)$.

\subsubsection{}

By definition, 
vertex functions for $X$ are computed using $\Ct_q$-equivariant 
K-theoretic localization on the open set 
$$
 \QM_\textup{nonsing $\infty$} \subset \QM
$$
formed by quasimaps nonsingular at $p_2=\infty$. It is therefore 
important to understand the structure of the fixed locus 
$\left(\QM_\textup{nonsing $\infty$} \right)^{\Ct_q}$. It is 
discussed, in particular, in Section 7.2 of \cite{OK}.

\subsubsection{}

The analysis of the fixed loci may be summarized as follows. We define
$$
\bbV_a = \bigoplus_{k\in \Z}  \bbV_a[k] =
H^0(\cV_a\big|_{C\setminus\{p_2\}}) \,,
$$
where $\bbV_a[k]$ is the subspace of weight $k$ with respect to
$\Ct_q$. By invariance, all quiver maps preserve this weight 
decomposition. We define the framing spaces $\bbW[k]$  in 
the same way and obtain
\begin{equation}
\bbW_a[k] = 
\begin{cases}
W_a\,,&  k\le 0 \,,\\
0\,,&  k> 0 \,,\\
\end{cases} 
\label{bbW}
\end{equation}
because the bundles $\cW_a$ are trivial. 

Multiplication by the coordinate induces an 
embedding 
\begin{equation}
\bbV_a[k]  \hookrightarrow \bbV_a[k-1] 
\hookrightarrow \dots \hookrightarrow \bbV_a[-\infty] = V_a
\,,\label{flagV}
\end{equation}
compatible with quiver maps, where $V_a$ is the quiver data for 
the point $f(\infty)\in X$. A $\Ct_q$-fixed stable quasimap $f$ takes
a constant stable value on $C\setminus \{0,\infty\}$ and, since 
$f$ is additionally assumed nonsingular 
at infinity, $f(\infty)$ is that generic value of $f$.

We conclude 
\begin{equation}
\left(\QM_\textup{nonsing $\infty$}\right)^{\Ct_q} = 
\left\{
\begin{matrix}
\textup{a stable quiver representation} \\
\textup{$+$ a flag of subrepresentations} \\
\textup{satisfying \eqref{bbW}}
\end{matrix}
\right\} \Bigg/ \prod GL(V_a)\,.  \label{flag_subr} 
\end{equation}

\subsection{Vertex functions}\label{s_vertex} 

\subsubsection{}

Vertex 
functions are defined as 
generating functions of equivariant counts of 
quasimaps of all degrees. Concretely, consider the evaluation map 
$$
\ev:  \QM_\textup{nonsing $\infty$}(X) \to X
$$
that records the value $f(\infty)$ of a quasimap $f$. We introduce a weighting 
by $z^{\deg f}$, where $z$ are the variables in the generating 
function, and define
\begin{equation}
\Vertex  = \ev_{*} \left(
\tO_\vir \,  z^{\deg f} \right) \in
K_{\bT\times\Ct_q}(X)_\textup{localized} \otimes \Q[[z]] \,, 
\label{defVertex}
\end{equation}
where the symmetrized virtual structure sheaf $\tO_\vir$ will be 
discussed below and $\Q[[z]]$ denotes formal power series in 
$z$ with exponents supported in the effective cone. 

\subsubsection{}

The push-forward in \eqref{defVertex} 
 is defined using $\Ct_q$-equivariant 
localization. (It is clear from the above description of the
$\Ct_q$-fixed quasimaps that the evaluation map is proper on 
these fixed loci.) Because of this, vertex functions are 
series in $z$ with coefficients in localized equivariant cohomology. 
Their denominators are the source of their richness and complexity; 
analogous functions without denominators (called the \emph{cap} in
the professional lingo, see below) lose this complexity. 

\subsubsection{}

The symmetrized virtual structure sheaf is defined by \cite{OK} 
\begin{equation}
\tO_\vir = \cO_\vir \otimes \left( \cK_\vir \, 
\frac{\det \cT^{1/2} \big|_{\infty} }{ \det \cT^{1/2} \big|_{0}  }
\right)^{1/2} \,, \label{def_tO}
\end{equation}
where $\cO_\vir$ is the virtual structure sheaf constructed in the 
standard way from the perfect obstruction theory of quasimaps, 
see \cite{BF,FG,CF3}, $\cK_\vir$ is the virtual canonical bundle, that is, the 
determinant of the virtual cotangent bundle, 
and the remaining 
term involves a choice of polarization of $X$ and is mainly 
needed 
to avoid square roots of $q$.

\subsubsection{}

The virtual tangent bundle to $\QM(X)$ may be described as follows 
\begin{align}
  T_\vir \QM(X) &= \Def - \Obs \notag \\
                      &= \Hd(C,\cT)\,, \label{TvirQM} 
\end{align}
where 
$$
\cT = \cM \oplus \hbar^{-1}\otimes \cM^* - (1+\hbar^{-1}) 
\oplus_a \, \cEnd(\cV_a)  
$$
is the virtual bundle on $C$ corresponding to the tangent 
bundle $TX$ of our Nakajima variety. 

\subsubsection{} 

By definition, a polarization 
$T^{1/2} \in K_\bT(X)$ is a choice of a half of the tangent 
bundle, that is, a choice of the solution of the equation 
$$
T^{1/2}X + \hbar^{-1} \otimes \left(T^{1/2}X\right)^\vee = TX 
$$
in $K_\bT(X)$, where vee denotes
dual. 
Natural polarization of Nakajima varieties correspond to 
choosing one out of every pair of quiver arrows. 
A polarization 
$T^{1/2}$ induces a virtual vector bundle $\cT^{1/2}$ on 
the domain $C$ of the quasimap such that 
\begin{equation}
\cT = \cT^{1/2} + \hbar^{-1} \otimes \left(\cT^{1/2}\right)^\vee
\,. \label{cTcT} 
\end{equation}
The fibers of the line bundle $\det \cT^{1/2}$
enter \eqref{def_tO}. The square root in \eqref{def_tO} exists, perhaps after 
introducing $\hbar^{1/2}$. For quasimaps, it can be given 
explicitly in terms of a polarization. See Section 6.1 of \cite{OK}
and 
also \cite{NO2} for general 
results in this direction. Vertex functions defined using different choice of polarization 
differ by a $q$-shift of the variables $z$ only. 

\subsection{Localization contributions} 

\subsubsection{}

Recall that the push-forward in \eqref{defVertex} is defined using 
$\Ct_q$-equivariant localization. General shape of virtual
localization formulas, see \cite{GP,CF3}, is the following. 
Restricted to the fixed locus, the obstruction theory splits 
\begin{equation}
\label{fixed_moving}
T_\vir  \big|_{\QM(X)^{\Ct_q}}  = 
T_{\vir,\textup{fixed}} \oplus T_{\vir,\textup{moving}} 
\end{equation}
into trivial and nontrivial $\Ct_q$-eigenspaces. The fixed part of the
obstruction theory produces a perfect obstruction theory for the fixed 
locus, from which its own virtual structure sheaf and the symmetrized 
virtual structure sheaf are derived. The moving 
parts of the obstruction theory enter the localization formula as
K-theoretic analog of the Euler class or, more precisely,
$\aroof$-genus \eqref{defaroof} for the virtual localization of $\tO_\vir$. 

\subsubsection{}\label{s_loc_plan}

Our plan for the analysis of the localization contributions is the
following: 
\begin{itemize}
\item[---]  first, we show \eqref{flag_subr} is 
a GIT quotient. This makes the techniques 
reviewed in the Appendix applicable 
to these fixed loci. 
\item[---] We identify the fixed part of the quasimap obstruction 
theory with the natural obstruction theory of \eqref{flag_subr}. 
\item[---] We include the moving contributions to derive an 
integral representation for the vertex functions. 
\end{itemize}

\subsubsection{}

In order to show \eqref{flag_subr} is 
a GIT quotient, one may analyze GIT-stability
 on the ambient space 
\begin{equation}
\left(\QM_\textup{nonsing $\infty$}\right)^{\Ct_q} \subset 
\left\{
\begin{matrix}
\textup{a stable quiver representation} \\
\textup{$+$ flags of subspaces in $V_a$} 
\end{matrix}
\right\} \Bigg/ \prod GL(V_a)\,,  \label{flag_subs} 
\end{equation}
where the flags of subspaces need not form a flag of 
subrepresentations. The required ample line bundles 
will be obtained by restriction from the ambient space in
\eqref{flag_subs}. 

The ambient space in \eqref{flag_subs} is not a 
Nakajima quiver variety, but it may be presented as
 a quiver variety in which 
the data for $X$ is extended 
by chains 
\begin{equation}
V_a \leftarrow V_a' \leftarrow V_a'' \leftarrow \dots 
\label{flagVp}
\end{equation}
attached to every $V_a$.  The spaces in \eqref{flagVp} correspond to 
subspaces in \eqref{flagV},  excluding repetitions. This 
extended quiver data 
is taken modulo $G \times G'$ where 
$$
G = \prod GL(V_a) \,, \quad G' = \prod GL(V'_a) \times GL(V_a'')
\times 
\cdots \,. 
$$
As a GIT stability parameter, 
we need to specify a character $\bar\chi$ of 
$G\times G'$. We take 
$$
\bar\chi= \chi^m \,\chi'\,, \quad m\gg 0
$$
where 
$\chi$ is the stability parameter for $X$ and $\chi'$ is the character
of $G'$  that forces
the maps \eqref{flagVp} to be injective. 

\begin{Lemma}
The quotient in \eqref{flag_subs} is the GIT quotient of the 
extended quiver data with the stability parameter $\bar\chi$.  
\end{Lemma}

\begin{proof}
We use King's reformulation of the GIT stability of quiver
representations in terms of slope stability, see for example 
the exposition in Section 2.3 of \cite{GinzNak}. Namely, a representation $R$
is semistable if and only if 
$$
\textup{slope}_{\bar\chi}(S) \le \textup{slope}_{\bar\chi}(R)
$$
for every nonzero subrepresentation $S\subset R$, where
\begin{equation}
\textup{slope}_{\bar\chi}(R) = 
\frac{ \bar\chi \cdot \dim R}{(1,\dots,1) \cdot \dim R} \label{slope}
\end{equation}
In \eqref{slope}, we interpret $\bar\chi$ and $\dim R$ as dimension
vectors for the extended quiver and take the usual dot product of 
dimension vectors. To include framing spaces in this formalism, one 
can replace them with arrows from an extra vertex $V_0\cong\C$, as first suggested
by Crawley--Boevey. See for example the discussion in 
Section 3.1 of \cite{GinzNak}.

Since $m\gg 0$, the $G\times G'$-semistability of the extended quiver data implies 
$G$-semistability of the original data, and hence its stability because 
in $X$ there are no strictly semistable points. Because of this
stability, any subrepresentation $S$ passes the slope test
automatically except when it contains all or none of the spaces $V_a$. 
In the latter case, we use our choice of $\chi'$ to conclude that 
the semistable representations of the extended quiver are 
the stable representations of the original quiver with a choice 
of injective maps \eqref{flagVp}. 
\end{proof}

\begin{Corollary}\label{corGIT}
The stable points in \eqref{flag_subr} are GIT semistable (=stable)
points for $\cL = \cL_0\otimes \chi^m$, $m\gg 0$, where 
$\cL_0$ any ample line bundle pulled back from 
the product of partial 
flag varieties and 
$\chi$ is the character that gives the stability condition for $X$. 
\end{Corollary}

\subsubsection{}

The natural obstruction theory of quasimaps is constructed relative
the map 
$$
\QM(X) \to \textup{stack of bundles $\{\cV_a\}$} 
$$
to a smooth stack of bundles on the domain $C$.
 The terms in the relative obstruction 
theory are given by the cohomology of the bundles giving 
the quiver data and
the moment map. 

We have 
\begin{equation}
\frac{\left\{
\begin{matrix}
\textup{flags of subspaces in $V_a$} 
\end{matrix}
\right\} }
%\bigg/ 
{\prod GL(V_a)}  = 
\begin{matrix}
\textup{stack of $\Ct_q$-equivariant bundles $\{\cV_a\}$} \\
\textup{with trivial $\Ct_q$-action on $\cV_a\big|_\infty$}  
\end{matrix} \,. 
\label{stackCtq}  
\end{equation}
The inclusion of Corollary \ref{corGIT} 
\begin{equation}
\left(\QM_\textup{nonsing $\infty$}\right)^{\Ct_q} \subset 
\left\{
\begin{matrix}
\textup{a quiver representation} \\
\textup{$+$ flags of subspaces in $V_a$} 
\end{matrix}
\right\} \Bigg/\!\!\!\!\!\!\Bigg/ \prod GL(V_a)\,,  \label{flag_GIT} 
\end{equation}
where the double slash denotes a GIT quotient, can 
be interpreted in quasimap terms as follows.

\subsubsection{}

Observe that 
\begin{equation}
H^0(\cG(-\infty))^{\Ct_q} = H^1(\cG)^{\Ct_q} = 0\label{H10}
\end{equation}
for any equivariant 
bundle $\cG$ on $\bP^1$ such that $\cG\big|_\infty$ is a trivial
$\Ct_q$-module. Here $\cG(-\infty)$ denotes the twist by the 
divisor $\infty\in\bP^1$. Therefore, from the exact sequence 
$$
0 \to \cG(-\infty) \to \cG \to \cG|_{\infty} \to 0 
$$
we get 
\begin{equation}
0 \to  \chi(\cG)^{\Ct_q}  = H^0(\cG)^{\Ct_q}
\to \cG|_{\infty}  \to H^1(\cG(-\infty))^{\Ct_q} \to 0 \,. \label{exact1}
\end{equation}
In particular, \eqref{exact1} applies to the bundles 
$$
\cG = \cHom(\cV_a,\cV_b)\,, \dots\,, 
$$
whose sections are the quiver maps. For them, the middle term in 
\eqref{exact1} gives the vector space $\Hom(V_a,V_b)$, while the 
zero locus of the map 
$$
\Hom(V_a,V_b) \to H^1(\cV_a^\vee \otimes \cV_b (-\infty))^{\Ct_q}
= \bigoplus_k \Hom(\bbV_a[k],V_b/\bbV_b[k]) 
$$
defines, together with the moment map equation, the 
inclusion \eqref{flag_GIT}. Here the summation is over 
all $\bbV_a[k]$ excluding repetitions. The $\Ct_q$-fixed part of 
the moment map equations  is a section of 
$$
\hbar^{-1} \otimes \bigoplus_a \Hom_\textup{flag}(V_a,V_a) = 
\hbar^{-1} \otimes \bigoplus_a \Hd(\cV_a^\vee \otimes \cV_a)^{\Ct_q} \,,
$$
where the subscript in $\Hom_\textup{flag}$ denotes maps 
that preserve the filtration by $\bbV_a[k]$. 

This completes 2/3 of the plan outlined in Section \ref{s_loc_plan}.

\subsection{Integral formulas}\label{s_integr} 

\subsubsection{}

In the full vertex function, the (symmetrized) virtual structure 
sheaf of the fixed locus enters simultaneously with the 
contributions of the moving parts of the quasimap 
obstruction theory. This leads to the following formulas 
for the vertex functions. (The necessary background information about K-theoretic computation on
GIT quotients is collected in the Appendix.)

\subsubsection{}

Let $\bS \subset \prod GL(V_a)$ be a maximal torus. 
The $\bS\times \bT$-fixed points on the prequotient in 
\eqref{flag_GIT} correspond to coordinate flags and 
zero quiver maps. Coordinate flags mean that $\bS$ appears
as a group preserving a splitting 
\begin{equation}
\cV_a = \bigoplus_\alpha \cL_{a,\alpha}\,, \quad
\cL_{a,\alpha}=\cO(d_{a,\alpha}[0])\,, \quad d_{a,\alpha} = \deg 
\cL_{a,\alpha}\,, \label{Vsplit}
\end{equation}
into a direct sum of line bundles. We denote by $s_{a,\alpha}$ the 
$\bS$-weight of $\cL_{a,\alpha}$. These are the coordinates in 
$\bS$ and the equivariant Chern roots of the bundles $V_a$ over 
$X$. These very same variables were denoted by $s_{a, \alpha} = x_{a, \alpha}^{-1}$, elsewhere in the paper. With our conventions, 
$$
\textup{weight}\left(\cL_{a,\alpha}\big|_0\right) = 
q^{d_{a,\alpha}} s_{a,\alpha} \,, \quad
\textup{weight}\left(\cL_{a,\alpha}\big|_\infty\right) = 
s_{a,\alpha}\,,
$$ 
where the weights are for the torus $\bS\times \bT\times \Ct_q$.

As we will see, integral formulas for the vertex will be more 
naturally written not in terms of 
the variables $s_{a,\alpha}$ but rather in terms of the weights 
of $\{\cV_a\}$ over the point  $0\in C$. 

\subsubsection{}

The basic ingredient in the integration formulas is the function 
\begin{equation}
\aroof(s) = s^{1/2} - s^{-1/2} \label{defaroof}
\end{equation}
extended to equivariant $K$-theory as a genus, that is, 
so that 
\begin{equation}
\aroof(\cG_1 + \cG_2) = \aroof(\cG_1) \, \aroof(\cG_2)
\,. \label{multK}
\end{equation}
The importance of this function is clear from the equality 
\begin{equation}
\left(\cO_\vir \otimes \cK_\vir^{1/2} \right)_\textup{moving} = 
\aroof\left(T_{\vir,\textup{moving}}\right)^{-1}\label{locaroof}
\end{equation}
for the moving part of the virtual structure sheaf in localization 
formulas. Formula \eqref{locaroof} is an immediate consequence 
of localization formula for $\cO_\vir$, see \cite{GP,CF3}. Here and 
in what follows the \emph{moving} terms are the terms of nontrivial 
weight with respect to $\bS\times \bT\times \Ct_q$. 

In particular, an algebraic consequence of the 
identification \eqref{stackCtq} is that 
$$
\frac{\Delta_\textup{Weyl} }
{\aroof(T \textup{Flags in \eqref{stackCtq}})} =  
\aroof(\Lie \textstyle{\prod} \Aut(\cV_a)^{\Ct_q}/\bS)  \,. 
$$
Thus the integration measure in \eqref{intchi} in the specific 
setting of \eqref{flag_GIT} naturally becomes a part of the 
localization weight \eqref{locaroof}, namely the part that comes from 
$\Ct_q$-equivariant automorphisms 
of $\{\cV_a\}$ other than those in $\bS$. 

We conclude the following 

\begin{Proposition}
For any $\cF\in K_{\bT}(X)$, we have 
\begin{equation}
  \label{VI1}
  \chi(X, \Vertex \otimes \cF) = \frac{1}{|W|} 
\sum_{\{d_{a,\alpha}\}} q^{-\frac12 \deg \cT^{1/2}}\prod_{a,\alpha}
z_a^{d_{a,\alpha}}
\int_{\gamma_\chi}  
\frac{\cF(s) \, d_\textup{Haar}s}{\aroof\left(T_{\vir,\textup{moving}}\right)}
\end{equation}
where the summation is over all splittings \eqref{Vsplit}, 
$\cT^{1/2}$ denotes the virtual bundle on $C$ induced by 
the chosen polarization, the cycle $\gamma_\chi$ corresponds
to a choice of stability parameter $\chi$ as in \eqref{def_gamma_chi}, 
and $\cF(s)$ is the expression of $\cF$ in the Chern roots of the 
tautological bundles. 
\end{Proposition}

Note that by Lemma \ref{l_dim0} in Section \ref{s_large_chi} 
the integration $\int_{\gamma_\chi}$
in \eqref{VI1} is really an iterated residue of the integrand. 

\subsubsection{}\label{s_two_ways} 

The summation over splittings in \eqref{VI1} can be treated in two 
complementary ways. 

On the one hand, one can sum over the whole lattice of splittings, and
this will be convenient for interpreting the eventual integral
\eqref{VI4} as a linear functional invariant under $q$-shifts. 

On the other hand, for most splittings, there are no stable quasimaps and hence those 
contribute zero to the sum in \eqref{VI1}. We call splittings that 
correspond to stable $\Ct_q$-fixed quasimaps \emph{effective}. 
A necessary condition for a splitting to be effective is discussed
in Section 7.2 of \cite{OK}.

\subsubsection{}

Formulas \eqref{TvirQM} and \eqref{cTcT} show that the following Lemma 
applies to the denominator in \eqref{VI1}.

\begin{Lemma}
For any bundle $\cG$ on $\bP^1$ we have 
\begin{equation}
  \label{arar}
q^{-\deg \cG/2} \,\, 
  \frac{\aroof\left(\cG|_\infty + \hbar^{-1} \cG^\vee|_\infty\right)}
{\aroof\left(\Hd\left(\cG + \hbar^{-1} \cG^\vee\right)\right)} = 
\left(-\hbar^{1/2}\right)^{-\deg \cG}
\,\, 
\frac{\varphi\left(\hbar \cG|_\infty\right) \varphi\left(q \cG|_0
  \right)}
{\varphi\left(q \cG|_\infty\right)  \varphi\left(\hbar \cG|_0
  \right)}\,,
\end{equation}
where $\varphi$ is the function \eqref{defvarphi} 
extended multiplicatively as in \eqref{multK}. 
\end{Lemma}

\begin{proof}
It is enough to prove \eqref{arar} for a line bundle, in which case it 
reduces to an elementary identity. 
\end{proof}

Note that for $\cG=\cT^{1/2}$ the prefactor in the LHS of \eqref{arar}
coincides with the power of $q$ in \eqref{VI1}, while the numerator 
in the LHS of \eqref{arar} is nothing but $\aroof(TX)$. 
As in Section 8.3 of \cite{OK}, we incorporate
the prefactor in RHS of \eqref{arar} into a shift $z_\#$ of the variable $z$
so that 
$$
z_\#^{\deg f}  = \left(-\hbar^{1/2}\right)^{-\deg \cT^{1/2}} \prod_{a,\alpha}
z_a^{d_{a,\alpha}}  \,. 
$$
With this notation, \eqref{VI1} may be restated as follows 
% %
% \begin{multline}
%   \label{VI2}
%   \chi(X, \Vertex \otimes \cF) = \\ 
% \frac{1}{|W|} 
% \sum_{\{d_{a,\alpha}\}} z_\#^{\deg f}
% \int_{\gamma_\chi}  
% \frac{\cF(s) \, d_\textup{Haar}s}{\aroof\left(TX_\textup{mov} \right)}\,\, 
% \frac{\varphi\left(\hbar T^{1/2}_\textup{mov}\right) 
% \varphi\left(q \cT^{1/2}_\textup{mov}\big|_0 \right)}
% {\varphi\left(q T^{1/2}_\textup{mov}\right)  
% \varphi\left(\hbar \cT^{1/2}_\textup{mov}\big|_0\right)}\,,
% \end{multline}
% %
%
\begin{multline}
  \label{VI2}
  \chi(X, \Vertex \otimes \cF) = \\ 
\frac{1}{|W|} 
\sum_{\{d_{a,\alpha}\}} z_\#^{\deg f}
\int_{\gamma_\chi}  
\left(
\frac{\cF(s) \, d_\textup{Haar}s}{\aroof\left(T \right)}\,\, 
\frac{\varphi\left(\hbar T^{1/2}\right) 
\varphi\left(q \cT^{1/2}_0 \right)}
{\varphi\left(q T^{1/2}\right)  
\varphi\left(\hbar \cT^{1/2}_0\right)}
\right)^{\sim}_\textup{moving}\,,
\end{multline}
where $\cT^{1/2}_0$ denotes the fiber of $\cT^{1/2}$ over $0\in C$, 
 tilde refers to the computation on the prequotient, and only 
moving terms are retained from the product of $\aroof$- and 
$\varphi$-functions. 

\subsubsection{}

{}From the point of view of difference equations, a better
normalization of the vertex functions is the following 
  \begin{equation}
    \label{tVx}
    \Vx = \be(z_\#) \, \varphi((q-\hbar)
    T^{1/2}) \, \Vertex  \,, 
  \end{equation} 
see Section 8.3 in \cite{OK} and also Section 6.1 in \cite{ese} (${\Vx}$ here was denoted by ${\tVx}$, in those papers). It 
solves certain $q$-difference equations in both the equivariant variables
and the  K\"ahler variables $z$. 
Here 
\beq\label{bige}
\be(z) = \exp\left(\frac{\bla(\ln z ,\ln t)}{\ln q}\right)\,, 
\eeq
where
$$
\bla: H^2(X,\C) \otimes \Lie \bT  \to \End K(X^\bT) \otimes \C
$$
extends by linearity the map that takes a line bundle 
$\cL\in \Pic(X)$ to the logarithm of the operator 
$\cL\otimes \textup{---}$, see the discussion in Section 8.2 of 
\cite{OK}. This function has an elementary description in terms 
of the prequotient. Indeed, the line bundle $\det V_a$ associated to the 
variable $z_a$ has weight $\prod_\alpha s_{a,\alpha}$, 
whence 
\begin{equation}
\bla  = \sum_{a,\alpha} \ln(z_a) \, \ln(s_{a,\alpha})
\,. \label{bla}
\end{equation}

A certain care is required in working with \eqref{tVx} because
$\phi(\hbar T^{1/2})$ may contain nonequivariant noninvertible 
factors (since these singular terms involve neither equivariant not 
K\"ahler variables, they are irrelevant from the point of view of 
difference equations). To avoid these complications, we define 
  \begin{align}
    \bVx &= \frac{\aroof(T)}{\theta(T^{1/2})} 
\Vx\notag\\
&= \hbar^{-\frac14 \dim X} \, \frac{\be(z_\#)}{\varphi(q T^\vee)} \, 
\Vertex 
    \label{bVx} \,, 
  \end{align} 
where $\theta(s)$ is the odd theta function defined in \eqref{vth}.
%
%\begin{equation}
 % \label{vth}
  %\theta(s) = \varphi(qs) \, \varphi(1/s) 
%\end{equation}
%
Note a slight difference 
with the odd theta function $\vartheta(s)= s^{1/2}  \theta(s) 
$ used in  \cite{ese}. 
We extend \eqref{vth} multiplicatively as
before. Note that 
the division of the $\Vx$-function by the theta function of 
the polarization $T^{1/2}$ is a part of the pole subtraction operator 
of \cite{ese}, see Section 6.3 there.

Substituting \eqref{bVx} in \eqref{VI2}, we obtain 
\begin{multline}
  \label{VI3}
  \chi(X, \bVx \otimes \cF) = \\ 
\frac{1}{|W|} 
\sum_{\{d_{a,\alpha}\}} z_\#^{\deg f}
\int_{\gamma_\chi}  \exp\left(\frac{\bla(z_\#,s)}{\ln q}\right)
\frac{\cF(s) \, d_\textup{Haar}s}{\theta(T^{1/2})}\,\, 
\frac{\varphi\left(q \cT^{1/2}_0\right)}
{\varphi\left(\hbar \cT^{1/2}_0\right)} 
\,,
\end{multline}
where the computation on the prequotient and the extraction of the
moving parts is understood. 

\subsubsection{}\label{s_fdef}
Let 
$$
\bd = \{ d_{a,\alpha} \}
$$
denote an effective splitting \eqref{Vsplit} and define 
$$
q^{\bd} s = \{ q^{d_{a,\alpha}} s_{a,\alpha} \} \,. 
$$
These are the weights of the bundles $\cV_a$ over $0 \in C$. Clearly, 
$$
\frac{\varphi\left(q \cT^{1/2}_0\right)}
{\varphi\left(\hbar \cT^{1/2}_0\right)}  = 
\left. \frac{\varphi\left(q T^{1/2}\right)}
{\varphi\left(\hbar T^{1/2}\right)} \right|_{s\mapsto q^\bd s}  \,.
$$
Also, from \eqref{bla}, we have
$$
z_\#^{\deg f}
\exp\left(\frac{\bla(z_\#,s)}{\ln q}\right) = 
\left. \exp\left(\frac{\bla(z_\#,s)}{\ln q}\right) \right|_{s\mapsto
  q^\bd s}  \,. 
$$ 
Therefore, it is natural to change variables in the 
integral \eqref{VI3}. So far, we made no assumptions on 
insertion $\cF(s)$. In \eqref{VI3}, it can be an arbitrary element 
of $K_\bT(X)$ or, more generally, an arbitrary analytic function on 
the spectrum of the ring $K_\bT(X)\otimes \C$ . We now assume it 
has the same automorphy as $\theta(T^{1/2})$, that is, 
we assume 
\begin{equation}
\frac{\cF(s)}{\theta(T^{1/2})} \quad 
\textup{is invariant under $s\mapsto
  q^\bd s$} \label{autcF} \,.
\end{equation}
This means $\cF$ is a section of a certain line bundle over the 
the scheme $\textup{Ell}_\bT(X)$, the equivariant elliptic cohomology 
of $X$. In principle, this section is allowed to have
singularities away from 
the integration cycle. 
With this assumption, a change of variables in \eqref{VI3} 
gives the following

\begin{Proposition}\label{p_bVx_int} 
For any insertion $\cF$ satisfying \eqref{autcF}, we have 
\begin{equation}
  \label{VI4}
  \chi(X, \bVx \otimes \cF) = 
\frac{1}{|W|} 
\int_{\sum q^{\bd} \cdot \gamma_\chi}  \exp\left(\frac{\bla(z_\#,s)}{\ln q}\right)
\frac{\cF(s) \, d_\textup{Haar}s}{\phi(T^\vee_\textup{moving})} 
\,,
\end{equation}
where the sum of residues is over all effective shifts of the cycle
$\gamma_\chi$. 
\end{Proposition}

Recall that the cycle $\gamma_\chi$ is, by construction, a sum 
of several cycles. For insertions $\cF$ supported on its 
particular components (such as the classes of torus-fixed 
loci in $X$), the integration cycle will be correspondingly 
smaller.

\subsubsection{}

As explained in Section \ref{s_two_ways}, the integration 
contour in \eqref{VI4} may be extended to all $q$-shifts of 
$\gamma_\chi$ as long as $\cF(s)$ is nonsingular on $\gamma_\chi$. 
Obviously, the integration 
$\displaystyle \int_{\sum q^{\bd} \cdot \gamma_\chi} d_\textup{Haar}s$ 
where the sum is over the whole lattice of splittings, is invariant
under $q$-shifts.

\subsubsection{}

For a simplest example, let us examine the statement of 
Proposition \ref{p_bVx_int} for 
$$
X = T^* \bP^{n-1} \,. 
$$
We will follow the notations of Section 6.2 of \cite{ese} 
and of Section \ref{s_ex_Pn} above. We take 
$$
\bT = \Ct_\hbar \times \bA 
$$
where $\bA$ is as in \eqref{bAGL} and the first factor scales 
the cotangent directions with weight $\hbar^{-1}$. We 
denote the $\bT$-fixed points by 
$$
X^{\bT} = \{ p_1, \dots, p_n\} \,. 
$$
The elementary analysis of the 
quasimap spaces recalled in \cite{ese} shows
\begin{equation}
  \chi(\Vx \otimes \cO_{p_k}) =  \frac{\hbar^{\frac14 {\dim X}}}{2\pi i}  \,
\int_{\gamma_k} \frac{ds}{s} \,  e^{\dfrac{\ln z_\# \, \ln s}{\ln q}} 
\varphi((q-\hbar) \, T^{1/2})_\textup{moving} \label{PnInt} \,,
\end{equation}
where
\begin{equation}
T^{1/2} = -\frac{1}{\hbar} + \sum_i \frac{1}{\hbar a_i
  s}\,,
\label{polprefac}
\end{equation}
and the contour $\gamma_k$ enclosed the poles 
\begin{equation}
x =  \frac{q^d}{a_k} \,, \quad d=0,1,2,\dots \,. \label{poles_k}
\end{equation}
In \eqref{PnInt} we restored the power of $\hbar$ that 
comes from $\cK_X \cong \hbar^{\frac12 \dim X} \cO_X$. 

Tautologically, 
$$
 \chi(\Vx \otimes \cO_{p_k}) =
\chi
\left(\bVx \otimes \frac{\cO_{p_k}}{\aroof{(T)}} \otimes
  \theta(T^{1/2})
\right) \,. 
$$
We have 
\begin{equation}
\left. \frac{\cO_{p_k}}{\aroof{(T)}}  \right|_{p_i} =
\hbar^{\frac14\dim X}
\delta_{ki} \,, 
\label{indic}
\end{equation}
which means that this insertion serves as a delta-function 
restricting the residues to the sequence \eqref{poles_k}. Thus
setting 
$$
\cF(s) = \theta(T^{1/2}) \otimes \cF'(s) 
$$
where $\cF'(s)$ is $q$-periodic and nonsingular at the 
points $\{a_i^{-1}\}$, we get from \eqref{PnInt} 
\begin{align*}
  \chi
\left(\bVx \otimes \cF(s)\right) &= 
\int_{\gamma} \frac{ds}{2\pi i s} \,\,  e^{\dfrac{\ln z_\# \, \ln s}{\ln
  q}} \, \cF'(s) \, 
\varphi((q-\hbar) \, T^{1/2})_\textup{moving} \\
&= 
\int_{\gamma} \frac{ds}{2\pi i s} \,\,  e^{\dfrac{\ln z_\# \, \ln s}{\ln
  q}} \, 
\frac{\cF(s)} 
{\varphi(T^\vee)_\textup{moving}} \,,
\end{align*}
where $\gamma = \sum_{k=1}^n \gamma_k$. This is a specialization of 
the general formula \eqref{VI4}.

\subsubsection{} 

Heuristically, it may be argued that \eqref{VI4} an
infinite-dimensional version of the formula \eqref{intW}, in which 
\begin{alignat*}{2}
  \tX  \quad &\mapsto &\quad&\QM(\tX) \\
G \quad &\mapsto &\quad&\textup{gauge transformations} \,. 
\end{alignat*}
Such or similar viewpoint is implicit in many papers on supersymmetric
gauge theories. Here, we don't try to turn this heuristic into precise
mathematical statements. The argument given above is technically much
simpler and sufficient for our needs.

\subsection{Vertex functions and ${\cal W}_{q,t}$ algebra correlators} 

In this section, we prove the following:
\begin{theorem}\label{t:W}
The vertex function  
 \begin{equation} \label{VI4b}
  \chi(X, \Vx \otimes \cF') = 
\frac{1}{|W|} 
\int_{\gamma_\chi}  \exp\left(\frac{\bla(z_\#,s)}{\ln q}\right)
{\cF'(s)} {\phi((q-\hbar) T^{1/2})}  \, d_\textup{Haar}s,
\end{equation}
is a 
 ${\cal W}_{q,t}(\fg)$ correlator 
\begin{equation}\label{match}
 \langle \mu'| \prod_{a, i} V_a^{\vee}(a_{a,i}) \; \prod_{a} (Q^{\vee}_a)^{d_a}|\mu\rangle,
\end{equation} 
 where a choice insertion $\cF'$ corresponds to a choice of ${\cF'(s)}$  contours of integration in the definition of screening charge operators, and $\mu=\frac{z}{\ln q}$. 
\end{theorem}

{\it Proof.} In section \ref{s:explicit}, we gave an explicit integral form of the   ${\cal W}_{q,t}(\fg)$ algebra correlator  in \eqref{match}. We will now show that this exactly equals the integral in \eqref{VI4b}.

Consider the ${\phi((q-\hbar) T^{1/2})}$ in the integrand \eqref{VI4b}.
Recall that for the Nakajima variety $X$, with quiver $\scQ$
$$T^{1/2} X = \sum_a V_a \otimes W_a^* +  \sum_{a , b} (I_{ab} - \delta_{ab}) \,V_a\otimes V_b^* 
$$
where $I_{ab}$ is the adjacency matrix of the Dynkin diagram, and we have identified the vector space $T^{1/2}X$ with its character under the action of the torus $S\times T$. We can choose coordinates on $S$ 
so that
$$
V_a = \sum_{\alpha} x_{a, \alpha}\,\hbar^{a/2}, \qquad W_a= \sum_{i} a_{a, i} \,\hbar^{(a-1)/2} ,
$$
where relative to conventions elsewhere in this section, $x_{a, \alpha} = s^{-1}_{a, \alpha}$, and the powers of $\hbar$ are a convenient choice of coordinates. (Hopefully, the reader will distinguish the subscript $a$ labeling the node of ${\scQ}$ and taking values from $1$ to ${\rm rk} \,{\fg}$.) With this, the contributions to ${\phi((q-\hbar) T^{1/2})}$ are:

\begin{itemize}
\item[---]  From $\Hom(V_a, W_a)$, we get 
\beq\label{sctdasl3b}
\prod_{\alpha, i}  { \varphi( q \,x_{\alpha,a}/\hbar^{1/2}  a_{i,a})\over \varphi(\hbar { x_{\alpha,a} / \hbar^{1/2}a_{i,a}})}.
\eeq
This coincides with \eqref{sctdasl3} if we recall that $v_a=\hbar^{1/2}$ and $t=q/\hbar$.
$$
\prod_{\alpha , i} \langle S^{\vee}_a(x_{a,\alpha})  V^{\vee}_a( a_{a,i})  \rangle
$$
\item[---] From for every pair of nodes adjacent nodes $a,b$ with $I_{ab}=1$, we get
$$
\prod_{\alpha, \beta}  { \varphi( q  \hbar^{a/2} x_{a,\alpha}/ \hbar^{b/2} x_{b,\beta})\over \varphi(\hbar \,  \hbar^{a/2} x_{a,\alpha} / \hbar^{b/2}x_{b,\beta})}
$$
from the contributions of $Hom(V_a, V_b)$ to $T^{1/2}X$. This coincides with 
$$
\prod_{\alpha , \beta} \langle S^{\vee}_a(x_{a,\alpha})   S^{\vee}_b(x_{b,\beta})   \rangle
$$
in  \eqref{sctdasl2}. 
\item[---] From $\Hom(V_a, V_a)$, we get
$$
\prod_{\alpha\neq  \beta}  { \varphi( \hbar \,x_{a,\alpha}/  x_{a,\beta})\over \varphi(q x_{a,\alpha} /x_{a,\beta})}
$$
up to an overall constant.
Recall that \eqref{sctdasl1}
$$
\prod_{\alpha < \beta} \langle S^{\vee}_a(x_{a,\alpha})S^{\vee}_a(x_{a,\beta}) \rangle = \prod_{\alpha \neq \beta}{ \varphi(  x_{a,\alpha}/x_{a,\beta})\over \varphi(   t\,x_{a,\alpha}/x_{a,\beta})}  \prod_{\alpha < \beta}{ \theta( t x_{a,\alpha}/x_{a,\beta})\over \theta(    \,x_{a,\alpha}/x_{a,\beta})} 
$$
Using $t=q/\hbar$, and the $\theta(x) = \theta(q/x)=\varphi(x)\varphi(q/x)$ property of theta function, above coincides with
\beq\label{thcol}
\prod_{\alpha < \beta} \langle S^{\vee}_a(x_{a,\alpha})S^{\vee}_a(x_{a,\beta}) \rangle = \prod_{\alpha \neq \beta}{ \varphi( \hbar x_{a,\alpha}/x_{a,\beta})\over \varphi(   q\,x_{a,\alpha}/x_{a,\beta})}  \prod_{\alpha<\beta}{ \theta( q x_{a,\alpha}/x_{a,\beta})\over \theta( \hbar  \,x_{a,\alpha}/x_{a,\beta})} 
\eeq
up to the ratio of theta functions.
\end{itemize}
In summary, we showed that the contribution of ${\phi((q-\hbar) T^{1/2})}$ to \eqref{VI4b} coincides with the contribution of $\Phi(x, a)$ to \eqref{match}, up to the collection of theta functions in \eqref{thcol}. The exponential terms in \eqref{VI4b} correspond to the exponential $x^\mu$ terms in \eqref{Wai}, with identification  
$$
z_{a} = q^{(\mu, e_a)}.
$$
From perspective of the difference equations the effect of the ratio of the theta functions in \eqref{thcol} is to produce a shift in the Kahler variables $z_a = q^{(\mu, e_a)}$ by a power of $\hbar^{1/2}$. These shifts are collected in \eqref{VI4b} in replacing  $z \rightarrow z_{\#}$.
This proves of the theorem.

\subsubsection{}\label{s_matchw}
Explicitly, for $X=T^*{\mathbb P}^{n-1}$, the right hand side of \eqref{match} becomes
$q$-conformal block
$$
\langle \mu'| V^{\vee}(a_{1}) \ldots V^{\vee}(a_{n}) \; Q^{\vee}|\mu\rangle.
$$
of the ${\cal W}_{q,t}(A_1)$ algebra (the algebra is also known as the $q$-Virasoro algebra.) The algebra has a single family of generators $e[k]$, $k\in {\mathbb Z}$, satisfying
$$
[e[k], e[\ell]] = {1\over k} (q^{k\over 2} - q^{-{k\over 2}})(t^{{k\over 2} }-t^{-{k\over 2} })(q^{k\over 2} t^{-{k \over 2}} + q^{-{k\over 2}}t^{k\over 2}) \delta_{k, -\ell},
$$
with Fock representation $\pi_{\mu}$  defined as 
$$
e[k]|\mu\rangle =0,\;\; \textup{for}\;\; k>0, \;\; \textup{and} \;\; e[0] |\mu\rangle = (\mu, e) |\mu \rangle.
$$
The screening charge operator is
\beq\begin{aligned}\label{screena1}
Q^{\vee} &= \int dx \; S^{\vee}(x)\;: \pi_0\rightarrow \pi_{- e \beta}
\end{aligned}
\eeq
where
\beq\label{magneticSa1}
S^{\vee}(x) =[\ldots]x^{-e[0]} : \exp\Bigl(\sum_{k\neq 0}{ e[k] \over q^{k\over 2} - q^{-{k \over 2}}} x^k\Bigr):.
\eeq
The $[\ldots]$ in stand in for operators responsible for the shift of $\mu$ in \eqref{screen}.The magnetic degenerate vertex operator is
$$
V^\vee(x) =[..] x^{w[0]}:  \exp\Bigl( - \;\sum_{k\neq 0} {w[k]\over q^{k\over 2} - q^{-{k \over 2}}}  \, x^{k}\Bigr) :
$$
where $w[k] = e[k]/(q^{k\over 2} t^{-{k \over 2}} + q^{-{k\over 2}}t^{k\over 2})$, and the dots stand for operator responsible for shifting the Fock vacuum
$$
V^\vee(x) \;: \pi_0\rightarrow \pi_{ w \beta}.
$$
From the definitions, one computes
$$
\langle \mu'|  V^{\vee}(a_{1}) \ldots V^{\vee}(a_{n}) \; Q^{\vee}|\mu\rangle = 
\int {dx} \; x^{-(\mu,e)-1}    \Phi(x,a).
$$
where
$$\Phi(x,a)=  \prod_{j=1}^n {\varphi(t x/ a_j) \over \varphi( x /a_j)},
$$
$t=q/\hbar$, and
$$
\mu' ={ \mu}+ (n w - e)\beta.
$$
By inspection, $\Phi(x,a) = \varphi((q-\hbar)T^{1/2})$, with $z_{\#} = q^{(\mu, e)}.$

%\subsubsection{}
%For future reference, lets compute the integral 
%\beq\label{vertexA}
%{\bf V}\big|_{p_{\ell}}  = \int_{\gamma_{\ell}} {dx}  \; x^{\lambda-1}   \Phi(x,a),
%\eeq
%by residues, with contour defined as in \eqref{poles_k}. We find it
%equals:
%\begin{align}
%  \label{vertex_hyper}
%\Vx\big|_{p_\ell}&= (a_{\ell})^{\lambda}\, 
 %\frac{\varphi( q/\hbar)}{\varphi(q)}  \prod_{i\ne \ell} \frac{\varphi(q a_{\ell}/\hbar a_i)}{\varphi(a_{\ell}/a_i)} \,  \bF \left[ \left. 
%\begin{matrix} \hbar a_1 /a_{\ell}, &  \hbar a_2/ a_{\ell}, &\dots \\
%                      q a_1/ a_\ell, &q a_2/a_\ell, &\dots
%\end{matrix} \,\, \right| \, z/{\hbar}^n \,
%                                               \right] \,.
%\end{align}
%where we used that $q^n q^{-\lambda}  =z$.

\subsubsection{}
While the ${\cal W}_{q,t}(\fg)$ algebra has a nice conformal limit, the same is not true of quantum $K$-theory. While we can formally take the limit \eqref{limitw} of the generating functions, their contributions have no enumerative meaning. (There is a natural limit of the theory where degenerate counts in K-theory to cohomology, but this is not the limit we need here.)
\section{Vertex functions and qKZ}

\subsection{Degeneration formula}\label{s_degenr}

\subsubsection{}
Recall that the domain C of the quasimaps \eqref{fquasi} is a fixed,
that is, parametrized curve. We can let it degenerate to a union $C_0$
of two rational curves, e.g.\ by taking a trivial family $C \times \C$ 
over $\C$ and blowing up a point $(c,0) \in C\times \C$. We denote 
by $\varepsilon$ the parameter of the degeneration and write
$\C_\varepsilon$ to denote the base 
$$
\pi: \bfC = \textup{Bl}_{(c,0)} C\times \C_\varepsilon \to \C_\varepsilon 
$$
of the degenerating family. 

Clearly
$$
C_\varepsilon = \pi^{-1}(\varepsilon) \cong C
$$
canonically for $\varepsilon \ne 0$, while 
the special fiber of the new family is the 
union 
\begin{equation}
C_0 = C_{0,1} \cup C_{0,2} \,, \quad C_{0,1} \cong C \label{C0}
\end{equation}
of two rational curves joined at the point $c\in C_{0,1}$. If $c=0$, that is, 
if $c$ is a fixed point of $\Ct_q$ other than $\infty\in C$, then this 
degeneration is $\Ct_q$-equivariant. 

\subsubsection{}

A key geometric question is to put a good central fiber into the 
family $\QM(C_{\varepsilon} \to X)$ over $\C_\varepsilon \setminus
\{0\}$ and it is answered by a beautiful theory due principally to 
Jun Li, see \cite{Li1,Li2,LiWu} and also, for example, \cite{OK} for an introductory
discussion. 
 The central fiber, which we still denote 
$\QM(C_0 \to X)$, is the moduli space of quasimaps from a 
nodal curve $C_0$; however, an important geometric idea is 
hidden here in the definition of a quasimap from a nodal source
curve. 

To keep the obstruction theory perfect, quasimaps need to be
nonsingular at the nodes of the source curve. To satisfy this 
constraint and the usual properness requirements at the same time, 
one has to say
what to do with a $1$-parameter family of quasimaps that 
develops a singularity at the node of a special fiber. It is treated
by version of the familiar semistable reduction process, in which the 
offending node is being blown up until the singularity at it goes
away. In the process, the node becomes replaced by a chain of rational 
curves, considered up to an isomorphism fixing the points at 
which it is attached to the original nodal curve. 

This motivates defining $\QM(C_0 \to X)$
 as the moduli spaces of quasimaps 
of the form 
\begin{equation}
\xymatrix{
C'_0 \ar[d]^g
\ar@{-->}^{f'}[rr]&& X\\
C_0
}\label{def_rel2}
\end{equation}
in which
\begin{enumerate}
\item[---] the map $g$ collapses a chain of rational curves to the
  node of $C_0$, 
\item[---] $f'$ is nonsingular at the nodes 
of $C'_0$, 
\item[---] the automorphism group of $f'$ is finite. 
\end{enumerate}
Here the source of automorphisms is the group 
$$
\Aut(C'_0,g) = \left(\Ct\right)^{\textup{\# of new components}} \,. 
$$
Pictorially, the domain $C_0'$ may be represented as in Figure
\ref{f_springs}. As usual, one of the uses of the finiteness of 
$\Aut(f')$ is to prevent unnecessary blowups in the course of 
the semistable reduction. 
\begin{figure}[!htbp]
  \centering
   \includegraphics[scale=0.5]{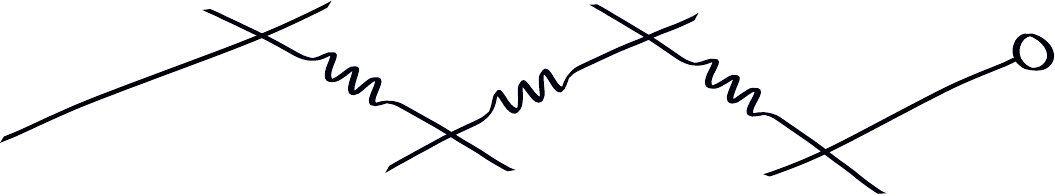}
 \caption{\sl
A semistable curve $C_0'$ whose stabilization is the nodal curve
   $C_0$. Components with $\Ct$ automorphisms are indicated by
   springs; they are often called accordions. The point $\infty\in C \cong \C_{0,1}$ at which the
   quasimaps are required to be nonsingular is indicated by a
   circle. }
  \label{f_springs}
\end{figure}

\subsubsection{}

The family 
\begin{equation}
\QM(C_\varepsilon \to X) \to \C_\varepsilon \label{famQM}
\end{equation}
has a natural relative obstruction theory, given by the cohomology,
that is, push-forward along $\pi$, of quiver sheaves in
question. Its restrictions to fibers of \eqref{famQM} is the
obstruction theory for the spaces $\QM(C_\varepsilon \to X)$ and hence
the virtual structure sheaves and the symmetrized virtual structure
sheaves of these spaces fit into a flat family over $\C_\varepsilon$. 
This means, one can count quasimaps from $C$ in terms of 
quasimaps from $C_0$. 

\subsubsection{}

Quasimaps from $C_0$ can, in turn, be glued out of pieces that
correspond to the pieces in the domain curve in Figure
\ref{f_springs}. Indeed, moduli of 
quasimaps from a fixed nodal curve $C_0'$ 
nonsingular at the nodes are the product of moduli of 
quasimaps from the
components over the evaluation maps at the nodes. Because
 the number of
the accordions, that is, nonparametrized components of $C_0'$ is
dynamical, the correct decomposition to take is the one 
depicted in Figure \ref{f_pieces}. 
\begin{figure}[!htbp]
  \centering
   \includegraphics[scale=0.4]{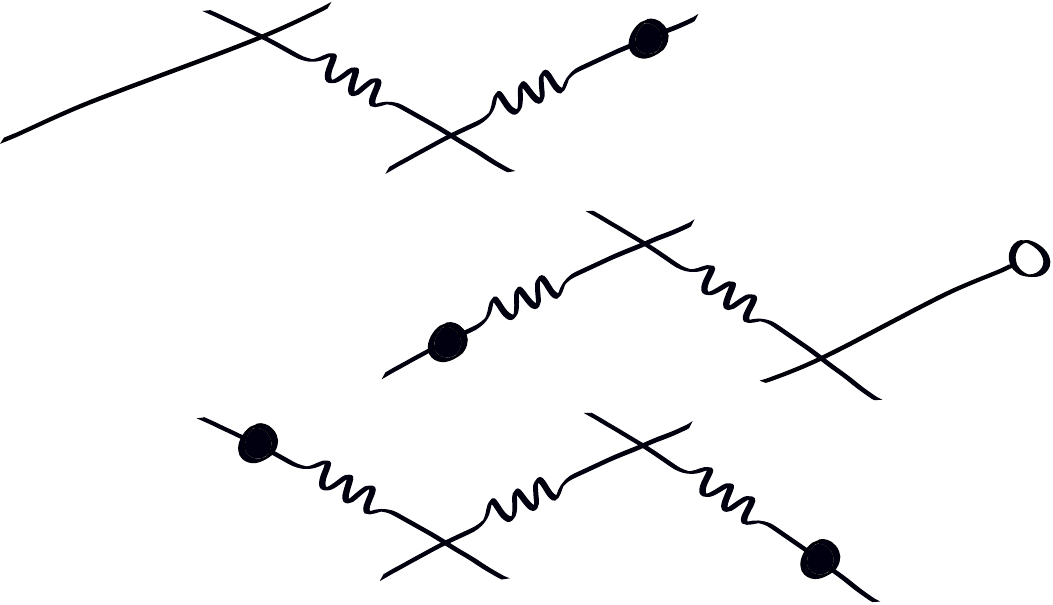}
 \caption{\sl
Three kind of moduli spaces that appear as pieces in the degeneration
formula. In the first line, $C_{0,2}$ is linked by a chain of
accordions to a marked point (bold circle), at which quasimaps must be
nonsingular. In the second line, we have the same for $C_{0,1}$
together with original evaluation point $\infty\in C \cong
C_{0,1}$. In the third line, the domain is a chain of accordions
joining two marked points. }
  \label{f_pieces}
\end{figure}

\subsubsection{}

The difference between the new marked points shown in bold and the
original evaluation point $\infty\in C \cong
C_{0,1}$ is the following. While evaluation at $\infty$ requires
explicitly throwing out quasimaps with singularities there, 
singularities cannot get to the bold points by the nature of  moduli
spaces. Any time a singularity tries to get to the point $\bullet$, a new
accordion opens (by semistable reduction), and the point $\bullet$ gets
away. 

This new kind of quasimaps are called quasimaps \emph{relative} a 
point $\bullet$ of the domain. The above informal discussion 
means formally that the evaluation maps 
\begin{alignat}{3}
\ev_\bullet: \,\, &\QM(C_{0,1})_\textup{relative $\bullet$} &&\to &&\quad X
\notag\\
\ev_{\bullet,\bullet}: \,\, &\QM(\textup{accordions})_\textup{relative $\bullet$,$\bullet$}
&&\to &&\quad X \times X 
  \label{ev2}
\end{alignat}
are \emph{proper}. Using them, we can define
\begin{equation}
\Cp  = \ev_{\bullet,*} \left(
\tO_\vir \,  z^{\deg f} \right) \in
K_{\bT\times\Ct_q}(X)\otimes \Q[[z]] 
\label{defCap}
\end{equation}
and 
\begin{equation}
\Glue  = \ev_{\bullet,\bullet,*} \left(
\tO_\vir \,  z^{\deg f} \right) \in
K_{\bT}(X^2)\otimes \Q[[z]] \,, 
\label{defGlue}
\end{equation}
where localization is not required in contrast to \eqref{defVertex}. 
Note that \eqref{defGlue} does not depend on $q$ since $\Ct_q$ acts
trivially on quasimaps from nonparametrized curves. 
The absence of denominators makes these tensors much simpler
objects than the vertex, or than its analog 
\begin{equation}
\Fusion  = \ev_{\bullet,\circ,*} \left(
\tO_\vir \,  z^{\deg f} \right) \in
K_{\bT\times \Ct_q}(X^2)_\textup{localized}\otimes \Q[[z]] \,. 
\label{defFusion}
\end{equation}
Here the quasimaps are from the domain shown in the middle line of
Figure \ref{f_pieces}. 
The tensor \eqref{defFusion} is called the \emph{capping} operator in \cite{OK} and 
denoted by $\mathsf{J}$ there. Here we use the same notation. It would
be nice to have a name for this operator which better reflects the
role in plays in the correspondence studied in the the present paper. 

\subsubsection{}

The correspondences
 \eqref{defGlue} and \eqref{defFusion} act on 
$K_{\bT\times\Ct_q}(X)$ and the statement
 of the degeneration formula may be written as follows 
 \begin{equation}
\Vertex = \Cp \, \Glue^{-1} \, \Fusion \label{CGF}
\end{equation}
where we compose the operators in the order in which we draw the
component of $C'_0$. {}From definitions, 
$$
\Glue = \cK_{X}^{1/2} + O(z)
$$
where $\cK_X$ is the canonical bundle of $X$ viewed as an operator 
of tensor multiplication, and 
so the inverse $\Glue^{-1}$ is well-defined as a formal series in
$z$. 

The discovery that the operator $\Glue^{-1}$ enters degeneration
formula was originally made by Givental in his study of K-theoretic
analogs of Gromov-Witten counts, see \cite{Giv,Lee}. The adaptation of this
idea to K-theory of quasimap moduli spaces is straightforward, see
e.g.\ \cite{OK} for the details. 

\subsection{Difference equations}\label{s_difference}

\subsubsection{} 
The geometric construction of the operator \eqref{defFusion} makes it
easy to show that it is a fundamental solution to a compatible system
of difference equations in both K\"ahler and equivariant variables,
see Section 8 of \cite{OK}. 

Here by a fundamental solution we mean an operator that conjugates a
difference connection to a constant coefficient difference connection
or to some other standard form. Concretely, for $q$-shifts of
equivariant variables discussed in Section 8.2 of \cite{OK}, that
standard form is a difference equation solvable in
$\varphi$-functions. This is the origin of $\varphi$-prefactor in
\eqref{tVx}. 

\subsubsection{}

An algebraic identification of these $q$-difference equations requires
a development of geometric representation theory ideas in the
present setting. In includes an identification 
\begin{align}
K_\bT(X) = \quad &\textup{weight subspace in}  \notag\\
 & \textup{a representation $F$ of $\cU_\hbar(\fgh)$\,,}
  \label{repr}
  \end{align}
for a certain quantum group $\cU_\hbar(\fgh)$. Such geometric 
realizations of quantum groups 
go back to the
pioneering work of Nakajima \cite{Nakajima} and have been studied by many
researchers since. The particular point of view on \eqref{repr}
developed in \cite{MO} and 
further in \cite{OS,ese,OK} will be important in what follows. It gives,
among other things, a natural collection of identifications
\begin{equation}
F\otimes \Q(\bT) \cong \bigotimes_{a\in I} \bigotimes_{\alpha=1}^{\dim W_a} 
F_a (a_{a,\alpha}) \otimes \Q(\bT) \label{tensF}
\end{equation}
indexed by all possible orderings of the coordinates 
of the maximal torus 
$$
\bA=\{\textup{diag}( a_{a,\alpha}) \} \subset \prod_a GL(W_a) 
\subset \Aut(X,\omega) \,. 
$$
In \eqref{tensF}, we have 
\begin{align*}
 a = \quad &\textup{an element of the set $I$ of vertices of the quiver} \,, \\
 F_a  = \quad&\textup{the corresponding fundamental representation of 
$\cU_\hbar(\fgh)$} \\
&\textup{a.k.a.\ Kirillov-Reshetikhin module\,,} \\
a_{a,\alpha} = \quad &\textup{equivariant parameter for $GL(W_a)$ and} \\
&\textup{an evaluation parameter for $F_a$} \,. 
\end{align*}
The identification \eqref{repr} is in integral K-theory, and so a
certain integral form of both the quantum group and of its module
appears in the right-hand side. In \eqref{tensF} we tensor with the
field $\Q(\bT)$ of rational functions of $\bT$, which corresponds to
localization in $\bT$-equivariant K-theory. Correspondingly,
$R$-matrices that intertwine the identifications \eqref{tensF} for
different ordering of the evaluation points act in localized
K-theory. 

Geometrically, it is the tensor structure, that is, the maps
\eqref{tensF} that are constructed first in the approach of \cite{MO}. 
They are a particular instance of certain very special maps of 
the form
$$
K_\bT(X^\bA) \to K_\bT(X) 
$$
called \emph{stable envelopes}, see e.g.\ Section 9 of \cite{OK} for
an introduction. 
The structure of a module over a quantum group is then reconstructed
from this tensor structure. 

\subsubsection{}
On the right-hand side of \eqref{tensF} we have a canonical 
$q$-difference connection in the evaluation parameters $a_{a,\alpha}$,
namely the quantum Knizhnik-Zamolodchikov connection of 
I.~Frenkel and N.~Reshetikhin \cite{FIR}. It takes as a parameter an 
element 
$$
z \in \left(\Ct\right)^I = e^{\fh}\,, \quad \fh \subset \fg \subset \fgh
$$
of the torus of group-like elements of $\cU_\hbar(\fgh)$. 
It corresponds to the Cartan subalgebra $\fh\subset \fg$ of the 
Lie algebra $\fg$, the affinization of which is $\fgh$. This torus 
is naturally identified with the K\"ahler torus from before. 

A technical result of \cite{OK} identifies the geometric 
$q$-difference connection in variables 
$\{a_{a,\alpha} \}$ with the qKZ connection. See Section 10 in \cite{OK} and
also \cite{MO} for a proof in the setting of equivariant
cohomology. Thus 
\begin{equation}
  \label{FusionqKZ_}
 \Fusion = \textup{fundamental solution of qKZ} \,,
\end{equation}
with the following very important detail that needs to be mentioned. 

\subsubsection{}
The difference connection in $a$ and $z$ 
solved by \eqref{defFusion} is by
construction flat. Moreover, it is regular in either $a$ or $z$
separately. However, it is \emph{not} jointly regular in the variables
$a$ and $z$. This simple, but important new 
phenomenon for difference
equations is discussed at length in \cite{ese}.  It does not occur for differential equations by a deep
theorem of Deligne. As a
result, one cannot find a fundamental solution which will be
holomorphic in both $z$ and $a$ in some asymptotic region of the torus
of variables. 

Recall that one is usually looks for solutions of qKZ analytic in an
asymptotic 
region of the schematic form 
$$
|a_{2,5}| \gg |a_{1,7}| \gg |a_{3,3}| \gg \dots \,,
$$
that is, in a certain neighborhood of a fixed point in a toric
compactification of $\bA$. We will call such solutions $a$-solutions. 
By the results of \cite{FIR}, $q$-deformed WZW conformal blocks are
$a$-solutions of the qKZ equations. 

Instead, \eqref{defFusion} is a series in $z$, which means it is 
holomorphic is a neighborhood of a torus-fixed point of the 
K\"ahler moduli space. This K\"ahler moduli space is the toric variety constructed from the 
fan of ample cones of flops of $X$ in $\Pic(X)$. We call such
solutions $z$-solutions. A more precise version of \eqref{FusionqKZ_}
is thus the following 

\begin{Theorem}[\cite{OK}] The operator \eqref{defFusion} 
is the fundamental $z$-solution of qKZ equations in 
variables $\{a_{a,\alpha}\}$\,. 
\end{Theorem}

\subsubsection{}
Meromorphic 
solutions to a $q$-difference equation form a vector space over 
$q$-periodic meromorphic 
functions of dimension equal to the rank. Therefore,
there exists a uniquely defined matrix transforming $z$-solution to
$a$-solutions. Taking into account the constant coefficient
$q$-difference equations to which fundamental solutions conjugate 
the original equation, this matrix is best seen as a
meromorphic section of a certain vector bundle on the elliptic curve
$$
E= \Ct / q^\Z \,. 
$$ 
It is called the \emph{pole subtraction} matrix in \cite{ese}, as it
quite literally removes the poles in one set of variables at the
expense of poles in another set of variables. 

This matrix is linked to elliptic analog of stable envelopes in
\cite{ese}. Concretely, Theorem 4 in \cite{ese} shows that elliptic 
stable envelopes transform $z$-solutions of the equations satisfied 
by the vertex functions to the corresponding $a$-solutions. 

\subsubsection{}
Difference equations satisfied by the vertex functions follow from the 
following qualitative 

\begin{Proposition}\cite{OS,Sm_cap} 
The cap \eqref{defCap} and the glue operator \eqref{defGlue} are 
rational functions of all variables, including the K\"ahler
variables. 
\end{Proposition}

The statement about the glue operator follows from the 
results of \cite{OS} because  the glue operator
may be obtained as a $q\to\infty$ limit of operators of the 
K\"ahler $q$-difference connection,  see
 Section 8.1 of \cite{OK}. The statement about cap 
is shown in \cite{Sm_cap}. In both cases, there is 
an explicit formula for these objects that makes 
rationality manifest.

Thus \eqref{CGF} gives an explicit gauge equivalence between the
scalar difference equation of degree $\textup{rk} K(X)$ satisfied by the
vertex functions and the quantum Knizhnik-Zamolodchikov equations. 

% \begin{proof}
% {}From the results in Section 8.1 of \cite{OK}, the glue operator
% may be obtained as a $q\to\infty$ limit of operators of the 
% K\"ahler $q$-difference connection. These operators were computed in 
% \cite{OS} and the result is a rational function of all parameters,
% whence the conclusion for \eqref{defGlue}. 

% The large framing vanishing discussed in Section 7.5 of \cite{OK}
% shows that the cap is trivial, that is 
% $$
% \Cp = \cK_{X}^{1/2} \,, 
% $$
% with no $z$-dependent contributions, as soon as the vector 
% $(\dim W_i)$ is
% sufficiently large. The statement for all possible framings is then 
% a formal consequence of the relationship between the cap
% for $X$ and the cap for $X^\bA$, where $\bA$ is a torus acting 
% on the framing spaces. Along these lines, an
%  explicit formula for the cap and, 
% more generally, for the cap with descendents was obtained by 
% A.~Smirnov \cite{Sm_cap}. 
% \end{proof}

The results of \cite{OS} identify the operators of the 
 K\"ahler $q$-difference connection with the lattice in 
what can be called the dynamical quantum affine Weyl 
group of $\cU_\hbar(\fgh)$. It coincides with the object
studied by Etingof and Varchenko in \cite{EV} for quivers 
of finite type and generalizes it to the case when $\fgh$ is not
generated by real root subspaces. {}From this perspective, 
the glue operator generalizes the longest element in the 
finite quantum dynamical Weyl group. 

\subsubsection{}
The cap with descendents mentioned above refers to the 
generalization of \eqref{defCap} constructed as follows 
\begin{equation}
  \Cp(\lambda) \,
  = \ev_{\bullet,*} \left(
\tO_\vir \,  z^{\deg f} \otimes \lambda(\cV_i|_0) \right)
\label{defCd}
\end{equation}
where 
$$
\lambda \in K_{\bT\times GL(V)} (\pt)  
$$
is a tensor functor in the fibers of the tautological bundles 
$\cV_i$ over a $\Ct_q$-fixed point $0$ in the domain of the
quasimap. We can identify $K_{\bT\times GL(V)} (\pt) = K_\bT(\fR)$,
where $\fR$ is the stack of quiver representations that contains 
$X$ as the set of the stable points satisfying the moment map 
equations. Using the surjectivity of \cite{mn}, 
we get 
\begin{align}
K_{\bT\times \Ct_q} (\fR)[[z]]
&\xrightarrow{\,\, \Cp(\,\cdot\,) \,\,} K_{\bT\times \Ct_q}(X)[[z]] \to
0 \label{capdesc} \,, \\
\lambda & \mapsto \lambda\big|_X + O(z)  \quad  \notag
  \end{align}
and Smirnov gives an explicit rational function formula for this map \cite{Sm_cap}. 

The degeneration formula \eqref{CGF} remains unmodified, 
giving 
 \begin{equation}
\Vertex(\lambda) = \Cp(\lambda) \, \Glue^{-1} \, \Fusion 
\,. \label{CGFl}
\end{equation}
 In particular, one can choose the descendent insertions so that they 
precisely cancel the glue matrix in \eqref{CGFl}, and this shows 
\begin{equation}
\Fusion \subset \{\textup{Vertices with descendents} \} 
\label{Fusion_desc} \,. 
\end{equation}

\subsubsection{}\label{s_vert_desc} 

Let $\bVx(\lambda)$ denote the vertex with descendents 
normalized as in \eqref{bVx}. The descendents are expressed 
in terms of the Chern roots of the bundles $\cV_i \big|_0$ which 
are precisely the integration variables in \eqref{VI4}. 
Therefore, we have the following immediate generalization of 
Proposition \ref{p_bVx_int}. 

\begin{Proposition}\label{p_vert_desc} 
For any insertion $\cF$ satisfying \eqref{autcF}, we have 
\begin{equation}
  \label{VI5}
  \chi(X, \bVx(\lambda) \otimes \cF) = 
\frac{1}{|W|} 
\int_{\sum q^{\bd} \cdot \gamma_\chi}  \exp\left(\frac{\bla(z_\#,s)}{\ln q}\right)
\frac{\cF(s) \, \lambda(s) \, d_\textup{Haar}s}{\phi(T^\vee_\textup{moving})} 
\,,
\end{equation}
where the sum of residues is over all effective shifts of the cycle
$\gamma_\chi$. 
\end{Proposition}

Smirnov's formula lets one construct collections $\{\lambda_k\}$ 
such that 
the matrix of the corresponding descendent vertices is 
the fundamental $z$-solution of qKZ. Theorem 4 of \cite{ese} applies
equally well to
both ordinary and descendent vertices, therefore, elliptic stable
envelopes provide a connection matrix between this fundamental 
solution and the fundamental $a$-solutions. In particular, for 
quivers of finite type, these $a$-solutions are the $q$-deformed
WZW conformal blocks.

This can be summarized as follows
\begin{Theorem}
There exist a linear map 
\begin{equation}
K_T(X) \owns \alpha \mapsto \lambda_\alpha \in K_T(\fR) \otimes
\Q(z,q)\label{KtoK}
\end{equation}
such that 
$$
\lambda_\alpha|_{X,z=0} = \alpha 
$$
and such that 
the corresponding vertex functions \eqref{VI5} form a fundamental 
$z$-solution of qKZ. With the insertions of the elliptic stable
envelopes, these become the fundamental $a$-solutions of qKZ, that is, 
a basis of the q-conformal blocks for $\cU_\hbar(\fgh)$. The entry 
corresponding to the identity function $\lambda=1$ is the
corresponding ${\mc W}$-algebra q-conformal block. 
\end{Theorem}

A remarkably simple formula for an equivalent version of \eqref{KtoK} 
is obtained in \cite{Bethe}.

\section{${\fg}=A_1$ Example}

To illustrate the results, it may be helpful to work out one example in its completeness. Take
$\fg =sl_2$ with finite dimensional representations $\rho_i$ of highest weights $w_i$ attached to points
$$
x=a_i, \qquad i=1, \ldots n
$$
of the Riemann surface. 

\subsection{qKZ equation and its $z$-solutions}
The $q$-conformal block of $U_{\hbar}({\fgh})$ with this data is a chiral correlation function from \eqref{WZ}
\beq\label{WZa}
\Psi(a_1, \ldots , a_\ell, \ldots a_n) \in (\otimes_i  \, \rho_i)_{\lambda - \lambda'},
\eeq
and where $\lambda_{0,\infty}$ are the weights $|\lambda_{0,\infty}\rangle$, the highest
weight vectors of Verma module representations which enter \eqref{WZ}.

As \cite{FIR} explained, \eqref{WZa} solves the qKZ equation of \eqref{qKZ} 
where the  ${\cal R}$ matrices take the following explicit form.
\subsubsection{}
Let $v_i$ be the highest weight vector of representation $\rho_i$ (with weight $w_i$). Let $f$ be the lowering operator of $\fg =sl_2$. The ${\cal R}$ matrix acts by
\beq\label{R1}
{\cal R}_{ij}(a) v_i \otimes v_j = v_i \otimes v_j
\eeq
\beq\label{R2}
{\cal R}_{ij}(a) f v_i \otimes v_j =  { a \hbar^{m_j} - \hbar^{m_i} \over a -  \hbar^{m_i+m_j}} f v_i \otimes v_j +  { 1- \hbar^{2m_j}  \over a - \hbar^{m_i+m_j}}   v_i \otimes  f v_j
\eeq
\beq\label{R3}
{\cal R}_{ij}(a)  v_i \otimes f v_j =  {a( 1- \hbar^{2m_i} )\over a -  \hbar^{m_i+m_j}} f v_i \otimes v_j +  { a \hbar^{m_i} - \hbar^{m_j} \over a -\hbar^{m_i+m_j}}  v_i \otimes  f v_j
\eeq
where 
$$
m_i =(w_i, e)/2,
$$
and $e$ is the positive root of $sl_2$. Throughout, one should keep in mind the identifications in \eqref{identify}.
Furthermore, $(\hbar^{\mu})_{\ell}$ acts on $\ell$'th component of the tensor, corresponding to representation $\rho_{\ell}$ of highest weight vector $v_{\ell}$ of highest weight $w_{\ell}$ by
$$
\hbar^{\mu}(v_{\ell}) =\hbar^{(\mu, w_{\ell})} v_{\ell}, \qquad \hbar^{\mu}(f v_{\ell}) = \hbar^{(\mu, w_{\ell} - e)}.
$$
The Weyl vector $\rho$, which also enters \eqref{qKZ}, is equal to half the sum of positive roots, $\rho = e/2$ in this case.

\subsubsection{}

The solutions to the qKZ equation, in $n$-dimensional the subspace of weight 
$$
\lambda' - \lambda = w_1+ \ldots +w_n -e,
$$
can be written out explicitly, as follows \cite{Matsuo}. Let
\beq\label{solution}
{\Psi} (a_1, \ldots, a_n) =\,  \sum_{i=1}^n \; {\varphi}_i(a_1, \ldots a_n) \; v_1 \otimes \ldots \otimes f v_i \otimes \ldots \otimes v_n.
\eeq
Further, it is useful to define
\beq\label{solutionctd}
\varphi_i (a_1, \ldots a_n)= q^{(\beta_{i+1} + \ldots + \beta_n)/2} \, a_1^{\beta_1} \ldots a_n^{\beta_n} \; \scF_i( q^{\beta_1/2} a_1, \ldots , q^{\beta_n/2} a_n)
\eeq
where
$$
q^{\beta_i} = \hbar^{(w_i, e)}, \qquad q^{\eta} =\hbar^{-(\lambda ,e)}.
$$
Then, \cite{Matsuo} proves \eqref{solution} is the solution of the qKZ equation for

\beq\label{fusion}
\scF_i(a)  = \int_{\gamma} {dx}  \; x^{\eta -1} \; K_i(x,a)\, \times \prod_{j=1}^n  {\varphi(x/a_j) \over \varphi(q^{\beta_j} \, x /a_j)},
\eeq
where we defined  
$$K_i(x, a) = \prod_{j=1}^{i-1} {(1-q^{\beta_j}\, x/a_j) \over (1-x/a_j)} \times {1\over 1 - x/a_i}.
$$
and $\varphi(x) = \prod_{n=0}^{\infty}(1-q^n x)$, as before.
The equation \eqref{solution} gives solutions to qKZ for any contour $\gamma$ for which the integral\footnote{This implies that, from perspective of difference equations, the integrand $I(x)$ in \eqref{fusion} is equivalent to $I(qx)$: $\int {dx\over x} I(x) = \int{dx \over x}I(qx)$.}  is invariant under $x\rightarrow q x$. One proves this by explicitly studying difference equations satisfied by \eqref{fusion} with respect to operators that take $a_i \rightarrow q a_i$.

\subsubsection{}
Our main example corresponds to $\rho_i$ which is the two-dimensional representation $\rho_i=\rho$, for all $i$. Its highest weight is the fundamental weight $w_i=w$ of ${\fg}=sl_2$, $(w,e)=1$, so that $q^{\beta_i} = \hbar$. Up to redefinition of integration variable $x$, replacing it with $\hbar x$, we have

\beq\label{fusionn}
\scF_i(a)  = \int_{\gamma} {dx}  \; x^{\eta -1}  K_i(x,a) \times \prod_{j=1}^n {\varphi(\hbar^{-1} {x / a_j}) \over \varphi( {x /a_j})},
\eeq
where  we defined  
$$K_i(x, a) = \prod_{j=1}^{i-1} {(1-\, {x/ a_j} )\over (1-   \hbar^{-1}{x/ a_j})} \times {1\over 1 - \hbar^{-1} { x/ a_i}}.
$$
The equations \eqref{solutionctd} and \eqref{fusionn} provide a solution to the qKZ equation, for any choice of the contour $C$. 
The set of linearly independent solutions one gets, by varying the contour $C$ have a geometric and representation theoretic interpretation.

\subsection{Geometric interpretation in terms of $X=T^*{\mathbb P}^{n-1}$}

The geometric interpretation is in terms of counts of quasi-maps to
\beq\label{NE}
X= T^*{\mathbb P}^{n-1},
\eeq
where $z$ keeps track of the degree of the map. $X$ is the Nakajima quiver variety \eqref{NV} corresponding to an ${\fg} =A_1$ quiver ${\scQ}$ a single node and a pair of vector spaces 
$V = {\mathbb C}$ and  $W = {\mathbb C}^n$ associated to it, acted on by
 $$
G_{\scQ} = GL(1) , \qquad G_{W} = GL(n).
$$
The dimension vectors of $W$ and $V$ are determined, respectively, by the highest weight of the module
$$\otimes_i \,\rho_i = \otimes^n \rho,
$$
and the weight of its subspace in which \eqref{WZa} takes values, as explained in \ref{s_Nakajima}.

\subsubsection{}

The vertex function of $X$, counting quasi-maps from $C$ to $X$, has an integral representation \eqref{PnInt},  as one recalls from Section \ref{s_vertex}:

\beq\label{Vertexn}
\Vx= \int_{\gamma} {dx}  \; x^{\eta -1}   \Phi(x,a) 
\eeq
where, in terms of $t=q /\hbar,$

$$
\Phi(x,a)=  \prod_{j=1}^n {\varphi(t \,x/a_j) \over \varphi( x /a_j)}.
$$

Using that, it is easy to recognize that the solutions to the qKZ equation in \eqref{fusionn}
can be rewritten in terms of the geometric quantities of $X$:

\beq\label{fusionab}
{\scF}_i  = \int_{{\gamma}} {dx}  \; x^{\eta-1}  {\Stab}^K_{i}(x,a)  \Phi(x,a).
\eeq
where 

\beq\label{stzero}
{\Stab}_i (x,a)= \prod_{j=1}^{i-1} {(1-\, x/a_j) }  \times \prod_{j=i+1}^n {(1 - \hbar^{-1} x/a_i)},
\eeq
is a collection of classes in $K_T(X)$. 
The integrands in \eqref{fusionab} and \eqref{fusionn} are equal.

This shows that ${\cal F}$, 
the fundamental $z$-solution to the qKZ equation in \eqref{fusionn}, is the geometrically defined operator $\Fusion$ in  \eqref{FusionqKZ_},
$$
 {\cal F} = {\Fusion}.
$$ 
and that the geometric corresponding of \eqref{fusionab} is in terms of vertex functions, counting quasimaps $\C  \dasharrow X$, with descendant insertions at $0\in X$ from  \eqref{Fusion_desc}.
The basis of insertions that leads to the qKZ equation with ${\cal R}$ matrices in the standard form is a special one, as will be explained in \cite{Bethe}. The classes in \eqref{stzero} give the 
K-theoretic stable basis of $X$, defined in \cite{OK}.  For a suitable choice of a chamber, slope and polarization \cite{OK}, \eqref{stzero} gives the basis element corresponding to stable envelope of the $i$-th ${\rm T}$-fixed point in $X$.

\subsubsection{}
We can also consider the stable envelope with slope ${s}$  \cite{OK} 
\beq\label{slopes}
{\Stab}^{K}_{i, (s)}(x,a) \equiv x^{s} \;  {\Stab}^K_i(x,a).
\eeq
The role of the slope $s$ is to change the weight $\lambda$ in \eqref{WZ}, and
leads to a family of solutions to qKZ, differing by the choice of the highest weight vector $|\lambda\rangle$ in \eqref{WZ}.

%%%%%%%%%

\subsection{$q$-Virasoro conformal blocks}
 The vertex function in \eqref{Vertexn}  as we saw in Section  \ref{s_vertex}, 
 coincides with a $q$-conformal block of the ${\cal W}_{q,t}(\fg)$ algebra for ${\fg}=A_1$; the algebra which is the $q$ deformation of Virasoro algebra. Its $q$-conformal blocks are %
\beq\label{oneea}
\langle \mu'| V^{\vee}(a_{1})\ldots V^{\vee}(a_{n}) \; Q^{\vee}|\mu\rangle = 
\int_{\gamma} dx  \; x^{-(\mu, e) -1}   \Phi(x,a).
\eeq
%where 
%
%$$\Phi(x,a)=  \prod_{j=1}^n {\varphi(q x/\hbar a_j) \over \varphi( x /a_j)},
%$$
%
%$t=q^{\beta} = q/\hbar$, and
To completely define the $q$-conformal block  in \eqref{oneea} of we need to specify the contour $\gamma$. As in section \ref{s_matchw}, we will
define the ${\cal W}_{q,t}(\fg)$ algebra blocks to be the components of the vertex function of $X= T^*{\mathbb P}^{n-1}$,

$$
\Vx = \langle \mu'| V^{\vee}(a_{1})\ldots V^{\vee}(a_{n}) \; Q^{\vee}|\mu\rangle 
$$
where the Kahler variable $z$ equals $z=q^{(\mu,e)}$, up to unimportant shift.
The component of the vertex function $\Vx_{\ell}$ where the point at infinity in $C$ maps to the fixed point $p_{\ell}$ in $X$ corresponds to 
$$
\Vx_{\ell} = \chi(X, \Vx \otimes\, {\cal O}_{p_{\ell}})
$$
section \ref{s_fdef}. The insertion of $ {\cal O}_{p_{\ell}}$ amounts the to picking the contour $\gamma_{\ell}$ which picks up the poles at
\beq\label{contour}
\gamma_{\ell}: \qquad x=q^{-n}\; a_\ell , \qquad n=0,1,\ldots.
\eeq
Computing the integral by residues, we find

\begin{align}
  \label{vertex_hyper}
\Vx_{\ell}&= (a_{\ell})^{\eta}\, 
 \frac{\varphi(t)}{\varphi(q)}  \prod_{i\ne \ell} \frac{\varphi(t a_{\ell}/a_i)}{\varphi(a_{\ell}/a_i)} \,  \bF \left[ \left. 
\begin{matrix} \hbar a_1 / a_{\ell}, &  \hbar a_2/  a_{\ell}, &\dots \\
                      q a_1/ a_\ell, &q a_2/a_\ell, &\dots
\end{matrix} \,\, \right| \, z/\hbar^n
                                               \right] \,.
\end{align}
and 

\beq\label{vertexB}
 \bF \left[ \left. 
\begin{matrix}\, \,\hbar a_i/ a_{\ell} & \\
                     \, \,q a_i/a_{\ell} &
\end{matrix} \right| z/\hbar^n
                                               \right] =
\sum_{d\ge 0} (
z/\hbar^n)^d \, \prod_i \frac{(\hbar a_i/a_{\ell})_d}{(q
                  a_i/  a_{\ell})_d}   \notag \\ = \Vertex_{\ell}
\eeq
is the $q$-hypergeometric function. It is also the $\Vertex$ function of $X$, in its canonical normalization. The function ${\bf V}$ differs from it by contributions of constant, zero degree maps (see \eqref{tVx}).

%where $z=q^n z_A$ and $z_\#=q^{-\eta_\#} = q^{-(mu,e)}$.

\subsubsection{}
The $\Vx$ function (a vector) can be written as the covector $W$ contracted with the operator ${\scF} =\Fusion$.

\beq
\label{vertexcoe}
 \sum_{i=1}^n W_i \; {{\scF}^{i}}_{\ell} = {\Vx} _{\ell} 
\eeq
which is the content of \eqref{CGFl}, in the present example. 

The coefficients $W_i$ can be found as follows: The K-theoretic stable envelopes of fixed slope provide a basis of $K(X)$-theory of $X$, so in particular,  the trivial insertion $1$ at $0$ in $X$ written in the stable basis from \eqref{stzero} as:

\beq\label{vertexAC}
1 = \sum_{i=1}^n W_i \;{\Stab}^K_{i}(x,a),
\eeq
where $W_i$ are the coefficients in \eqref{vertexcoe}. The stable basis is upper triangular, as ${\Stab}^K_i(x,a)$ vanishes at $x=a_j/\hbar$, for $i<j$. This lets us find $W_i$ solving \eqref{vertexAC} recursively, solving for $W_i$ in terms of $W_{i+1}, \ldots , W_n$.

%%%%%%%

\subsection{Elliptic stable envelope and $z$- and $a$-solutions}
%The vertex functions ${\bf V}$ in \eqref{vertex_hyper} and \label{IntV} are, up to overall normalization, holomorphic in the Kahler variable $z$, in a neighborhood of $z=0$. 
${\scF}_i$'s generate a space of solutions of qKZ equation, by varying the contours $\gamma$. The solutions to qKZ obtained in \eqref{fusionn} or \eqref{fusionab} are not $q$-conformal blocks of $U_{\hbar}(\fgh)$ since they are not $a$-solutions of qKZ, which are solutions 
 ${\it jointly}$ analytic in a chamber of $A$-parameter space.\footnote{For example, $\Vx\big|_{\ell=1}$ is analytic for $|a_1|< |a_2|, |a_3|, \ldots$, but $\Vx\big|_{\ell=2}$ is analytic for $|a_2|<|a_1|, |a_3|, \ldots$.} Instead, they are  the $z$-solutions, analytic functions of the Kahler variable $z$. The map between the $z$ solutions and the $a$-solutions is provided by elliptic stable envelopes of $X$. 

\subsubsection{}
Pick an $a$-chamber  
$$
\fC : \qquad |a_j| < |a_i|, \qquad\textup{for} \qquad  j<i.
$$
Starting with the vertex function  ${\bf V}=\int  dx\; x^{\eta-1} \; {\Phi}(x) $, we obtain a new vertex function ${\bf V}_{\fC}$, which solves the same set of difference equations as ${\bf V}$ and which is analytic in chamber $\fC$, as follows.  We take

\begin{equation}\label{ThVxInt} 
{\bf V}_{\fC} = \oint  dx \, x^{\eta-1} \; {\Phi}(x) \; {\fP}_{\fC}(x)                                                                                                                                                                                                                                                                                                                                                                                                                                                                                                                                                                                                                                                                                                                                                                                                                                                                                                                                                                                                                                                                                                                                                                                                                                                                                                                                                                                                                                                                                                                                                                                                                                                                                                                                       
\end{equation}
with 

\beq\label{PSM}
{\fP}_{\fC,{\ell}}(x, a) = 
 U_{\fC, \ell}\;  { {\prod\limits_{ i<\ell}} \theta(a_{i}/x)\;\theta(z^{-1} \hbar^{\ell}\,a_\ell/x)\; \prod\limits_{\ell<i} \theta(\hbar \, a_{i}/x) \over \theta( z^{-1} \hbar^{\ell} ) \, \prod_i \theta(\hbar \, a_{i}/x) } {\be}(z,x)^{-1}.
\eeq
The contour of integration is spelled out below equation \eqref{pl2}.
${\bf V}$ and ${\bf V}_{\fC}$ solve the same set of difference equations, since  ${\fP}_{\fC, \ell}(x)$ are pseudo-constants satisfying  

$$\begin{aligned}\label{PS}
{\fP}_{\fC,{\ell}}(x, a_1, \ldots a_i, \ldots a_n; z) & = {\fP}_{\fC,{\ell}}(q x, a_1, \ldots a_i, \ldots a_n; z)\\
& = {\fP}_{\fC,{\ell}}(x, a_1, \ldots q a_i, \ldots a_n; z)\\
& ={\fP}_{\fC,{\ell}}(x, a_1, \ldots a_i, \ldots a_n; qz).
\end{aligned}
$$
The function in \eqref{PSM} is,  up to normalizations, the {\it elliptic stable envelope} of a fixed point $p_{\ell}$ in $X$ in the chamber $\fC$

$$
 {\Stab}^{\textup{\it ell}}_{\fC,{\ell}}(x, a) = 
 { {\prod\limits_{ i<\ell}} \theta(a_{i}/x)\;\theta(z^{-1} \hbar^{\ell}\,a_\ell/x)\; \prod\limits_{\ell<i} \theta(\hbar \, a_{i}/x) \over \theta( z^{-1} \hbar^{\ell} )}.
$$
defined geometrically in \cite{ese}. The normalizations involve 
\beq\label{uc2}
{\be}(z,x)^{-1} =\exp{\log(x) \log(z)\over \log(q)},
\eeq
and
\beq\label{uc1}
  U_{\fC, \ell}=\exp{\Bigl( \log(a_\ell) \log(z^{-1} \hbar^{\ell}) -\sum_{i\leq \ell} \log(a_i) \log( \hbar)  \Bigr)\over \log(q)},
\eeq
which help ensure that ${\fP}_{\fC, \ell}(x)$ satisfy \eqref{PS}. 

\subsubsection{}
The contour of integration in \eqref{ThVxInt}  is defined to separate the poles of the integrand, located at
\begin{alignat}{3} 
  x&=q^{-n} \, a_i \qquad && \ell \le i \,, \quad
  &&n=0,1,\dots \label{pl1} \\
 x&= q^{n} \, a_i\,\hbar  \,  , \qquad && i\le \ell \,, \quad
 &&n=0,1,\dots\,. \label{pl2}
\end{alignat}
For $|q|<1$, the poles in \eqref{pl1} accumulate to $x=\infty$ while the poles in
\eqref{pl2} accumulate to $x=0$.

For $z<1$ we can deform the contour to enclose all poles of the form \eqref{pl1}, to obtain
\beq\label{ld}
{\bf V}_{\fC, \ell}=  \sum_{\ell'}  {\bf V}_{\ell'}\; {\fP}_{\fC, \ell}^{\ell'}\;
\eeq
where 
$${\fP}_{\fC, \ell}^{\ell'} = {\fP}_{\fC, \ell}(a_{ \ell'}).$$
The linear change of basis in \eqref{ld}, determined by elliptic stable envelopes, is the {\it pole subtraction matrix} for chamber ${\fC}$. The name reflects the fact that ${\bf V}_{\fC}$ is pole-free in a neighborhood of $0_\fC$, the origin of the chamber ${\fC}$. 
Note the pole subtraction matrix is triangular,
$$
{\fP}_{\fC, \ell}^{\ell'}  =0, \qquad \ell'<\ell
$$
since the numerators in \eqref{PS} eliminate poles from \eqref{pl1} for $i=\ell'<\ell$. 
 
\subsubsection{}

A contour integral becomes singular when the poles of the integrand, located on opposite sides 
of the contour, coalesce. By studying poles of the integrand in \eqref{ThVxInt}, it follows that
${\bf V}_{\fC}$ is  pole-free in a neighborhood of $0_\fC$. This motivates the name we gave to the matrix ${\fP}_{\fC}$ in \eqref{ld}.

The contour of the integration
in \eqref{ThVxInt} per definition separates  the poles in \eqref{pl1} accumulate to $x=\infty$ while the poles in
\eqref{pl2} accumulate to $x=0$. 
This means that the contour integral in  \eqref{ThVxInt} has singularities at 
$$\frac{a_i}{a_j} =  q^{n}\hbar \,,
$$
with $j \le \ell \le i$ and $n \ge 0$. This is the complement of the chamber $\fC$ in which $|a_j/a_i|<1$ for $j<i$.
 %By our conventions, the chamber ${\fC}$ means that $x_i/x_j <1$ for $j<i$. 

\subsubsection{}
More generally, replacing ${\scF}_i$ in \eqref{ThVxInt} with
 
\beq\label{genF}
( {\scF}^{\fC})_{i\ell} = \oint dx\; x^{\eta-1}\;  {\Stab}^K_{i}(x) \; \Phi(x)\; {\fP}_{\fC,{\ell}}(x) 
\eeq
for each fixed $\ell$, we get a solution of qKZ of the form \eqref{solution}\eqref{solutionctd} which is analytic in chamber $\fC$.  This is a $q$-conformal block of $U_{\hbar}(\Lfgh)$.
Both the $K$-theoretic and the elliptic stable envelopes enter \eqref{genF}, but their roles are different. The K-theoretic stable envelope produces vector-valued solutions of $U_{\hbar}(\fgh)$ qKZ from scalar $q$-conformal blocks of ${\cal W}_{q,t}(\fg)$ algebra. To get analytic solutions of qKZ in chamber ${\fC}$ requires knowing the elliptic stable envelope, which enters the definition of ${\fP}_{\fC,{\ell}}(x, a)$. 

\subsection{ $X=T^*{\mathbb P}^1$ example}
Let's make this fully explicit for $n=2$, when $X=T^*{\mathbb P}^1$. The vertex functions associated to the two fixed points in $X$, corresponding to the north and the south poles of the ${\mathbb P}^1$ are:

$$\begin{aligned}
{\bf V}_{1} &=  \; a_1^{\eta}\; {\varphi(t) \over \varphi(q)}  {\varphi(t a_1/ a_2)\over\varphi({a_1/a_2})} 
\bF \left[ \left. 
\begin{matrix} { \hbar} &  { \hbar} {a_2\over a_1} \\
                      q & q {a_2\over a_1} 
\end{matrix} \,\, \right| \,t { z}'
                                               \right] \\
                                               & = a_1^{\eta_\#} {\theta(t a_1/a_2)\over\theta({a_1/a_2})} \; {\varphi(t) \over \varphi(q)} {\varphi({\hbar} z')\over\varphi({z'})}   \bF \left[ \left. 
\begin{matrix} { t}  &  {t}{ z'} \\
                      q &{q  z' }
\end{matrix} \,\, \right| \, \hbar  {a_2  \over a_1}
                                               \right]
                                               \end{aligned}
                                               $$
$$\begin{aligned} {\bf V}_{2} &=  \; a_2^{\eta}\; {\varphi(t) \over \varphi(q)}  {\varphi(t a_2/ a_1)\over\varphi({a_2/a_1})} 
\bF \left[ \left. 
\begin{matrix} {\hbar} &  { \hbar} {a_1\over a_2} \\
                      q & q {a_1\over a_2} 
\end{matrix} \,\, \right| \, t{z}'
                                               \right]\\
                                               & =a_2^{\eta_\#} {\theta(ta_2/ a_1)\over\theta({a_2/a_1})} \; {\varphi(t) \over \varphi(q)} {\varphi({\hbar} z')\over\varphi({z'})}   \bF \left[ \left. 
\begin{matrix} {\hbar} &  { \hbar}{ z'} \\
                      q &{q  z' }
\end{matrix} \,\, \right| \, \hbar  { a_2  \over  a_1}
                                               \right] 
                                                  \end{aligned}
                                               $$
where we defined $z' = t z$. The right hand side of the equations follows using standard identities for $q$-hypergeometric functions. Clearly the vertex function ${\bf V} = ({\bf V}_1, {\bf V}_2)$ has no nice analyticity properties as functions of $a$'s, but they are analytic for $|z|<1$.

\subsubsection{}
The elliptic stable envelopes provide a change of basis to solutions which are (quasi)-analytic in $a$'s. In the chamber $\fC$ where $|a_1|<|a_2|$, we have 
\beq\label{ch1}
{\fP}_{\fC}=  U_{\fC} \left( \begin{array}{cc}
{1\over \theta(\hbar )} & 0\\
{\theta( { \hbar a_1/z a_2 }) \over  \theta( {\hbar a_1/a_2})  \theta( \hbar/ z)} & {\theta(a_1/a_2) \over  {\theta(\hbar a_1/a_2 ) \theta({\hbar})}}
 \end{array} \right) \be^{-1}
\eeq
where the  $\ell \ell'$ entry of the matrix in \eqref{ch1} corresponds to ${\fP}_{{\fC},\ell}^{\ell'}$ in \eqref{ld}. The matrices $\be$, $U_{\fC}$ are both diagonal, with eigenvalues, and can be read off from \eqref{uc1}, \eqref{uc2}.

Explicitly, using various q-hypergeometric function identities, we find:
$${\bf V}_{\fC, 1} = {1 \over \theta(t)}  \; a_1^{\eta_\#}\; {\varphi(t) \over \varphi(q)}  {\varphi({ \hbar} z')\over\varphi({z'})} 
\bF \left[ \left. 
\begin{matrix} t &  t/ z' \\
                      q &q/ z' 
\end{matrix} \,\, \right| \, \hbar {a_1  \over a_2}
                                               \right] \,.
$$
$${\bf V}_{\fC, 2} = {1 \over \theta(t) }  \; a_2^{\eta_\#}\; {\varphi(t) \over \varphi(q)}  {\varphi({ \hbar/  z'})\over\varphi({1/ z'})}
\bF \left[ \left. 
\begin{matrix} t &  {t z'} \\
                      q &{q  z' }
\end{matrix} \,\, \right| \, \hbar  {a_1  \over a_2}
                                               \right] \,.
$$
The vertex function in the stable basis ${\bf V}_{\fC} = ({\bf V}_{\fC,1}, {\bf V}_{\fC,2})$ is now clearly analytic in the chamber ${\fC}$, corresponding to $|a_1|<|a_2|$.
The map to conformal blocks can be read off from the elliptic stable envelope, and goes as follows:

\beq\label{idc1}\begin{aligned}
{\bf V}_{\fC, 1} &\qquad \longrightarrow \qquad { H}^{\rho_1}_{\lambda_0, \lambda_0 -w } \otimes {H}^{\rho_2}_{\lambda_0 -w , \lambda_0}\\
{\bf V}_{\fC, 2} &\qquad \longrightarrow \qquad {H}^{\rho_1}_{\lambda_0, \lambda_0 +w } \otimes {H}^{\rho_2}_{\lambda_0 +w , \lambda_0}
\end{aligned}
\eeq

\subsection{Conformal limit}\label{mycon}
In the conformal limit, $q\rightarrow 1$ limit, one gets the familiar expressions for the integral solutions of KZ equation, and the corresponding Virasoro conformal blocks. We will review these in detail in Sec. \ref{glc}, for now simply note that the Virasoro block, given by \eqref{oneea} still has the same form, with

\beq\label{oneeb}
\langle \mu'| V^{\vee}(a_{1})\ldots V^{\vee}(a_{n}) \; Q^{\vee}|\mu\rangle = 
\int_{\gamma} d x  \; x^{-(\mu, e) -1}   \Phi(x,a).
\eeq
with
\beq\label{Philb}
\Phi(x,a)   \qquad \longrightarrow\qquad  \prod_{j=1}^n   (1-x/a_i )^{-\beta}.
\eeq

Now consider the limit of the $z$-solutions of qKZ equation, as obtained from the $\Fusion$ operator of $X$, geometrically.
becomes
%$$
%{\Psi} (a_1, \ldots, a_n)=\,  \sum_{i=1}^n \; a_1^{\theta} \ldots a_n^{\theta} \; \scF_i( a_1, \ldots , a_n) \; v_1 \otimes \ldots \otimes f v_i \otimes \ldots \otimes v_n.
%
\beq\label{Philbc}
{\scF}_i  = \int_{{\gamma}} {dx}  \; x^{-(\mu,e)-1}\; {\Stab}^K_{i}(x,a) \; \Phi(x,a),
\eeq
which leads to the \cite{EFK} integral form of solutions of KZ equation, which we review in Sec. \ref{glc}. Namely, the limit of stable envelopes is
$$
 {\Stab}^K_{i}(x,a) \; \qquad \longrightarrow\qquad  \prod_{\stackrel{j\neq i}{ j=1}}^n {(1 - x/a_j)}.
$$
and taking \eqref{Philb} and \eqref{Philbc} together, we get
\beq\label{fusionCFTac}
\scF_i(a)  = \int_{\gamma} {d x}  \; x^{(\mu,e)-1} {1\over 1-x/a_i} \; \times \; \prod_{j=1}^n {(1 - x/a_j)}^{-\beta +1 }.
\eeq
Recalling  \eqref{zero1}, 
$$ \beta - 1 = \theta = 1/ ^L(k+h^{\vee}),
$$
we recognize in \eqref{Philbc} the integral form of solutions of the KZ equation in the weight $n-1$ subspace.

\subsubsection{}
Recall that ${\fP}_{\fC,{\ell}}(x, a)$ is a pseudo-constant with respect to $q$-shifts of all the variables, see \eqref{PS}. Thus, $q$ goes to $1$ it becomes a constant, depending only on 
$$
q' = e^{- 2\pi i{ \ln \hbar \over \ln q}} = e^{2\pi i\over ^L(k+h^{\vee})},
$$ 
but not on any continuous variables. It follows
% The limit is best taken however, not at the level of integral formulas (since the value of the constant depends on whether $|x/a_i|> 1$ of $|x/a_i|<1$, for any $i$), but at the level of their residues such as \eqref{ld}. One easily finds
$$
{{\fP}_{\fC,{\ell}}}^{\ell'} \qquad \rightarrow \qquad (q')^{\#_{\fC, \ell, \ell'}}
$$
where $\#_{\fC, \ell}$ is a number depending only on $\ell$ and $\fC$.

\section{Isomorphism of conformal blocks and the geometric Langlands
  Correspondence}    \label{glc}

In the previous sections, we have established an isomorphism of
$q$-deformed conformal blocks of the deformed ${\mc W}$-algebra
associated to a simple Lie algebra $\fg$ and the quantum affine
algebra associated to the Langlands dual Lie algebra $^L\neg\fg$. It is
natural to ask whether the appearance of dual Lie algebras here is in
some ways related to the geometric Langlands correspondence and its
one-parameter deformation known as the quantum geometric Langlands.

The way the Langlands dual Lie algebra manifests here is all the more
striking because deformed ${\mc W}$-algebras do not exhibit the
duality known to exist in the conformal limit $q \to 1$. Indeed,
recall that in the conformal case, we have an isomorphism between the
${\mc W}$-algebras associated to $\fg$ and $^L\neg\fg$, after a change
of the parameter \cite{FF} (it is recalled in Theorem \ref{Walg}
below). But after the $q$-deformation, no such isomorphism is
available. In other words, there is no longer an isomorphism between
the deformed ${\mc W}$-algebras associated to $\fg$ and $^L\neg\fg$
(unless of course $^L\neg\fg = \fg$). This brings the difference
between the above two algebras, ${\cal W}_{q,t}({\fg})$ and
$U_{\hbar}(\wh{^L\neg\fg})$, into a sharper focus.

We stress that it is quite common in representation theory that
introducing an additional parameter enables one can see a particular
phenomenon more clearly, and notice aspects of it that hitherto could
be more easily missed or ignored. Often, this results in revisiting
the original phenomenon (before the deformation) and adjusting one's
point of view.

For instance, consider the Harish-Chandra isomorphism $c: Z(\fg)
\overset{\sim}\longrightarrow \on{Fun}(\fh^*)^W$ between the center of
the universal enveloping algebra $U(\fg)$ of a simple Lie algebra
$\fg$ and the algebra of Weyl-invariant polynomial functions on the
dual space to the Cartan subalgebra $\fh$ of $\fg$.
%One way to
%construct $c(z), z \in Z(\fg)$ is to say that its value at $\la \in
%\fh^*$ is equal to the value of $z$ on the Verma module $M_\la$ with
%the highest weight $\la$. After that, one can easily check that $c(z)$
%is invariant under a shifted action of the generators of the Weyl
%group $W$, and hence $W$ itself. (The fact that this is an isomorphism
%requires an additional argument, but that that is not important to us
%at the moment.)
There is a strange aspect of this formula that is easy to ignore: $c$
maps $Z(\fg)$ to $\on{Fun}(\fh^*)^W$ rather than
$\on{Fun}(\fh)^W$. But actually this is significant: the
Harish-Chandra isomorphism already contains a germ of the Langlands
duality. The point is that we have a canonical isomorphism between
$\fh^*$ and the Cartan subalgebra ${}^L\neg\fh$ of $^L\neg\fg$. It is more
insightful to express this isomorphism as $Z(\fg) \simeq
\on{Fun}({}^L\neg\fh)^W$, but it is difficult to convince oneself that
this is how we should view it if we remain squarely within the
finite-dimensional context because $\fh$ and $^L\neg\fh$ are so close to
each other (they are canonically isomorphic up to an overall scalar).

However, this phenomenon becomes much more clear after affinization
(which we can think of as introducing an additional parameter into the
picture). Indeed, the affine analogue of the Harish-Chandra
isomorphism is the isomorphism of \cite{FF,Frenkel2} between the
center $Z_{\on{crit}}(\ghat)$ of the enveloping algebra of $\ghat$ at
the critical level and the classical ${\mc W}$-algebra ${\mc
  W}_\infty(^L\neg\fg)$, viewed as a subalgebra of the algebra
$\on{Fun}(\on{Conn}_{^L\neg\fh})$ of functions on the space of connections
on a certain $^L\neg H$-bundle on the punctured disc. As explained in
\cite{Frenkel2}, the algebra $\on{Fun}(\on{Conn}_{^L\neg\fh})$ can be
viewed as an affine analogue of $\on{Fun}({}^L\neg\fh)^W$. Furthermore,
${\mc W}_\infty(^L\neg\fg)$, and hence $Z_{\on{crit}}(\ghat)$, can be
described as the subalgebra of $\on{Fun}(\on{Conn}_{^L\neg\fh})$
consisting of elements invariant under the classical limits of the
screening operators (which can be viewed here as affine analogues of
the simple reflections from $W$). The essential point is that, unlike
in the finite-dimensional case, ${\mc W}_\infty(^L\neg\fg)$ and ${\mc
  W}_\infty(\fg)$ are no longer isomorphic to each other (as Poisson
algebras) if $^L\neg\fg \neq \fg$. Therefore the phenomenon of Langlands
duality becomes much more transparent, and in retrospect, it forces us
to look at the original Harish-Chandra isomorphism in a new light.

This is what we hope our results on the $q$-deformed conformal blocks
can bring us as well: a sharper manifestation of certain phenomena
that would be difficult to see or appreciate in the context of the
undeformed quantum geometric Langlands, at $q=1$ (the same way as the
appearance of $^L\neg\fh$ in the Harish-Chandra isomorphism would be
difficult to appreciate). Understanding such phenomena for $q \neq 1$
could then shine a new light on what was considered as well-known or
well-understood in quantum geometric Langlands.

While we do not claim that we fully understand it yet, we consider the
canonical isomorphism of $q$-deformed conformal blocks conjectured in
this paper (and proved in the simply-laced case) as a significant
phenomenon in the framework of a conjectural $q$-deformed quantum
Langlands correspondence. We believe that it deserves further
study. As far as we know, this is the first attempt to make a precise
statement about $q$-deformed quantum Langlands correspondence (even
though its existence had been anticipated, see e.g. the end of
\cite{Frenkel1}). We hope that more information will come to light in
the future that will enable one to formulate the $q$-deformed quantum
Langlands more precisely.

In this section, we give a brief overview of some aspects of the
geometric Langlands correspondence and its quantum deformation, and
then explain in what sense our isomorphism of $q$-deformed conformal
blocks could be seen as a manifestation of a $q$-deformation of the
quantum geometric Langlands correspondence.

\subsection{Overview}    \label{glc-overview}

As we mentioned in the Introduction, the geometric Langlands
correspondence is usually understood today as a conjectural
equivalence between certain categories of sheaves on two moduli stacks
related to a smooth projective algebraic curve ${\mc C}$ and a pair of
connected Langlands dual complex reductive Lie groups $G$ and $^L\neg
G$. One is the derived category of ${\mc D}$-modules on the moduli
stack $\on{Bun}_{^L\neg G}$ of $^L\neg G$-bundles on ${\cal C}$ and the other
is a certain modification of the derived category of ${\mc O}$-modules
on the moduli stack $\on{Loc}_G$ of flat $G$-bundles on ${\mc C}$ (see
\cite{AG} for a precise formulation; in the abelian case this
equivalence is a version of the Fourier--Mukai transform that has been
proved in \cite{Laumon,Rothstein}).\footnote{Note that our notation
  for $G$ and $^L\neg G$ is opposite to the standard one.}

In \cite{KW}, Kapustin and Witten have connected this equivalence to
the homological mirror symmetry of sigma models with the Hitchin
moduli spaces of $G$ and $^L\neg G$ as target manifolds and to the
$S$-duality of maximally supersymmetric 4d gauge theories with the
gauge groups being the compact forms of $G$ and $^L\neg G$.

This equivalence is expected to satisfy various properties; in
particular, the compatibility with certain functors acting on the two
categories: the Hecke functors on the $^L\neg G$ side and the ``Wilson
functors'' on the $G$ side (they are connected to the 't Hooft and
Wilson line operators of the 4d gauge theory \cite{KW}).

%In particular, consider the ${\mc D}$-module ${\mc F}_{\mc E}$ on
%$\on{Bun}_G$ assigned by this equivalence to the skyscraper ${\mc
%  O}$-module ${\mc O}_{\mc E}$ supported at a generic\footnote{Here
%  ``generic'' means that ${\mc E}$ has no automorphisms other than the
%  elements of the center of $^L\neg G$. Note that our notation for $G$ and
%  $^L\neg G$ is opposite to the standard one.} point of $\on{Loc}_G$
%corresponding to a flat $G$-bundle ${\mc E}$. Since ${\mc O}_{\mc E}$
%is clearly an eigen-object of the Wilson functors (that is to say,
%when we apply a Wilson functor to ${\mc O}_{\mc E}$, it gets tensored
%with a finite-dimensional vector space), it follows that ${\mc F}_{\mc
%  E}$ should be an eigen-object of the Hecke functors. In other words,
%${\mc F}_{\mc E}$ is a ``Hecke eigensheaf''. In fact, this is how the
%geometric Langlands correspondence was thought of initially: as a
%correspondence assigning a Hecke eigensheaf on $\on{Bun}_{^L\neg G}$ to a
%flat $G$-bundle on ${\mc C}$ (see, e.g., \cite{Frenkel,F:bourbaki} for
%a survey and further references). The categorical version has the
%advantage that it restores the symmetry between the two
%sides. Furthermore, in the abelian case this equivalence of categories
%turns out to be a variant of the Fourier--Mukai transform
%\cite{Laumon,Rothstein}, and so one can think of the categorical
%geometric Langlands correspondence in the non-abelian case as a kind
%of non-abelian Fourier transform. This is how it was envisioned by
%Beilinson and Drinfeld.

In \cite{BD}, Beilinson and Drinfeld constructed an important part of
the geometric Langlands correspondence in which on the $G$-side one
takes the subcategory of ${\mc O}$-modules supported on a substack of
$G$-opers in $\on{Loc}_G$. In the case that $G$ is a simple Lie group
of adjoint type (i.e., with the trivial center), to which we restrict
ourselves in this subsection, $\on{Op}_{G}$ is an affine space that is
isomorphic to the space of all flat connections on a specific
$G$-bundle on ${\mc C}$ \cite{BD}.

In their construction, Beilinson and Drinfeld used the description of
the center of the vertex algebra of $\wh{^L\neg\fg}$ at the critical level
given in \cite{FF}. Namely, it was proved in \cite{FF} (see also
\cite{Frenkel2} for a survey) that the center of the completed
enveloping algebra of $\wh{^L\neg\fg}$ at the critical level
$^L\neg k=-{}^L\neg h^\vee$ is isomorphic (as a Poisson algebra) to the
classical ${\mc W}$-algebra associated to $\fg$. The latter is, by
definition, the algebra of functions on the space of $G$-opers on the
punctured disc, and the Poisson structure on it is obtained via the
(classical) Drinfeld--Sokolov reduction. Equivalently, the center of
the vertex (or chiral) algebra of $\wh{^L\neg\fg}$ at the critical level
is isomorphic to the commutative (Poisson) vertex algebra ${\mc
  W}_\infty(\fg)$.

This isomorphism enabled Beilinson and Drinfeld to construct a family
of critically twisted ${\mc D}$-modules on $\on{Bun}_{^L\neg G}$
parametrized by those conformal blocks of ${\mc W}_\infty(\fg)$ on $X$
that are algebra homomorphisms. Furthermore,
%$\wh{^L\neg\fg}$-modules at the critical level parametrized by $G$-opers
%on the disc. To each of them we then associate a ${\mc D}$-module of
%conformal blocks on $\on{Bun}_{^L\neg G}$ using the standard
%``localization'' construction. (It goes back to the works of
%Beilinson--Bernstein and Brylinski--Kashiwara on the localization of
%$\fg$-modules on the flag variety; see, e.g., Section 9 of
%\cite{Frenkel}.) The resulting ${\mc D}$-module is non-zero if and
%only if the $G$-oper extends from the disc to the entire curve ${\mc
%  C}$, and Beilinson and Drinfeld have proved \cite{BD} that if that
%is the case, then this ${\mc D}$-module is a Hecke eigensheaf with
%respect to this oper.
the Beilinson--Drinfeld construction can be placed in the framework of
2d CFT, even though the Kac--Moody chiral algebra at the critical
level is quite unusual (it is missing the stress tensor $T(z)$ because
the quadratic Sugawara current becomes central in this case). See
Section 9 of \cite{Frenkel} for more details.

%The sheaves of conformal blocks of this chiral algebra can still be
%defined, and they give rise to a non-trivial ${\mc D}$-module on
%$\on{Bun}_{^L\neg G}$ (rather than the moduli spaces of punctured curves),
%which is where Hecke eigensheaves are supposed to ``live''. More
%precisely, the fibers of this ${\mc D}$-module are the (duals of)
%spaces of conformal blocks of representations of $\Lfgh$ of the
%critical level, as explained in Section 9 of
%\cite{Frenkel}. Furthermore, these conformal blocks turn out to be
%closely related to the conformal blocks of the (commutative) classical
%${\mc W}$-algebra associated to $\fg$ (see Section 9.3 of
%\cite{Frenkel}). Thus, we see that conformal blocks of both $\Lfgh$ of
%the critical level and the classical ${\mc W}$-algebra associated to
%$\fg$ are inextricably linked to the geometric Langlands
%correspondence.

\subsection{Global quantum Langlands
  correspondence}    \label{glc-def}

As soon as it became clear that there is a link between the
Beilinson--Drinfeld construction of the geometric Langlands
correspondence and 2d CFT at the critical level, a natural question
arose: is it possible to to deform the geometric Langlands
correspondence away from the critical level? The first conjectural
formulation was proposed by Beilinson and Drinfeld themselves (see
\cite{FFS}): the global quantum geometric Langlands correspondence
should be an equivalence of suitably modified derived categories of
twisted ${\mc D}$-modules on $\on{Bun}_G$ and $\on{Bun}_{^L\neg G}$,
provided that the corresponding twist parameters, which can be
identified with the levels $^L\neg k$ and $k$, satisfy the relation
\eqref{kdual} below. There is a precise sense in which the $k$-twisted
${\mc D}$-modules can be identified with ${\mc O}$-modules on
$\on{Loc}_G$ in the limit $k \to \infty$ (see, e.g., Section 6.3 of
\cite{Frenkel}), and it is in this sense that the $k \to \infty$ limit
of this equivalence is expected to yield the categorical Langlands
correspondence of the previous subsection (that is, for $^L\neg
k=-{}^L\neg h^\vee$). A closely related equivalence (of certain
categories of A-branes) was also suggested in the framework of the 4d
gauge theory picture in \cite{KW,Kapustin}.

On the other hand, it is natural to try to develop ``quantum geometric
Langlands'' within the framework of 2d CFT, as a deformation of the
Beilinson--Drinfeld construction at the critical level.

One immediate complication for doing this is that while the chiral
algebra $V_{-{}^L\neg h^\vee}(\wh{^L\neg\fg})$ of $\wh{^L\neg\fg}$ of
level $-{}^L\neg h^\vee$ deforms to the chiral algebra $V_{^L\neg
  k}(\wh{^L\neg\fg})$ of $\wh{^L\neg\fg}$ of level $^L\neg k$, only
the part of the center of $V_{-{}^L\neg h^\vee}(\wh{^L\neg\fg})$
generated by the quadratic Sugawara operators can be deformed. The
center itself cannot be deformed inside $V_{^L\neg k}(\wh{^L\neg\fg})$
if $\fg \neq sl_2$. Luckily, there is another definition of the center
that can be deformed: namely, the definition via the quantum
Drinfeld--Sokolov reduction.

The quantum Drinfeld--Sokolov reduction \cite{FF} (see \cite{FF:ds,BO}
for earlier works and \cite{FB}, Ch. 15 for a survey) is defined by
introducing a BRST complex which is the tensor product of
$V_{^L\neg k}(\wh{^L\neg\fg})$ and the free fermion vertex algebra built on the
Clifford algebra generated by $^L\neg\n(\!(z)\!)  \oplus
{}^L\neg\n^*(\!(z)\!)dz$. Mathematically, it is the complex of Feigin's
semi-infinite cohomology of the Lie algebra $^L\neg\n(\!(z)\!)$ with
coefficients in $V_{^L\neg k}(\wh{^L\neg\fg})$ tensored with a non-degenerate
(Whittaker-like) character. It turns out that this cohomology is
non-zero only in cohomological degree $0$, and the cohomological
degree $0$ part is a vertex algebra called the (quantum) ${\mc
  W}$-algebra associated to $^L\neg\fg$ and level $^L\neg k$ (see
\cite{FF,FB}). (This is one of two known definition of this ${\mc
  W}$-algebra; the other definition, as the intersection of kernels of
the screening operators, is equivalent to it, as explained in
\cite{FF,FB}.)

The notation used for this algebra in \cite{FF,FB} is ${\mc W}_{^L\neg
  k}({}^L\neg\fg)$, but here we will use the notation ${\mc
  W}_{^L\neg\beta}({}^L\neg\fg)$, where $^L\neg\beta = m \cdot {}^L\neg(k+h^\vee)$
($m$ being the lacing number of $^L\neg\fg$ and $\fg$). In particular, in
our notation ${\mc W}_\infty({}^L\neg\fg)$ is the classical ${\mc
  W}$-algebra associated to $^L\neg\fg$ (viewed as a commutative vertex
Poisson algebra).

It turns out that if ${}^L\neg k=-{}^L\neg h^\vee$, then the corresponding ${\mc
  W}$-algebra ${\mc W}_0({}^L\neg\fg)$ also becomes commutative and is in
fact isomorphic to the center of $V_{-{}^L\neg h^\vee}(\wh{^L\neg\fg})$. More
precisely, the natural embedding of the center (placed in
cohomological degree $0$) into the above BRST
complex induces an isomorphism of the cohomologies. This is proved in
\cite{FF} (see also \cite{FB}, Ch. 15). Thus, we obtain an alternative
description of the center at the critical level as the (commutative)
vertex algebra ${\mc W}_0({}^L\neg\fg)$. This description makes it clear how to
deform this vertex algebra: we simply take ${\mc W}_{^L\neg\beta}({}^L\neg\fg)$.

Now recall the isomorphism of \cite{FF} between the center of
$V_{-{}^L\neg h^\vee}(\wh{^L\neg\fg})$ and the classical ${\mc W}$-algebra
associated to $\fg$. In our current notation, it takes the form
\begin{equation}    \label{cWduality}
{\mc W}_0({}^L\neg\fg) \simeq {\mc W}_\infty(\fg).
\end{equation}
It turns out that this isomorphism has a one-parameter deformation
\cite{FF}:

\begin{theorem}    \label{Walg}
For arbitrary complex parameters $\beta$ and $^L\neg\beta$ satisfying
  the relation
\begin{equation}    \label{dualbeta}
\beta = \frac{m}{^L\neg \beta},
\end{equation}
there is an isomorphism of vertex algebra
\begin{equation}    \label{Wduality}
{\mc W}_{^L\neg\beta}({}^L\neg\fg) \simeq {\mc W}_\beta(\fg),
\end{equation}
whose limit as $\beta \to \infty$ is the isomorphism \eqref{cWduality}.
\end{theorem}

\begin{proof} In \cite{FF}, Proposition 5, the isomorphism
  \eqref{Wduality} was proved for generic values of $\beta$ and
  $^L\neg\beta$ satisfying \eqref{dualbeta}. Note that
  relation \eqref{dualbeta} is equivalent to
\begin{equation}    \label{kdual}
m(k+h^\vee) = \frac{1}{{}^L\neg(k+h^\vee)}.
\end{equation}
Furthermore, in \cite{FF} the precise definition of the limit $\beta
\to \infty$ was given so that \eqref{Wduality} becomes
\eqref{cWduality} in this limit.

The isomorphism for arbitrary $\beta$ and $^L\neg\beta$ satisfying
\eqref{dualbeta} follows easily from the results of of
\cite{FF:laws}. Namely, according to Theorem 4.6.9 of
\cite{FF:laws},\footnote{To avoid confusion, we note that the
  parameter $\beta$ we use here is equal to $m/\beta_{\on{FF}}^2$,
  where $\beta_{\on{FF}}$ is the parameter denoted by $\beta$ in
  \cite{FF:laws}.} both ${\mc W}_\beta(\fg)$ and ${\mc
  W}_{^L\neg\beta}({}^L\neg\fg)$, with the parameters $\beta, {}^L\neg\beta$
satisfying \eqref{dualbeta}, can be embedded as vertex subalgebras of
the Heisenberg vertex algebra associated to the Cartan subalgebra
${\mathfrak h}$ of $\fg$, so that their graded characters are
independent of $\beta$. Since the two images coincide for generic
$\beta$, they coincide for all $\beta$.
\end{proof}

In the next subsection, as a small aside, we express Theorem
\ref{Walg} in a slightly more satisfying way, as an isomorphism over
the ring of Laurent polynomials $\C[\beta^{\pm 1}] =
\C[({}^L\neg\beta)^{\pm 1}]$.

\subsection{Oneness}

One possible point of view on the isomorphism of Theorem \ref{Walg} is
that there are two families of ${\mc W}$-algebras: for $\fg$ and for
$^L\neg\fg$, and there is an isomorphism between them if we reverse the
parameter: $\beta \mapsto m/\beta$. However, a more fruitful point of
view might be that there is only one ${\mc W}$-algebra, but we can
look at it from two different points of view: as being associated to
$\fg$ or to $^L\neg\fg$. Accordingly, this quantum ${\mc W}$-algebra has
two classical limits corresponding to these two points of view. In
other words, there is one quantum ${\mc W}$-algebra, but it can be
perceived as the quantization of two different vertex Poisson
algebras.

This can be made more precise by exhibiting this ``unified'' ${\mc
  W}$-algebra as a free $\C[\beta^{\pm 1}]$-module which contains
inside a $\C[\beta^{-1}]$-lattice and a $\C[\beta]$-lattice, the
former ``hailing'' from $\fg$ and the latter from $^L\neg\fg$ (in other
words, $\beta^{-1}$ is the quantization parameter from the point of
view of $\fg$, and $\beta$ is the quantization parameter from the
point of view of $^L\neg\fg$).

According to Theorem 4.6.9 of \cite{FF:laws}, ${\mc W}_\beta(\fg)$ is
freely generated by $\ell=\on{rank}(\fg)$ generators
$W_1,\ldots,W_\ell(z)$ such that the degree (or conformal dimension)
of $W_i$ is $d_i+1$, where $d_i$ is the $i$th exponent of $\fg$. For
non-zero $\beta$, the first of these generators, $W_1$, generates the
Virasoro algebra, and each of the remaining generators $W_i,
i=2,\ldots,\ell$, can be chosen so that it is a highest weight vector
of this Virasoro algebra.

The Heisenberg vertex algebra is, as a vector space, the Fock
representation $\pi_0$ of the Heisenberg Lie algebra with the
generators $b^i_n, i=1,\ldots,\ell; n \in \Z$, which we normalize by
the requirement that they satisfy the relations
$$
[b^i_n,b^j_m] = \beta^{-1} (\alpha_i,\alpha_j) n \delta_{n,-m}.
$$
Consider $\pi_0$ as a free $\C[\beta^{-1}]$-module with the basis of
monomials in $b^i_n, i=1,\ldots,\ell; n<0$, applied to the vacuum
vector. This is a vertex algebra over $\C[\beta^{-1}]$. It follows from the
proof of Theorem 4.6.9 of \cite{FF:laws} that each $W_i$ can be
normalized in such a way that $W_i = W^{(0)} + \beta^{-1}(\ldots)$, where
$W^{(0)}_i$ is a polynomial in $b^i_{-1}, i=1,\ldots,\ell$, invariant
under the action of the Weyl group. Furthermore, $W^{(0)}_i,
i=1,\ldots,\ell$, is a set of generators of the ring of Weyl group
invariant polynomials in $b^i_{-1}$ (in fact, using the conformal
dimension $\Z$-grading on $\pi_0$, we find also that $W^{(0)}_i$ is
the symbol of $W_i$ with respect to the standard PBW filtration on
$\pi_0$; note also that the numbers $d_i+1$ are precisely the degrees
of the generators of the ring of Weyl group invariant polynomials).

According to Theorem 4.6.9 of \cite{FF:laws}, the lexicographically
ordered monomials in the creation operators corresponding to the $W_i,
i=1,\ldots,\ell$, applied to the vacuum vector in $\pi_0$, span a free
$\C[\beta^{-1}]$-submodule and a vertex subalgebra of $\pi_0$. This is
${\mc W}_\beta(\fg)$, viewed as a vertex algebra over $\C[\beta^{-1}]$.

On the other hand, let $^L\neg W_i = \beta^{(d_i+1)} W_i$. Then, applying
the same argument, but from the point of view of $^L\neg\fg$ and
$({}^L\neg\beta)^{-1} = \beta/m$ and using the $^L\neg W_i$'s, we construct a free
$\C[\beta]$-submodule of $\pi_0 \underset{\C[\beta^{-1}]}\otimes
\C[\beta^{\pm 1}]$. This is ${\mc W}_{m/\beta}({}^L\neg\fg)$, viewed as a
vertex algebra over $\C[\beta]$.

Finally, tensoring both of these vertex algebras with $\C[\beta^{\pm
  1}]$, we obtain the promised ``unified'' ${\mc W}$-algebra (of $\fg$
{\em and} $^L\neg\fg$), which contains ${\mc W}_\beta(\fg)$ and ${\mc
  W}_{m/\beta}({}^L\neg\fg)$ as a $\C[\beta^{-1}]$- and a
$\C[\beta]$-lattice, respectively. The two classical limits, ${\mc
  W}_\infty(\fg)$ and ${\mc W}_\infty(^L\neg\fg)$, are defined using these
two lattices (as quotients by the maximal ideal in $\C[\beta^{-1}]$
and $\C[\beta]$, respectively). They are commutative vertex Poisson
algebras.

\subsection{Conformal blocks and quantum geometric Langlands}

These results offer a particular interpretation of the quantum
geometric Langlands correspondence in the language of 2d CFT. Recall
from \cite{Frenkel} (see also the discussion at the end of Subsection
\ref{glc-overview}) that the fibers of the $D$-modules
$\on{Bun}_{^L\neg G}$ constructed by Beilinson and Drinfeld \cite{BD} can
be identified with the duals of the spaces of conformal blocks of
$V_{-{}^L\neg h^\vee}(\Lfgh)$, which can in turn be described in terms of
certain conformal blocks of the (commutative) classical ${\mc
  W}$-algebra ${\mc W}_\infty(\fg)$ (namely, those conformal blocks
that are algebra homomorphisms ${\mc W}_\infty(\fg) \to \C$).

Motivated by this observation, we propose that one of the
manifestations of quantum Langlands correspondence is an {\em
  isomorphism between conformal blocks of representations from certain
  categories of representations of the two vertex (or chiral)
  algebras: the affine Kac--Moody vertex algebra $V_{^L\neg
    k}(\wh{^L\neg \fg})$ and the ${\mc W}$-algebra ${\mc
    W}_\beta(\fg)$}, provided that $\beta$ and $^L\neg k$ are generic
or rational with $\beta<0$ and satisfy the relation equivalent to
\eqref{dualbeta} and \eqref{kdual}:
\begin{equation}    \label{betak}
\beta = \frac{1}{^L\neg (k+h^\vee)}.
\end{equation}
We stress that this is not the only manifestation of the quantum
Langlands correspondence, but it is one that fits well with the
isomorphism of $q$-deformed conformal blocks established in this paper
(see Subsection \ref{GL} for a brief discussion of the links with
other approaches).

Here, and in what follows, ``generic'' means a complex number that is
not rational. However, we expect that most of our results and
conjectures below also hold for those rational $^L\neg k$ that are
less than $-{}^L\neg h^\vee$ (which is equivalent to $\beta<0$ under
the relation \eqref{betak}). We will refer to such $^L\neg k$ as
negative rational.

Let us define precisely the two categories of representations
mentioned above. Representations of the vertex algebra $V_{^L\neg
  k}(\wh{^L\neg \fg})$ are the same as representations of the affine
Lie algebra $\Lfgh$ of level $^L\neg k$ satisfying a finiteness
condition: every vector is annihilated by the Lie subalgebra $z^N
\cdot {}^L\neg \fg[[z]]$ for sufficiently large $N$. We denote the
category of such representations by $\Lfgh_{^L\neg k}$-mod. Let
$\Lfgh_{^L\neg k}$-mod$^0$ be the category of those representations of
$\Lfgh$ of level $^L\neg k$ on which the action of the Lie subalgebra
$^L\neg \fg[[z]]$ can be exponentiated to the corresponding Lie group
$^L\neg G[[z]]$. This is the same as a full subcategory of the usual
category ${\mc O}$ of $\Lfgh$ of level $^L\neg k$ whose objects are
the representations whose restriction to the constant Lie subalgebra
$^L\neg \fg$ decomposes as a direct sum of finite-dimensional
representations.

We will assume that $^L\neg k$ is generic or negative rational. Then
simple objects of this category are labeled by dominant integrals
weights $\lambda \in {}^L\neg P^+$ of $^L\neg \fg$. The simple object
$L_{\lambda,{}^L\neg k}$ corresponding to $\lambda \in {}^L\neg P^+$
is the unique irreducible quotient of the Weyl module over $\Lfgh$ of
level $^L\neg k$ induced from the irreducible representation of
$^L\neg \fg$ with highest weight $\lambda$ (note that for any dominant
integral weight, $L_{\lambda,{}^L\neg k}$ is the Weyl module itself if
$^L\neg k$ is not a rational number). The category $\Lfgh_{^L\neg
  k}$-mod$^0$ can be defined as the full subcategory of $\Lfgh_{^L\neg
  k}$-mod whose objects are representations with irreducible
subquotients of this form. Imposing an additional property that
representations have finitely many irreducible subquotients (i.e.,
have finite composition series), we obtain the category extensively
studied by Kazhdan and Lusztig, who in particular defined the
structure of a braided tensor category on it, see \cite{KL}. This will
be our category of representations of the vertex algebra $V_{^L\neg
  k}(\wh{^L\neg \fg})$.

Note that if $^L\neg k$ is generic in the above sense, then the category
$\Lfgh_{^L\neg k}$-mod$^0$ is a semi-simple abelian category that is
equivalent to the category of finite-dimensional representations of
$^L\neg \fg$.

Next, we define a subcategory of the category of representations of
the ${\mc W}$-algebra ${\mc W}_\beta(\fg)$. Recall that the quantum
Drinfeld--Sokolov reduction yields a functor \cite{FF,FKW} from the
category $\Lfgh_{^L\neg k}$-mod to an analogous category of modules
over ${\mc W}_{^L\neg \beta}({}^L\neg \fg)$, which is isomorphic ${\mc
  W}_\beta(\fg)$ if $\beta$ satisfies \eqref{betak} \cite{FF}. We will
henceforth denote the latter category by ${\mc W}_\beta(\fg)$-mod and
the Drinfeld--Sokolov reduction functor by $H^{^L\neg
  \fg}_{\on{DS}}$. This functor sends a $\Lfgh_{^L\neg k}$-module $M$
to the semi-infinite cohomology of $^L\neg \n(\!(z)\!)$ with
coefficients in $M$ tensored with the same non-degenerate character
that was used to define ${\mc W}_{^L\neg \beta}({}^L\neg \fg)$. We
denote it by $H^{^L\neg \fg}_{\on{DS}}(M)$.

It follows from the results of Arakawa \cite{Arakawa} that for generic
$k$ the functor $H^{^L\neg \fg}_{\on{DS}}$ is exact on $\Lfgh_{^L\neg
  k}$-mod$^0$ (see also \cite{Raskin}).

Let now ${\mc W}_\beta(\fg)$-mod$^0$ be the full subcategory of ${\mc
  W}_\beta(\fg)$-mod whose objects have finite composition series with
irreducible subquotients being the modules $H^{^L\neg
  \fg}_{\on{DS}}(L_{\lambda,{}^L\neg k}), \lambda \in {}^L\neg
P^+$. These modules are irreducible and non-zero, according to
\cite{Arakawa1}. This will be our category on the ${\mc W}$-algebra
side.

There is an alternative way to describe the simple objects of ${\mc
  W}_\beta(\fg)$-mod$^0$, in terms of the quantum Drinfeld--Sokolov
reduction of $\ghat_k$ rather than $\Lfgh_{^L\neg k}$. Namely, instead
of applying the quantum Drinfeld--Sokolov reduction functor $H^{^L\neg
  \fg}_{\on{DS}}$ to $L_{\lambda,{}^L\neg k}, \lambda \in {}^L\neg
P^+$, we apply the quantum Drinfeld--Sokolov reduction
$H^{\fg,\lambda}_{\on{DS}}$, with the standard character twisted by
the element $\lambda(z) \in H(\!(z)\!)$, to the vacuum module
$L_{0,k}$ over $\ghat_k$. Using the methods of \cite{FF}, one can
show\footnote{{\em Note added in proof}: This has been proved in
  T. Arakawa and E. Frenkel, Quantum Langlands duality of
  representations of ${\mathcal W}$-algebras, arXiv:1807.01536.}
that for generic $k$ we have
\begin{equation}    \label{iso W}
H^{\fg,\lambda}_{\on{DS}}(L_{0,k}) \simeq H^{^L\neg
  \fg}_{\on{DS}}(L_{\lambda,{}^L\neg k}),
\end{equation}
Namely, using the free field realization of the ${\mc W}$-algebras
along the lines of \cite{FF}, we can construct the modules on the two
sides of \eqref{iso W} as zeroth cohomologies of a finite BGG-type
resolution and use it to see that their characters coincide. We expect
\eqref{iso W} to be true for negative rational ${}^L\neg k$ as well.

From the point of view of ${\mc W}_\beta(\fg)$, these simple modules
correspond to the ``magnetic'' vertex operators, while from the point
of view of ${\mc W}_{^L\neg \beta}({}^L\neg \fg)$ they correspond to
the ``electric'' vertex operators (the duality of ${\mc W}$-algebras
exchanges electric and magnetic vertex operators). We find it
convenient to view quantum Drinfeld--Sokolov reduction $H^{^L\neg
  \fg}_{\on{DS}}$ as a functor from the category $\Lfgh_{^L\neg
  k}$-mod$^0$ to the category ${\mc W}_\beta(\fg)$-mod$^0$, which
sends irreducible modules to irreducible ``magnetic'' modules (the
reason for this will become clear in Subsection \ref{a q-def}).

We conjecture that this functor is in fact an equivalence of
braided tensor categories. This agrees with the fact known in 2d CFT
that magnetic vertex operators for the ${\mc W}$-algebra and vertex
operators corresponding to the Weyl modules over $\Lfgh_{^L\neg k}$
braid as representations of the quantum group $U_{q'}({}^L\neg g)$,
where $q' = e^{2\pi i/{}^L \neg(k+h^\vee)}$.

The fusion tensor product of modules over a given vertex algebra has
been defined by Huang and Lepowsky (see \cite{HL} and references
therein). Although ${\mc W}_\beta(\fg)$ does not satisfy the
conditions of \cite{HL}, various results from 2d CFT suggest that the
fusion tensor product endows ${\mc W}_\beta(\fg)$-mod$^0$ with the
structure of a braided tensor category. This leads us to the following

\begin{conjecture}    \label{cateq}
The category ${\mc W}_\beta(\fg)$-mod$^0$ is a braided tensor category
(with respect to the fusion tensor product) which is equivalent to
the category $\Lfgh_{^L\neg k}$-mod$^0$ as a braided tensor category if
$\beta$ and $^L\neg k$ satisfying equation \eqref{betak} with $\beta$
generic or negative rational.
\end{conjecture}

A more precise statement is that this should be an equivalence of
chiral categories (see \cite{Raskin:ch} for the definition).

We view Conjecture \ref{cateq} as a purely algebraic manifestation of
the quantum Langlands correspondence. (A closely related
Gaitsgory--Lurie conjecture discussed in Subsection \ref{GL} below has
the algebraically defined category $\Lfgh_{^L\neg k}$-mod$^0$ on one
side and a Whittaker category, which is defined in geometric terms, on
the other side; here the role of the Whittaker category is played by an
algebraically defined category ${\mc W}_\beta(\fg)$-mod$^0$ as
well.) It is a local statement, but it implies non-trivial global
statements: namely, isomorphisms of the spaces of conformal blocks of
representations from the categories $\Lfgh_{^L\neg k}$-mod$^0$ and
${\mc W}_\beta(\fg)$-mod$^0$.

Recall that for any vertex algebra $V$ and a collection of $V$-modules
$M_1,\ldots,M_n$ attached to points $p_1,\ldots,p_n$ of a smooth
projective algebraic curve ${\mc C}$, one can define the vector space
of conformal blocks $C_V({\mc C},(p_i),(M_i))$ (see \cite{FB}, Section
10.1). In the case of the vertex algebra $V = V_k(\ghat)$, this is the
standard definition of the space of conformal blocks for $\ghat$ (see,
e.g., \cite{Frenkel}).

Suppose that ${\mc C} = {\mathbb C}{\mathbb P}^1$. If $M_1,\ldots,M_n$
are objects of a braided tensor category with respect to a fusion
tensor product $\otimes$, then the space $C_V({\mathbb C}{\mathbb
  P}^1,(p_i),(M_i))$ can be expressed as a Hom of this category:
$$
C_V({\mathbb C}{\mathbb P}^1,(p_i),(M_i)) \simeq \on{Hom}_V(M_1
\otimes \ldots \otimes M_{n-1},M_n^\vee).
$$

Therefore Conjecture \ref{cateq} implies (at least, for ${\mc C} =
{\mathbb C}{\mathbb P}^1$, and this should be true for all ${\mc C}$
if Conjecture \ref{cateq} is true at the level of chiral categories):

\begin{conjecture}    \label{isomblocks}
There are isomorphisms of the spaces of conformal blocks
\begin{equation}    \label{isomcb}
C_{V_{^L\neg k}(\wh{^L\neg \fg})}({\mc C},(p_i),(L_{\lambda_i,{}^L\neg k}))
\simeq
C_{{\mc W}_\beta(\fg)}({\mc C},(p_i),(H_{\on{DS}}(L_{\lambda_i,{}^L\neg k})))
\end{equation}
provided that the parameters satisfy the conditions of Conjecture
\ref{cateq}.
\end{conjecture}

In the case of ${\mc C} = {\mathbb C}{\mathbb P}^1$, then for generic
$^L\neg k$ the isomorphisms \ref{isomcb} can indeed be constructed
using the integral representation of the spaces of conformal blocks,
as we discuss in Subsections \ref{int} and \ref{exp}. This gives us a
concrete way to prove Conjectures \ref{isomblocks} and \ref{cateq},
and more general Conjectures \ref{isomblocks1} and \ref{cateq1} below.

\subsection{A $q$-deformation}    \label{a q-def}

At this point, it is natural to ask to what extent it is necessary to
invoke the dual Lie algebra in the above conjectures. Indeed, using
the duality of ${\mc W}$-algebras \cite{FF} (see Theorem \ref{Walg}),
we can replace ${\mc W}_\beta(\fg)$ by ${\mc W}_{^L\neg
  \beta}({}^L\neg \fg)$ with $^L\neg \beta = m/\beta$ in Conjectures
\ref{cateq} \ref{isomblocks}. So, at first glance it may appear that
the above results and conjectures can be accounted for by the
Drinfeld--Sokolov reduction alone, and that there is no need to invoke
the Langlands dual Lie algebra (in the same way, one would tend to
dismiss the appearance of $^L\neg\fh$ in the Harish-Chandra
homomorphism, as we explained at the beginning of this section).

However, there are two reasons why Langlands duality is relevant
here. First, as we already explained at the beginning of this section,
{\em the isomorphism of conformal blocks ${\mc W}_\beta(\fg)$ and
  $\wh{^L\neg \fg}$ can be $q$-deformed}, and after the
$q$-deformation the appearance of $\fg$ can't be written off because
there is no longer an isomorphism between the deformed ${\mc
  W}$-algebras associated to $\fg$ and $^L\neg \fg$ (if $^L\neg \fg
\neq \fg$). It is really ${\mc W}_{q,t}(\fg)$ that appears in our
isomorphism of $q$-deformed conformal blocks, and not ${\mc
  W}_{t,q}({}^L\neg \fg)$ (unless $\fg={}^L\neg
\fg$). Furthermore, our isomorphism involves the deformations of the
magnetic vertex operators over the ${\mc W}$-algebra of $\fg$, and
these are no longer equal to the deformations of the electric vertex
operators of the ${\mc W}$-algebra of $^L\neg\fg$. Thus, the
appearance of the Langlands dual Lie algebras becomes more meaningful
after the $q$-deformation (similarly to how the appearance of the
Langlands dual Lie algebra in the Harish-Chandra isomorphism becomes
more meaningful after affinization). This suggests that it is fruitful
to view the isomorphism between conformal blocks at $q=1$, and the
corresponding equivalences of categories, in the light of Langlands
duality as well.

The second reason is that actually Conjectures \ref{cateq} and
\ref{isomblocks} are special cases of more general conjectures
corresponding to the generalized dualities $T^NS$ of the group
$PSL_2(\Z)$ familiar from 4d gauge theory (the standard Langlands
duality corresponds to $S$, i.e. $N=0$). But to apply this duality we
must first apply the duality $S$, exchanging $\fg$ and $^L\neg\fg$,
and then apply $T^N$ (which preserves $\fg$ and
$^L\neg\fg$). Therefore, if we wish to look at the dualities $T^NS$
with $N\neq 0$, then using the dual Lie algebra is necessary already
at $q=1$.

In fact, and this is a crucial point, the isomorphism of conformal
blocks obtained from our canonical isomorphism of $q$-deformed
conformal blocks in the limit $q \to 1$ corresponds not to the
standard relation \eqref{betak} but to the relation
\begin{equation}    \label{betak1}
\beta = \frac{1}{^L\neg(k+h^\vee)} + m,
\end{equation}
where $m$ is the lacing number of $\fg$. Indeed, this is the relation
we obtain when we take the limit $q\to 1$ in the relation
\eqref{first} between the parameters of the algebras ${\mc
  W}_{t,q}({}^L\neg \fg)$ and $U_{\hbar}(\wh{^L\neg \fg})$ using
equation \eqref{zero}.

Formula \eqref{betak1} differs from formula \eqref{betak} in the shift
of $\beta$ by $m$. This shift corresponds to applying, in addition to
the standard Langlands duality $S$, the quantum Langlands duality
$T$. Let us recall how the dualities $T$ and $S$ act on the
parameters of 4d gauge theory.

The duality $S$ exchanges the gauge groups $G$ and $^L\neg G$ (and
hence the corresponding Lie algebras) and acts on the 4d gauge theory
coupling constant $\tau$ as
$$
S: \tau \mapsto -1/m\tau.
$$
The duality $T$ preserves the gauge group and acts on $\tau$ as
$$
T: \tau \mapsto \tau+1
$$
(it is well-defined if $G$ is simply-connected, which we will now
assume to be the case; in general, only certain powers of $T$ are
well-defined). These two dualities generate a subgroup of $PSL_2(\Z)$
(see \cite{KW}).

The connection to our parameters is as follows:
$$
\tau = \beta/m, \qquad {}^L\neg\tau = - {}^L\neg(k+h^\vee).
$$
Hence, formula \eqref{betak1} is equivalent to
\begin{equation}    \label{TS}
\tau = -1/m{}^L\neg\tau + 1 = TS({}^L\neg\tau).
\end{equation}
In order to interpret the relation between conformal blocks
corresponding to $\beta$ and $^L\neg k$ related via formula
\eqref{betak1}, we need a generalization of Conjectures \ref{cateq}
and \ref{isomblocks} in which we replace the relation \eqref{betak}
corresponding to the duality $S$ with \eqref{betak1} corresponding to
$TS$.

We will consider an even more general relation corresponding to the
duality $T^N S$:
\begin{equation}    \label{betakN}
\beta = \frac{1}{^L\neg (k+h^\vee)} + Nm, \qquad N \in \Z,
\end{equation}
and the following conjectures:

\begin{conjecture}    \label{cateq1}
The categories ${\mc W}_\beta(\fg)$-mod$^0$ and $\Lfgh_{^L\neg
  k}$-mod$^0$ are equivalent as braided tensor categories (or chiral
categories) if $\beta$ and $^L\neg k$ satisfy equation \eqref{betakN}
with $\beta$ generic or negative rational.
\end{conjecture}

\begin{conjecture}    \label{isomblocks1}
There are isomorphisms \eqref{isomcb} of the spaces of conformal
blocks provided that the parameters satisfy the conditions of
Conjecture \ref{cateq1}.
\end{conjecture}

Now, the isomorphism \eqref{isomcb} with $\beta$ and $^L\neg k$
related by formula \eqref{betak1} is precisely the $q \to 1$ limit of
the canonical isomorphism of $q$-deformed conformal blocks which we have
conjectured in this paper and established in the simply-laced case. It
is in this sense that we can view our isomorphism as a manifestation
of a $q$-deformation of the quantum geometric Langlands.

\subsection{Connection with the Gaitsgory--Lurie
  conjecture}    \label{GL}

Conjecture \ref{cateq} is related to a conjecture of Gaitsgory and
Lurie (proved by Gaitsgory in \cite{gaitsW} for generic parameters; see
also \cite{gaitsQ}) stating an equivalence of two braided tensor
categories (or chiral categories in the terminology used in
\cite{gaitsW,gaitsQ,Raskin:ch}). In our notation, one of them is the
above category $\wh{^L\neg \fg}_{{}^L\neg k}$-mod$^0$ (which is
denoted by $\operatorname{KL}_{\check{G}}^{\check{c}}$ in
\cite{gaitsQ}). The other is the ``Whittaker category'' denoted by
$\operatorname{Whit}^c(\operatorname{Gr}_G)$ in \cite{gaitsQ}.

Combining Conjecture \ref{cateq} with the theorem of \cite{gaitsW}
(the Gaitsgory--Lurie conjecture for generic $c$), we obtain

\begin{conjecture} \label{lurieW}
  The categories $\operatorname{Whit}^c(\operatorname{Gr}_G)$ and
  ${\mc W}_c(\fg)\on{-mod}^0$ are equivalent as braided tensor (or
  chiral) categories for generic and negative rational $c$.
\end{conjecture}

We note that both categories have simple objects labeled by $\lambda
\in {}^L\neg P^+$, and they should correspond to each other under this
equivalence. There is also a natural functor from
$\operatorname{Whit}^c(\operatorname{Gr}_G)$ to ${\mc
  W}_c(\fg)$-mod. Indeed, according to the definition given in
\cite{gaitsQ}, $\operatorname{Whit}^c(\operatorname{Gr}_G)$ is
$\on{Whit}({\mc D}_k(\on{Gr}_G)\on{-mod})$, the category of
$({}\n(\!(z)\!),\chi)$-equivariant objects in the category of twisted
${\mc D}$-modules on the affine Grassmannian $\on{Gr}_G$ \cite{gaitsQ}
(here $\chi$ is the ``Whittaker functional'' used in the quantum
Drinfeld-Sokolov reduction, and the twisting parameter should be, in
our notation, the level $k$ such that $c=m(k+h^\vee)$). The functor of
global sections on $\on{Gr}_G$ then yields a functor from the latter
category to $\operatorname{Whit}(\ghat_k\on{-mod})$, which is
equivalent to the category ${\mc W}_c(\fg)$-mod according to the
results of \cite{Raskin}.

%Furthermore the conjecture is compatible with the limits
%$\beta \to 0$ and $\beta \to \infty$ (for
%$\operatorname{Whit}^c(\operatorname{Gr}_G)$, these limits are
%described in \cite{gaitsQ}).

Conjecture 4.5 of \cite{gaitsQ} links the statement of the
Gaitsgory--Lurie conjecture to the global quantum Langlands
correspondence discussed in Section \ref{glc-def} above. Therefore,
Conjecture \ref{lurieW} provides a link between our Conjecture
\ref{cateq} and the global quantum Langlands correspondence.

\subsection{Integral representation of conformal
  blocks}    \label{int}

Conjecture \ref{isomblocks} in genus 0 can be tested using the
integral formulas for the conformal blocks of affine Kac--Moody
algebras obtained by Schechtman and Varchenko \cite{SV} (as solutions
of the KZ equations). These formulas can also be obtained using the
free field (Wakimoto) realization of $\Lfgh$, see
\cite{ATY,FFR,EFK}. In this section we compare these formulas to the
integral formulas for conformal blocks of ${\mc W}$-algebras. In gives
us a concrete interpretation of the limit of our isomorphism of
$q$-deformed conformal blocks as $q \to 1$.

Our notation for the conformal blocks will be similar to the notation
we used for the $q$-deformed conformal blocks. Namely, we have a
vertex operator ${\Phi}_{{}^L\neg \rho_i}(a_i)$ corresponding to a
finite-dimensional representation $^L\neg \rho_i$ of $^L\neg \fg$ of dominant
integral highest weight $\lambda_i \in {}^L\neg P^+$ inserted at the point
$a_i \in \pone$, for $i=1,\ldots,n$. Then conformal blocks may be
viewed as (multivalued) functions of the $a_i$ with values in a weight
space in the tensor product $\otimes_i {}^L\neg \rho_i$. This weight is
given by the same formula as \eqref{weight} (here we use a slightly
different notation; in particular, we denote the simple roots of
$^L\neg \fg$ by $\al_i$):
\begin{equation}
\gamma = \sum_{i=1}^n {}\lambda_i - \sum_{j=1}^r {}\alpha_{i_j}. 
\end{equation}

In the integral representation, these functions are written as
integrals, over a suitable integration cycle $\Gamma$ (discussed
below), in the space
$$
(\C \backslash \{ a_1,\ldots,a_n \})^r \backslash \on{diag}
$$
with coordinates $x_1,\ldots,x_r$, of a function that is a
product of two factors:

(1) The first factor is the multivalued function, denoted by ${\cal
  I}(x,a)$, which is the product of factors of three types:
$$
(a_i-a_j)^{(\la_i,{}\la_j)/{}^L\neg (k+h^\vee)},
$$
\begin{equation}    \label{mvf}
(a_i-x_j)^{-({}\la_i,{}\al_{i_j})/{}^L\neg (k+h^\vee)}, \qquad
(x_j-x_p)^{(\al_{i_j},{}\al_{i_p})/{}^L\neg (k+h^\vee)}
\end{equation}
(here we use the inner product normalized as in Section
\ref{sec:two}).

(2) The second factor is a rational function $|x_1^{i_1} \ldots
x_r^{i_r}\rangle$ in the $a_i$ and $x_j$ with values in $\left(
  \otimes_i {}^L\neg \rho_i \right)_\gamma$; it can be realized as a
conformal block of the bosonic $\beta\gamma$-system involved in the
free field realization of $\Lfgh$ (see Theorem 4 of
\cite{FFR} as well as \cite{ATY}).

The product of the factors appearing in equation \eqref{mvf} defines a
rank one local system ${\mc L}$ on $(\C \backslash \{ a_1,\ldots,a_n
\})^r \backslash \on{diag}$. For the integral to be well-defined, the
integration cycle $\Gamma$ should be viewed as an element of the $r$th
homology group of
$$
(\C \backslash \{ a_1,\ldots,a_n \})^r \backslash \on{diag}
$$
with coefficients in the dual local system ${\mc L}^*$.

It is known that for generic $^L\neg  k$ the resulting integrals
\begin{equation}    \label{intsol}
\int_\Gamma \; {\cal I}(x,a)  \; |x_1^{i_1} \ldots
x_r^{i_r}\rangle \; dx_1 \ldots dx_r
\end{equation}
(with varying $\Gamma$) span the subspace of highest weight vectors of
the weight space $\left( \otimes_i {}^L\neg \rho_i \right)_\gamma$ with
respect to the diagonal action of $^L\neg \fg$. This may seem puzzling
because if we only had vertex operators ${\Phi}_{{}^L\neg \rho_i}(a_i),
i=1,\ldots,n$, in our set-up, then the space of conformal blocks would
have been isomorphic to the subspace of $^L\neg \fg$-invariant vectors in
$\left( \otimes_i {}^L\neg \rho_i \right)_\gamma$. The explanation is
that we have ``cheated'' a bit because to make this calculation work
we actually need to insert a vertex operator at the point $\infty \in
\pone$ with the lowest weight $-\gamma$ (this is explained in
\cite{FFR}). When we take this into account, the corresponding space
of $^L\neg \fg$-invariant vectors gets identified with the space of highest
weight vectors in $\left( \otimes_i {}^L\neg \rho_i
\right)_\gamma$. (Note that the measure in \eqref{intsol} is different
from the rest of the paper because in this section $x$ is a coordinate
on the complex plane rather than a cylinder. To connect the formulas
in this section to the formulas elsewhere, one should use the change
of variables $x_{\textup{cyl}} = e^{{R} x_{\textup{plane}}}$ and
take $R$ to zero.)

In fact, it follows from the results of Varchenko \cite{Varchenko1}
that for generic $^L\neg k$ the above homology space can be identified with
the space of highest weight vectors in $\left( \otimes_i
  {}^L\neg \rho^{q'}_i \right)_\gamma$ where $\rho^{q'}_i$ is the
representation with the same highest weight $\la_i$
but over the quantum group
$U_{q'}({}^L\neg \fg)$ with ${q'}=e^{2\pi i/{}^L\neg (k+h^\vee)}$. As explained
in \cite{Varchenko1}, these integral formulas may therefore be thought
of as providing a nondegenerate pairing between these spaces of highest
weight vectors, one for the Lie algebra $^L\neg \fg$ and one for the
quantum group $U_{q'}({}^L\neg \fg)$. (As shown in \cite{Varchenko1}, the
fact that these are solutions of the KZ equations can be used to
derive the Kohno--Drinfeld theorem identifying the $R$-matrices of
$U_{q'}({}^L\neg \fg)$ with the ``half-monodromies'' of solutions of the KZ
equations corresponding to exchanging the points $a_i$ and $a_j$. See
also \cite{SV1} and the closely related work by Bezrukavnikov,
Finkelberg, and Schechtman \cite{BFS}.)

It is possible to modify the construction slightly to obtain the
entire weight space $\left( \otimes_i {}^L\neg \rho^{q'}_i \right)_\gamma$
(rather than its subspace of highest weight vectors). For that, we
also insert a vertex operator at the point $0 \in \pone$ as well as a
vertex operator at the point $\infty$ (we assume that $a_i \neq 0$ for
all $i=1,\ldots,n$). If the highest weight of the former is $\la$ and
the lowest weight of the latter is $-\la'$ so that $\la-\la' =
\gamma$, then we can identify our conformal blocks with the matrix
elements
\beq\label{electric1} \langle \lambda'| \;\prod_{i=1}^n {
  \Phi}_{{}^L\neg \rho_i} (a_{i}) \;|\lambda \rangle \eeq
as in formula \eqref{electric} (note that here we switch the points
$0$ and $\infty$ compared to Section \ref{sec:two}; our $\lambda,
\lambda'$ are therefore $\la_\infty, \la_0$ of formula
\eqref{electric}).

The advantage is that we now get solutions that span the entire weight
space $\left( \otimes_i {}^L\neg \rho_i \right)_\gamma$. Indeed, if
$\lambda$ is chosen to be generic, then the space of highest weight
vectors in the tensor product of the Verma module with the highest
weight $\lambda$ and $\otimes_i {}^L\neg \rho_i$ can be identified with
$\otimes_i {}^L\neg \rho_i$. The disadvantage, however, is that we have
to modify formula \eqref{intsol} by inserting an additional factor,
which is a product of powers of the $x_j$ -- this factor comes from
the ``interaction'' of the vertex operator at the point $0$ and the
screening operators. The resulting formula for the conformal blocks
reads
\begin{equation}    \label{intsol1}
\int_\Gamma  \; \prod_{j=1}^r
x_j^{-(\la,{}\al_{i_j})/{}^L\neg (k+h^\vee)}\;  {\cal I}(x,a)  \; |x_1^{i_1} \ldots
x_r^{i_r}\rangle \; dx_1 \ldots dx_r.
\end{equation}
Accordingly, $\Gamma$ is now a cycle in the $r$th homology of
$$
(\C \backslash \{0, a_1,\ldots,a_n \})^r \backslash \on{diag}
$$
with coefficients in the dual local system of the rank one local
system obtained by modifying ${\mc L}$ to include the monodromies
around $0$ specified by the extra factor in \eqref{intsol1}. This
homology space is, according the results of \cite{Varchenko1},
isomorphic to the weight space $\left( \otimes_i {}^L\neg \rho^{q'}_i
\right)_\gamma$.

\bigskip

Now let us discuss conformal blocks of the ${\mc W}$-algebra ${\mc
  W}_\beta(\fg)$, where $\beta$ is related to $^L\neg k$ by formula
\eqref{betak}. We now insert at the points $a_i$ vertex operators of
${\mc W}_\beta(\fg)$ corresponding to the representations
$H_{\on{DS}}(L_{\la_i,{}^L\neg k})$. In the free field realization of
${\mc W}_\beta(\fg)$, these vertex operators are given by the standard
bosonic vertex operators. However, in the same way as in the
Kac--Moody case, we can insert integrals of the screening currents
which commute with the ${\mc W}$-algebra.

As in the case of the deformed ${\mc W}$-algebra, there are two sets
of screening currents: the ``electric'' and ``magnetic'' ones (see
Section 8.6 of \cite{Frenkel}). They are the conformal limits of the
screening currents $S_a(x)$ and $S_a^{\vee}$ corresponding to the roots
and coroots of $\fg$, respectively (see Section
\ref{sec:four}). However, since we only consider the insertions of the
vertex operators corresponding to the representations
$H_{\on{DS}}(L_{\la_i,{}^L\neg k})$, where each $\la_i$ is a dominant
integral coweight of $\fg$ (equivalently, weight of $^L\neg \fg$), only the
magnetic screening currents corresponding to the coroots of $\fg$
(equivalently, roots of $^L\neg \fg$) appear in the formulas for conformal
blocks.

The resulting formula for the conformal blocks, which are the
$q\rightarrow 1$ limits of the deformed blocks given by
\eqref{magnetic}, is
\begin{equation}    \label{Wblock}
\int_\Gamma \; {\cal I}(x,a)  \; \; dx_1 \ldots dx_r
\end{equation}
if we do not include a vertex operator at the point $0$, and
\begin{equation}    \label{Wblock1}
\int_\Gamma \; \prod_{j=1}^r
x_j^{-(\mu,{}\al_{i_j})} \;  {\cal I}(x,a) \; dx_1 \ldots dx_r
\end{equation}
if we do. We include the ${\mc W}$-algebra vertex operator at $0$ with
momentum $\mu$. It is natural to use the definition of the momentum
corresponding to $\fg$ rather than $^L\neg \fg$; for this reason $\mu$ does
not get rescaled by $\beta$ in \eqref{Wblock1}. The powers of the
$x_j$ are the same in \eqref{Wblock1} and \eqref{intsol1}, if we let
$\mu=\lambda/{}^L\neg (k+h^{\vee})$.

The difference between formulas \eqref{intsol} and \eqref{intsol1} on
one side, and \eqref{Wblock} and \eqref{Wblock1} is that the former
take values in $\left( \otimes_i {}^L\neg \rho_i \right)_\gamma$ whereas
the latter are scalar-valued functions. But what matters is that they
are parametrized by the integration cycles $\Gamma$ which belong to
the same homology space. In both case (with or without a vertex
operator at $0$) the spaces of conformal blocks for $\Lfgh$ and for
${\mc W}_\beta(\fg)$ are therefore identified with the {\em same}
homology space.

This enables us to identify the two spaces of conformal blocks, in
effect proving Conjecture \ref{isomblocks} for ${\mc C}$ of genus zero and
generic values of $\beta$ and $^L\neg k$ satisfying the relation
\eqref{betak}.

\bigskip

However, it would be desirable to identify the integral formulas more
directly. By that we mean finding a linear functional (covector)
$\langle W|$ on $\left( \otimes_i {}^L\neg \rho_i \right)_\gamma$ so that
pairing it with the Kac--Moody conformal block \eqref{intsol1} we
would get the conformal block \eqref{Wblock1} of the ${\mc
  W}$-algebra. Morally, $\langle W|$ should be a ``Whittaker-like''
functional (which makes sense since the ${\mc W}$-algebra is obtained
from the affine Kac--Moody algebra via the quantum Drinfeld-Sokolov
reduction that uses a Whittaker functional).

However, by inspecting of formulas \eqref{intsol1} and \eqref{Wblock1}
we can see that such a covector $\langle W|$ does {\em not}
exist. Indeed, for the formulas to match, we need to have
$$
\langle W|x_1^{i_1} \ldots x_r^{i_r}\rangle = 1,
$$
where $|x_1^{i_1} \ldots x_r^{i_r}\rangle$ is the vector appearing in
formula \eqref{intsol1}, but $\langle W|$ should not depend on the
integration variables $x_j$. Explicit formula for $|x_1^{i_1} \ldots
x_r^{i_r}\rangle$ (see, e.g., Theorem 4 of \cite{FFR}) shows that it
goes to $0$ if we take all of the $x_j$ to $\infty$. Therefore the
covector $\langle W|$ satisfying the above formula does not exist.

The results of this paper show nonetheless that such a covector does
exist for the generic $q$-deformation of conformal blocks subject to
the relation \eqref{first} (and it is indeed something like a
Whittaker functional as it represents the identity in the equivariant
K-theory of the corresponding quiver variety). However, the $q \to 1$
limit of this relation is not the standard relation \eqref{betak} but
rather the relation \eqref{betak1} in which $\beta$ is shifted by
$m$. We have conjectured in Conjecture \ref{isomblocks1} that there is
an isomorphism of conformal blocks in this case, and even more general
case of relation \eqref{betakN}, in which $\beta$ is shifted by $Nm$.

Let us discuss this shift in the framework of the above integral
formulas. Recall that the inner product $(\cdot,\cdot)$ on the dual
space to the Cartan subalgebra of $^L\neg \fg$ is normalized in such a
way that the long roots have square norm 2, and so the short roots
have square norm $2/m$ (here $m$ denotes the lacing number of $\fg$
and $^L\neg \fg$, as before). Given that all the $\lambda_i$ are
dominant integral weights of $^L\neg \fg$, we see that the rank one
local system defined by the multivalued function ${\cal I}$ appearing
in the above integral formulas does not change if we shift $\beta$ by
an integer multiple of $m$. Therefore we find that the corresponding
homology groups remain the same, in agreement with Conjecture
\ref{isomblocks1}.

However, the case $N=1$ (relation \eqref{betak1}) turns out to be
special. In this case, we obtain a direct identification of the
integral formulas for the conformal blocks using a covector
$\langle W|$.

In the next subsection, we will give some explicit examples of this
covector.

\subsection{Explicit identification of conformal blocks}    \label{exp}

Let us discuss a concrete example of the identification of conformal
blocks of $\Lfgh$ and ${\mc W}_\beta(\fg)$, with the parameters
satisfying the relation \eqref{betak1}, in the case of $^L\neg \fg =
sl_2$. To simplify our notation, we will denote $^L\neg  k$ by $k$ in this
subsection.

First, suppose there are two points on the complex plane, $a_1$ and
$a_2$, and we insert at each of them the vertex operator of
$\wh{sl_2}$ corresponding to the two-dimensional representation
$\C^2$, in which we choose a basis $\{ v,fv \}$, with $v$ a highest
weight vector and $f$ the standard generator of $sl_2$. We also put a
vertex operator at $\infty$, but for now we will not put a vertex
operator at the point $0$. We will choose the one at $\infty$ to be of
lowest weight $0$, so that the resulting conformal blocks take values
in the subspace of weight $0$ in $\C^2 \otimes \C^2$. This subspace is
two-dimensional, with a basis $\{ v \otimes fv, fv \otimes v \}$.

Since we are not putting anything at $0$, the space of conformal
blocks is one-dimensional, and can be identified with the space of
highest weight vectors of $\C^2 \otimes \C^2$ -- but viewed as a
representation of $U_{q'}(sl_2)$ (with $q'=e^{2\pi i/(k+2)}$) rather than
$sl_2$.

As we discussed in the previous subsection, this is a general
phenomenon: $\Lfgh$-conformal blocks {\em take values} in the subspace
of highest weight vectors in the tensor product of representations of
$^L\neg \fg$, but the space of conformal blocks itself is isomorphic to the
space of the {\em cycles of integration} that can be identified
\cite{Varchenko1} with the same subspace in the tensor product of the
finite-dimensional representations of the same highest weights, but
taken over the corresponding quantum group $U_{q'}({}^L\neg \fg)$. 

In the case at hand, the integral solution \eqref{intsol} of the KZ
equations is given by the formula
\begin{equation}    \label{2pt}
(a_1-a_2)^{\theta/2} \int_\Gamma \left( \frac{fv \otimes v}{x-a_1} +
  \frac{v \otimes fv}{x-a_2} \right) (x-a_1)^{-\theta}
(x-a_2)^{-\theta} \; dx,
\end{equation}
where
$$
\theta = \frac{1}{k+2}.
$$
There are two things to note:

(1) This solution takes values in the subspace of the weight $0$
subspace, spanned by the vector
$$
fv \otimes v - v \otimes vf,
$$
which is precisely the subspace of highest weight vectors of weight 0
in $\C^2 \otimes \C^2$, as expected (see the discussion in the
previous subsection). This follows from the formula
$$
\int_\Gamma \left( \frac{1}{x-a_1} + \frac{1}{x-a_2} \right)
(x-a_1)^{-\theta} (x-a_2)^{-\theta} \; dx
$$
$$
= -\frac{1}{\theta}
\int_\Gamma d \left( (x-a_1)^{-\theta} (x-a_2)^{-\theta} \right) = 0.
$$

(2) For generic $k$, the first twisted
homology of $(\C \backslash \{ a_1,a_2 \})^2 \backslash \on{diag}$
with coefficients in the rank one local system appearing in formula
\eqref{2pt} is one-dimensional. There is a unique (up to a scalar)
cycle of integration $\Gamma$, which generates this homology group.

Indeed, note that in this case the monodromies around $a_1$ and $a_2$
are the same: $e^{-2\pi i \theta}$. The cycle $\Gamma$ can be chosen
as follows: starting at some point $z$ and going counterclockwise
around $a_1$, coming back to $z$ and then going clockwise around $a_2$
and returning to $z$. When we apply the differential of the standard
twisted homology complex to this contour, the first of the two
contours gives the point $z$ multiplied by $(1-e^{-2\pi i \theta})$,
and the second one gives minus the same expression, so they cancel
each other. As explained in \cite{Varchenko1} in the general case, the
action of the differential can be identified with the action of the
generator $e$ of $U_{q'}(sl_2)$. In this case, it is the action on the
weight $0$ subspace of $\C^2_{q'} \otimes \C^2_{q'}$, where
$\C^2_{q'}$ is the two-dimensional irreducible representation of
$U_{q'}(sl_2)$ (and in general, with the action of a sum of the
generators $e_i$ of the quantum group, acting from the given weight
space to the weight spaces corresponding to the shift of the weight by
$\al_i$). This is why one can identify the homology group with the
space of highest weight vectors.

\medskip

Now, let's see whether we can get a conformal block for the Virasoro
algebra ${\mc W}_\beta(\fg)$ by pairing the above solution with a
covector $\langle W|$.
Set
$$
\langle W| = (a_1-a_2)^{-1/2} ((fv \otimes v)^* - (v \otimes fv)^*).
$$
Applying this functional to the conformal block \eqref{2pt} and using
formula
$$
\frac{1}{x-a_1} - \frac{1}{x-a_2} = \frac{a_1-a_2}{(x-a_1)(x-a_2)},
$$
we get
$$
(a_1-a_2)^{1/2(\theta+1)} \int_\Gamma (x-a_1)^{-(\theta+1)}
(x-a_2)^{-(\theta+1)} \; dx,
$$
which is a conformal block of the Virasoro algebra with the parameter
$$
\beta = \theta + 1 = \frac{1}{k+2} + 1.
$$
Here we recognize the shift of $\beta$ by $1$, as in formula
\eqref{betak1}.

\bigskip

Let us now insert a vertex operator at the point $0$ with generic
(non-integral) highest weight $\lambda$ while inserting a vertex
operator with lowest weight $-\lambda$ at $\infty$. Then we again
obtain conformal blocks with values in the weight 0 subspace of $\C^2
\otimes \C^2$, but now the highest weight condition is dropped. As the
result, the formula for the conformal block becomes
$$
(a_1-a_2)^{\theta/2} \int_\Gamma x^{-\lambda\theta} \left( \frac{v
    \otimes fv}{x-a_2} +
  \frac{fv \otimes v}{x-a_1} \right) (x-a_1)^{-\theta}
(x-a_2)^{-\theta} \; dx.
$$
As in the general formula \eqref{intsol1}, there is an extra factor
$x^{-\lambda\theta}$.

The cycle $\Gamma$ is now in the first homology group of $(\C
\backslash \{ 0,a_1,a_2 \})^2 \backslash \on{diag}$ which is
two-dimensional and can be identified with the weight $0$ subspace of
$\C^2_{q'} \otimes \C^2_{q'}$. The corresponding integrals span the
two-dimensional weight $0$ subspace of $\C^2 \otimes \C^2$.

When we take the pairing with $\langle W|$, we obtain
\begin{equation}    \label{virblock2}
(a_1-a_2)^{(\theta+1)/2} \int_\Gamma x^{-\lambda\theta} (x-a_1)^{-(\theta+1)}
(x-a_2)^{-(\theta+1)} \; dx.
\end{equation}
This as a conformal block of the Virasoro algebra with $\beta =
\theta+1$ and momentum $\mu = \lambda \theta$ at the point $0$.
% (we are
%not rescaling the momentum $\mu$ by $\beta$ here, because for the
%vertex operator at $0$ we take the definition of the momentum that is
%$T$-dual to the parameters of the vertex operators inserted at the
%points $a_i$).

%In general, there will be a similar rescaling $\la \mapsto \mu =
%\la/{}^L\neg (k+h^\vee)$ of the momentum of the vertex operator
%inserted at $0$ when we go from the Kac--Moody conformal blocks to the
%${\mc W}$-algebra conformal blocks.

\bigskip

Let us generalize the above example to the case of $n$ points
$a_1,\ldots,a_n$ with the insertion of the vertex operators
corresponding to the two-dimensional representation $\C^2$. We will
focus on the case of weight $2n-2$ subspace, which corresponds to the
case of a single screening operator.

The analogue of formula \eqref{2pt} is
\begin{equation}    \label{npt}
\prod_{i<j} (a_i-a_j)^{\theta/2} \int_\Gamma \sum_{i=1}^n \;
  \frac{v \otimes \ldots \otimes \underset{i}{fv} \otimes \ldots
    \otimes v}{x-a_i} \; \; \prod_{i=1}^n (x-a_i)^{-\theta} \; dx.
\end{equation}
The cycle $\Gamma$ is an element of the first twisted homology group
of $$(\C \backslash \{ a_1,\ldots,a_n \})^n \backslash \on{diag},$$
which is $(n-1)$-dimensional in this case (and can be identified with
the space of highest weight vectors in the weight $2n-2$ subspace of
$(\C^2_{q'})^{\otimes n}$).

The corresponding covector $\langle W|$ is given by the formula
$$
\langle W| = \prod_{i<j}(a_i-a_j)^{1/2} \sum_{i=1}^n \frac{(v \otimes
  \ldots \otimes \underset{i}{fv} \otimes \ldots \otimes
  v)^*}{\prod_{j \neq i} (a_i-a_j)}.
$$
Taking the pairing of $\langle W|$ and the Kac--Moody conformal block
\eqref{npt} and using the formula
$$
\sum_{i=1}^n \frac{1}{(w-a_i) \prod_{j \neq i} (a_i-a_j)} =
\frac{1}{\prod_{i=1}^n (w-a_i)},
$$
we obtain the Virasoro conformal block
\begin{equation}    \label{virblockn}
\prod_{i<j} (a_i-a_j)^{\beta/2} \int_\Gamma \prod_{i=1}^n
(x-a_i)^{-\beta} \; dx,
\end{equation}
where $\beta=\theta+1$.

For general highest weights $\la_1,\ldots,\la_n$ and multiple
screening operators, the explicit formula for the covector $\langle
W|$ becomes increasingly complicated. However, our general results
about the identification of the deformed conformal blocks guarantee
that such a covector always exists and pairing it with a
conformal block for $\Lfgh$ of level $^L\neg  k$, we obtain the
corresponding ${\mc W}_\beta(\fg)$-conformal block provided that the
parameters are related by formula \eqref{betak1}. This yields an
explicit identification stated in Conjecture \ref{isomblocks}.

\begin{remark}
  In the case $\fg=sl_2$, the Kac--Moody conformal blocks have been
  connected to the Virasoro conformal blocks, if the parameters are
  related by the formula $\beta = 1/(k+2)$, by two different changes
  of variables. In both cases, each representation of $sl_2$ (the
  $^L\neg \rho_{\la_i}$ in the notation of the previous subsection) is
  realized in the space of polynomials in one variable $x_i$, viewed
  the coordinate on the big cell of the flag manifold of $SL_2$. In
  the first approach, these variables $x_i$ are identified with the
  positions $a_i$ of the vertex operators \cite{FGPP,PRY}. In the
  second approach, the change of variables is obtained by deforming
  Sklyanin's separation of variables in the $SL_2$ Gaudin model
  \cite{Stoyan,JT,FGT}. In this case, the Fourier dual variables
  to the $x_i$ appearing on the Kac--Moody side are converted, on the
  Virasoro side, into positions of additional degenerate fields of
  type $\Phi_{1,2}$. It is unknown at present how to generalize these
  changes of variables to the case of arbitrary affine Kac--Moody
  algebras.

  In contrast, here we do not introduce any additional degrees of
  freedom. Rather, as a consequence of our general results on the
  identification of the $q$-{\em deformed} conformal blocks, we obtain
  that there exists a covector $\langle W|$ on the tensor product
  $\otimes_i {}^L\neg \rho_{\la_i}$ of finite-dimensional representations
  of $^L\neg \fg$, such that when we couple it with the corresponding
  $\Lfgh$ Kac--Moody blocks at level $^L\neg k$, we obtain conformal blocks
  of the ${\mc W}$-algebra ${\mc W}_\beta(\fg)$, if $\beta =
  1/{}^L(k+h^\vee) + m$. This provides an explicit identification of
  the two types of conformal blocks.
\end{remark}

\section{Quivers from String Theory}\label{sec:six}
\subsection{3d quiver gauge theory}
The quiver ${\scQ}$ from Section \ref{s_Nakajima} labels a gauge theory in three dimensions with ${\cal N}=4$ supersymmetry.
The ranks of vector spaces $V_a, W_a$ attached to the $a$-th node of the quiver ${\scQ}$ are the ranks of gauge $G_{\scQ}$ and global symmetry groups $G_W$: 
\beq\label{GQW}
G_{\scQ}=\prod_{a}U(d_a), \qquad G_W = \prod_{b}U(m_a).
\eeq
The arrows of the quiver encode the representation in which the matter fields transform. For every pair $a,b$ of nodes connected by a link of the Dynkin diagram we get a hypermultiplet transforming  in bifundamental representation $(d_a, \overline{d_b})$ of under $U(d_a)\times U(d_b)$. There are also $m_a$ hypermultiplets in fundamental representation $d_a$ of the $U(d_a)$ gauge group.  

\subsubsection{}
The Nakajima quiver variety $X$ is the Higgs branch of the gauge
theory. The Kahler parameters of $X$ encode Fayet-Illiopolous (FI)
terms in the gauge theory. The equivariant parameters are the real
masses, induced by weakly gauging $G_W$ symmetry. Both the FI terms
and the real masses get complexified once we compactify the gauge
theory on $S^1$, as we will shortly do.\footnote{There are additional
  parameters needed to define the theory, such as the gauge couplings,
  which are not relevant for us, as they do not affect the partition function.} 
The ${\mathbb C}^\times_{\hbar}$ action that scales the symplectic form on $X$ comes from a $U(1)$ subgroup of $SU(2)_H\times SU(2)_V$ $R$-symmetry group. %($SU(2)_H$ acts on the hypermultiplet scalars, and $SU(2)_V$ on vector multiplets. The $U(1)$ subgroup needed will be spelled out below.)
% R-symmetries are symmetries of the theory that act on the supercharges. For generic $\hbar$, supersymmetry is reduced to ${\cal N}=2$ in three dimensions.), generated by those supercharges that are neutral under this $U(1)$.)

\subsection{Quiver gauge theory from IIB string}

The quiver gauge theory with quiver ${\scQ}$ arises on D3 branes in IIB string theory compactified on 
$$Y \times M_6.
$$
Here, $Y$ is an ADE surface, a resolution of ${\mathbb C}^2/\Gamma_{\fg}$ singularity, where ${\Gamma}_{\fg} $ is a 
discrete group of $SU(2)$ related to $\fg$ by McKay correspondence; $M_6 = {\cal C} \times {\mathbb C}\times {\mathbb C}$ is the six-manifold in \eqref{6man}. The Riemann surface ${\cal C}$ is the same one we used to define the $q$-conformal blocks in Section \ref{s_statement}. 

\subsubsection{}
The ranks of the vector spaces $V_a, W_a$ are determined by the homology classes of 2-cycles in $Y$ that the D3 branes wrap.

Recall the relation of geometry of $Y$ to  representation theory of ${\fg}$:  The vanishing cycles of the ADE singularity are topologically $S^2$'s which intersect according to the Dynkin diagram of $\fg$. Denote the vanishing cycles by $S_a$; their homology classes are the positive simple roots of ${\fg}$, $e_a=[S_a]$. They span $H_2(Y, \mathbb Z)$, which can be identified with the root lattice of $\fg$ (with the norm coming from the intersection form on $Y$.) The weight lattice of $\fg$ is the same as the relative homology group $H_2(Y, \partial Y; {\mathbb Z})$. The latter is spanned by a collection of noncompact cycles $S_a^*$ whose homology classes are the fundamental weights, $w_a =[S_a^*]$. (A cycle in the class of $S_a^*$ is the fiber of the cotangent bundle at a generic point on $S_a$.)

To get the quiver ${\scQ}$, we take a collection of non-compact D3 branes in class $[S^*] \in H_{2}(Y, \partial Y; {\mathbb Z})$, where
\beq\label{NC}[S^*] = \sum_a m_a [S_a^*] , \qquad [S] = \sum_a d_a [S_a]
\eeq
together with a collection of compact D3 branes in the class $[S]\in H_{2}(Y, {\mathbb Z})$. 
In addition to their support in $Y$, the D3 branes are distributed at a collection of points on ${\cal C}$, and on the complex plane in $M_6$, associated with $q$ as the equivariant parameter.

\subsubsection{}

The D3 branes on the compact cycles in homology class $[S]$ in \eqref{NC} support $G_{\scQ}$ gauge fields in  \eqref{GQW}. 
The hypermultiplets in $(d_a, {\overline d_b})$ arise from (zero-modes of) strings at the intersections of cycles in classes $[S_a]$ and $[S_b]$, for $a\neq b$. The intersection number $\#(S_a, S_b) = I_{ab}$ is identified with the incidence matrix $I_{ab}$. Strings at the intersections between $S_a$ and $S_a^*$ cycles give rise to hypermultiplets in $(d_a, {\overline m_a})$ representation of $U(d_a)\times U(m_a)$. The flavor symmetry $G_W$ in \eqref{GQW} is the gauge group of non-compact D3 branes on $[S^*]$ in \eqref{NC}; due to non-compactness, the corresponding gauge fields are frozen. 

\subsubsection{}

D3 branes give rise to a three dimensional gauge theory on $S_{R'}^1\times {\mathbb C}$. The circle $S_{R'}^1$ is not geometric in IIB. It arises due to a stringy effect. 

The D3 branes are located at points on ${\cal C}$, which is a cylinder ${\cal C} = {\mathbb R} \times S^1_{R}$, with a circle of radius $R$. Due to strings which wind around the $S_{R}^1$, there are the infinitely particles in the theory on $\C$. They are labeled by the winding modes on $S^1_R$, which are in turn equivalent to momentum modes on another circle $S^1_{R'}$, with radius $R' = 1/(m_s^2 R)$. 

The three dimensional nature of the theory can be made manifest by $T$-duality. The duality relates IIB on ${S}^1_R$ with IIA on $S^1_{R'}$, and D3 branes in IIB at points on ${S}^1_R$ with D4 branes in IIA wrapping the $S^1_{R'}$; these theories are the same. The winding on ${S}^1_{R}$ corresponds to momentum on the $S^1_{R'}$.

\subsubsection{}
%The positions of compact D3 branes on ${\cal C}$ are the Chern-roots of the vector bundles in \eqref{}; 
The positions of the non-compact D3 branes  on ${\cal C}$ are the $A$-equivariant parameters and the complexified real masses: a D3 brane supported at a point $x= a_i$ on ${\cal C}$ leads to an equivariant parameter with same name. This is also an insertion point of a vertex operator in \eqref{electric} and \eqref{magnetic}.
The positions of compact D3 branes on ${\cal C}$ are dynamical parameters; they are the insertion points of screening charge operators in \eqref{magnetic}.
The Kahler moduli of $X$ are identified with the Kahler moduli of $Y,$ as both correspond to FI parameters in the 3d gauge theory. They determine the weights $\lambda$ in \eqref{electric} and \eqref{magnetic}.

\subsection{Little string theory from IIB string}
The ten-dimensional IIB string on $Y\times M_6$ has many more degrees of freedom than we presently need. There is a smaller theory, which captures the physics relevant for us. It is a six dimensional string theory, ``the little string theory'' with $(2,0)$ supersymmetry on $M_6$. 
\subsubsection{}

The ${\fg}$-type little string theory with $(2,0)$ supersymmetry is
defined as the limit of IIB string theory on $Y\times M_6$.  
The limit corresponds to taking the string coupling $g_s$ to zero, keeping fixed the characteristic mass $m_s$ of the IIB string, $m_s$ and the moduli of the 6d $(2,0)$ theory. (The moduli come from periods of five 2-forms in IIB string compactified on $Y \times M_6$, coming from the triplet of self-dual two-forms of $Y$ and  the two $B$-fields of IIB string, with appropriate normalizations.)
\subsubsection{}
The theory one is left with is a string theory on $M_6$: it contains strings whose tension is $m_s^2$, which are inherited from IIB strings. 
One reflection of the fact one gets a string theory, and not a point particle theory, is that the little string theory has a T-duality symmetry. T-duality relates the ${\fg}$-type $(2,0)$ little string, compactified on a circle of radius ${R}$, with the ${\fg}$-type $(1,1)$ little string theory on a circle of radius $R'=1/(m_s^2 {R})$. (The latter is obtained from IIA string on $Y$, in an analogous $g_s$ to zero limit.) The two string theories are equivalent.

\subsubsection{}
The D3 branes of IIB string on $Y$ give rise to codimension four defects of the little string theory on $M_6$. The theory on the defect D3 branes is the quiver gauge theory with quiver ${\scQ}$. The limit which reduces the 10d IIB string to the 6d little string on $M_6$ does not affect the gauge theory on D3 branes at all. The triplet of FI parameters of the 3d gauge theory, for example, is given by $R$ times the moduli of the little string, coming from the triplet of self-dual three-forms on $Y$. The gauge couplings are $R$ times the modulus originating from the NS B-field. Here, $R$ is the radius of the $S^1_R$ in ${\cal C}$. All these remain finite in the limit, since we are keeping both $R$ and  the moduli of $(2,0)$ little string fixed as we take $g_s$ to zero. (See \cite{AH} for more details.)

\subsection{Non-simply laced case}

To get non-simply laced theories, we make use of the fact that every non-simply laced Lie algebra ${\fg}$ arises as a subalgebra of a simply laced Lie algebra ${\fg}_0$, invariant under the outer automorphism group $H$ of ${\fg}_0$. Outer automorphisms of ${\fg}_0$ correspond to automorphisms its Dynkin diagram. 

\subsubsection{}

We start with IIB string on $Y_0$ an ADE singularity corresponding to ${\fg}_0$. We take $Y_0$ to be fibered over $M_6$ in such a way\footnote{In \cite{KimPest}, the twist is around the ${S}^1$ in ${\cal C}$ instead, or more precisely, around its T-dual circle.} that, as we go around the origin of the complex plane in $ M_6$ that supports the D3 branes, $Y_0$ comes back to itself only up to the action of a generator $h\in H$. 
The action of $h$ on $Y_0$ is by permuting the 2-cycles classes in $H_2(Y_0, {\mathbb  Z})$ in a way compatible with the action of $h$ on the root lattice of ${\fg}_0$, and the identification of the latter with $H_2(Y_0, {\mathbb  Z})$ (to our knowledge, this string construction was first used in \cite{Bershadsky}).
\subsubsection{}

The automorphism groups $H$ are all abelian, $H={\mathbb Z}_m$, generated by a single element $h\in H$, with $h^m=1$. The roots of ${\fg}$ are the combinations of roots of ${\fg}_0$ which are invariant under $H$. This way, from $({\fg}_0, H)$ one gets $ {\fg}$ with:
\begin{align}\label{folding}
(A_{2n-1}, Z_2)&\rightarrow C_n\notag\\
(D_{n+1},Z_2)& \rightarrow B_n \notag\\
(D_{4},Z_3)& \rightarrow G_2\notag\\
(E_{6},Z_2)& \rightarrow F_4\notag\\
\end{align}
The root lattice of ${\fg}$ is obtained from the root lattice of ${\fg}_0$ as follows. A simple positive root of ${\fg}$ is a sum over the simple positive roots of ${\fg}_0$ which are in a single orbit of $H$, normalized by the length of the orbit. The short roots of ${\fg}$ come from the simple roots in ${\fg}_0$ which lie in orbits of $H$ of length $m$. The long roots of ${\fg}$ are the simple roots of ${\fg}_0$ invariant under $H$. The length of the root is defined by $(e_a, e_a)$, where $(,)$ comes from the inner product on the root lattice of ${\fg}_0$. Since all the roots of ${\fg}_0$ have length $2$, the length of a short simple root of ${\fg}$ is $2/m$, and the length of a long root is $2$.  The coroots of ${\fg}$ are related to the roots of ${\fg}$ in the usual way $e^{\vee}_{a} =2 e_a/(e_a, e_a)$. It is easy to show that the result is the Cartan matrix of ${\fg}$: $C_{ab} = (e^{\vee}_a , e_b)$. 

\subsubsection{}

%The group $H$ acts on $Y_0$ by permuting its 2-cycles, in a way compatible with the identification of  $H_2(Y_0, {\mathbb Z})$ with the root lattice of ${\fg}_0$. This is implemented by making the moduli of $Y_0$ have monodromy corresponding to $h$ action
% as one goes around the origin of the copy of ${\mathbb C}$ where we place the twist. 
The action of $H$ on $Y_0$ translates into the action on D3 branes supported on the 2-cycles in \eqref{NC}, and on the quiver ${\scQ}_0$ that describes them. The D3 brane configurations that are allowed in the fibered geometry are in one to one correspondence with the configurations of 2-cycles on $Y_0$ which are invariant under the $H$ action: The D3 branes we are considering are supported on 2-cycles in $Y_0$ times the complex plane ${\mathbb C} \in M_6$ where the twist is; any $H$-invariant configuration in $Y_0$ gives rise to a configuration on the fibered product which comes back to itself up to the $h$-twist acting simultaneously on $M_6$ and on $Y_0$.

\begin{comment}
\subsubsection{}

For $H$ to leave the quiver ${\scQ}_0$ invariant, the ranks of vector spaces $(V_a, W_a)$ in a single orbit of $H$ have to be the same. From this it follows that the non-compact D3 branes, corresponding to $W_a$'s, are labeled by fundamental coweights of ${\fg}$, 
and that the compact D3 branes (corresponding to $V_a$'s) are labeled with the simple coroots of ${\fg}$.

D3-brane charges and the ranks of the vector spaces $V_a$ and $W_a$, are labeled by coroots and coweights, respectively \cite{NP}. The fundamental coweights and the simple coroots of ${\fg}$ are simply the sums of the fundamental coweights and the simple coroots of ${\fg}_0$ lying in a single orbit of $H$. This follows easily from the relation of the root lattices of ${\fg}$ and ${\fg}_0$, and the definitions of the coroots and coweights. (The coweight lattice is the lattice dual to the root lattice.) 

It is the representations of $^L{\fg}$ label the $H$-invariant D3 brane configurations, since the coroots and coweights of ${\fg}$, one recalls, correspond to the roots and the weights of the Langlands dual Lie algebra $^L {\fg}$. 

\end{comment}

\subsubsection{}

For $H$ to leave the quiver ${\scQ}_0$ invariant, the ranks of vector spaces $(V_a, W_a)$ associated to the nodes of the Dynkin diagram of ${\fg}_0$ which lie in a single orbit of $H$ have to be the same. From this it follows that the non-compact D3 branes, corresponding to $W_a$'s, are labeled by fundamental weights of $^L{\fg}$, the Lie algebra Langlands dual to ${\fg}$. Similarly, the compact D3 branes, corresponding to $V_a$'s, are labeled with the simple roots of $^L{\fg}$.

%D3-brane charges and the ranks of the vector spaces $V_a$ and $W_a$, are labeled by coroots and coweights, respectively \cite{NP}. 
To see this, one recalls that the simple roots and the fundamental weights of $^L {\fg}$ coincide with the simple coroots and fundamental coweights of ${\fg}$, respectively. The latter are, in turn, simply 
%t fundamental coweights and the simple coroots of ${\fg}$ are simply 
the sums of the fundamental coweights and the simple coroots of ${\fg}_0$ lying in a single orbit of $H$. These are exactly the data labeling the $H$-invariant quivers ${\scQ}_0$. (For the former statement, one merely needs to recall the relation of the root lattices of ${\fg}$ and ${\fg}_0$, and the definitions of the coroots and coweights. The coweight lattice is the lattice dual to the root lattice).

\subsubsection{}\label{sec:mtimes}
The fields of the quiver gauge theory on the D3 branes are a subset of those of the original ${\scQ}_0$ theory which are compatible with folding by $H$.
Let  $z$ be the complex coordinate on the ${\mathbb C}$-plane that supports the D3 branes, and $\phi(z)$ a field of the ${\scQ}_0$ quiver gauge theory.  The fields must obey 
\beq\label{eqf}
\phi(e^{2\pi i}z) = h\cdot \phi(z),
\eeq 
where $h\cdot \phi$ denotes the image of $\phi$ under the $h$ action on the quiver.
%
%The generator $h$ of $H$ acts by permuting ${\fg_0}$ gauge group factors, and the matter fields according to the geometric action on the nodes of the Dynkin diagram, and a
%$2\pi$ rotation of $z$. 
%In going around the disk, the fields come back to themselves only up to $h$ action. 
The later action is trivial for fields that only involve the long roots, corresponding to nodes of the Dynkin diagram of ${\fg}_0$ which are invariant under $H$. For fields $\phi$ that involve the short roots, coming from fields which transform in orbits of $H$ of length $m$, the $H$ action organizes $\phi(z),h\cdot \phi(z), \ldots, h^{m-1} \cdot \phi(z)$ into a single field (equal to their sum), which is single valued only on the $m$-fold cover of the ${\mathbb C}$-plane. If $w$ is a coordinate on the cover, $z= w^m,$ fields coming from orbits of $H$ of length $m$ have integer mode expansion in terms of $w = z^{1/m}$, but fractional mode expansion in terms of $z$.

\subsubsection{}

Langlands duality exchanges ${\fg}$ and $^L{\fg}$, and roots and coroots, while transposing the Cartan matrix. Since some define the Cartan matrix to be the transpose of ours, it is easy to mix-up ${\fg}$ and $^L{\fg}$. An unambiguous way to distinguish them is by the lengths of their roots. While the norm of the inner product $(,)$ is a matter of convention, the ratio of the lengths of the roots is not. For example, $B_n$ has one short root, and $(n-1)$ long ones, while $C_n$ has $(n-1)$ short roots, and one long one. ($B_n$ and $C_n$ are exchanged under Langlands duality, while $F_4$ and $G_2$ map to themselves).

\subsection{Conformal limit}

The $(2,0)$ little string is a string theory, containing strings whose characteristic size is $1/m_s$. It becomes a point particle theory, the conformal field theory in 6d with $(2,0)$ supersymmetry in the limit where one sends $m_s$ infinity,
$$
m_s\qquad \rightarrow \qquad \infty.
$$ 
We will call this theory theory $\scX$, for short. In the conformal limit, we want to keep the moduli of the $(2,0)$ theory fixed, since they become the moduli of theory $\scX$. We also want to keep fixed the Riemann surface it is compactified on, and the positions $x=a_i$ of D-branes on it.  

In the conformal limit, the gauge theory description of the defects is lost. The inverse gauge coupling of the defect 3d quiver theory is given by the modulus of the $(2,0)$ theory (which has dimensions of mass square) times $1/m_s^4$. Thus, in the $m_s$ to infinity limit, the gauge coupling becomes infinite. This means that there is no sense in which we can describe the theory on the defects as a gauge theory.

\section{Vertex Functions from Physics}\label{sec:seven}

%The  explain why the relation of 3d gauge theories to Langlands correspondence is predicted on physical grounds will require us to explain 

The vertex function of Nakajima quiver variety has two closely related, but distinct physics interpretations.
Most directly, they are partition functions of 3d quiver gauge theory from Sec. \ref{sec:six} with quiver $\scQ$, computed on $\C\times S^1_{R'}$. The gauge theory interpretation lets one can make direct contact with vertex functions both their defining formulation, in terms of counting quasimaps $\C \dasharrow X$, and in the integral form of Sec. \ref{sec:quasimap-integral}.

The more far reaching interpretation, however, is that they are also the partition functions of ${\fg}$-type $(2,0)$ little string theory on $M_6$, with codimension four-defects, where the quiver $\scQ$ captures data of the defect. This explains why vertex functions have implications for Langlands duality. We will return to this in Sec. \ref{sec:nine}. 

The two interpretations are related: the partition function of little string theory we need, turns out to {\it equal} the partition function of the theory its the defects. 
We will define the relevant partition functions, explain the mechanism between the equality of the bulk and the defect partition functions, and show how results of Sec. \ref{sec:quasimap-integral}. emerge from the 3d gauge theory perspective.
 %This is the 3d quiver gauge theory from Section \ref{sec:six} with quiver $\scQ$, computed on $\C\times S^1_{R'}$. We will show that this language makes direct contact with vertex functions of $X$, in both their defining formulation, and in the integral form of Section \ref{sec:quasimap-integral}.

\subsection{Little string partition function}
The partition function of the ${\fg}$-type $(2,0)$ little string on $M_6$ is most easily defined in the $T$-dual language, using $T$-duality with respect to the circle in ${\cal C}$.  

The dual of the $(2,0)$ string theory on $M_6$ is the $(1,1)$ little string on 
$${M'_6} = {\cal C}'\times {\mathbb C}\times {\mathbb C},
$$ 
where ${\cal C}' = S^1_{R'}\times {\mathbb R}$. The $(1,1)$ string theory is, at low energies, a six dimensional gauge theory with maximal supersymmetry, and gauge group based on the Lie algebra ${\fg}$. 
Its partition function on $M'_6$ is a supersymmetric index
\beq\label{Trace}
{\mathsf{Index}}={\rm Tr} (-1)^F {\bf g}.
\eeq 
The trace is the trace in going around the $S^1_{R'}$; $F$ is the fermion number so $(-1)^F$ counts bosons and fermions with signs. 
The insertion of ${\fg}$ in the trace has the effect of turning $M_6'$ into a twisted product: as we go around the $S^1_{R'}$, we rotate the two complex planes ${\mathbb C}\times {\mathbb C}$ by  $q$ and $t^{-1}$, respectively.
This is known as the ${\Omega}$-background, defined by Nekrasov and studied e.g. in \cite{LMN, NO1, NW} and in many other papers. 
\subsubsection{}\label{subsec:detail}Explicitly, ${\bf g}$ is the product of generators 
\beq\label{gtwist}
{\bf g} = q^{S-S_H} \times t^{S_H-S_V}.
\eeq
We denoted by $S$ the generator of the rotation of the $\C$-plane in $M'_6$ which is rotated by $q$. $S_V$ generates the action that rotates the second ${\mathbb C}$-plane by $t^{-1}$. $S_H$ is the generator of the $U(1)$ subgroup of $SU(2)$ $R$-symmetry group of the 6d theory. 
The $R$-symmetry twist is needed for the partition function to preserve supersymmetry. 

\subsection{Localization to defects} 
In the absence of defects, the partition function in \eqref{Trace} is trivial. In the presence of defects, it equals to the partition function of the theory on the defects. One simply ends up computing \eqref{Trace}, restricted to the modes on the defect.

\subsubsection{} Without any defects on $M_6'$, the insertion of ${\bf g}$ in \eqref{Trace} ends up commuting with four of the sixteen supercharges of the 6d theory. This is too many for the index to receive non-trivial contributions: the supersymmetries end up relating bosons and fermions in pairs and their contributions to the index cancel out. To get a non-trivial partition function one must reduce the supersymmetries by a half. We will achieve this by adding defect D-branes.\footnote{The relation of bulk and defect perspective is described in more detail in \cite{AHKS, AS, An, AH}.}

With defects present, supersymmetry is broken, but only near the defects \cite{Tasi}. Away from the defects, local physics is that of the 
$(1,1)$ little string, compactified on a circle, with all of its supersymmetries intact. This leads to localization: the only nontrivial contributions to the partition function can come from modes supported on the defects. Computing the trace restricted to such modes is the same as computing the partition function of the theory on the defect.
(The notion of localization used here is in its essence the same mechanism as in the more familiar applications of the term. The defect is fixed by a linear combination of the supersymmetries in the bulk. See \cite{wm} for more explanation.)
\subsubsection{} 
The defects we will use are the D3 brane configurations in Section \ref{sec:six}. The quiver ${\scQ}$ which encodes the data of the defects, as in previous section, also encodes the 3d quiver gauge theory on the defects.  T-duality maps D3-branes at points on ${\cal C}$ in $M_6$ to D4 branes winding around the $S^1_{R'}$ in ${\cal C}' = {\mathbb R} \times S_{R'}^1$, and at the same points in the radial direction.  The position of D3 branes on $S_{R}^1$ becomes the holonomies of the D4 brane gauge fields around $S^1_{R'}$. 
T-duality makes it manifest that the gauge theory on these D-branes is a three dimensional theory on $S_{R'}^1\times {C}$, where $C$ is identified with the complex plane in $M_6'$ supporting the defect; this is  the copy of $\C$ which is rotated by $q$.

(In addition to D3 brane defects, there are other kinds of defects which lead to the same localization effect. Adding D5 brane defects at points in ${\cal C}$ and filling ${\mathbb C}\times {\mathbb C}$, for example, will lead to Langlands correspondence with ramifications.)
\subsection{Defect partition function}

The index \eqref{Trace}, computed in the 3d quiver gauge theory on the defect, becomes the supertrace over the Hilbert space of the theory on $\C$.  The trace is around the $S^1_{R'}$ as before. The identification of $\C$ with the complex plane in $M'_6$ supporting the defect, determines the action of all the generators of ${\bf g}$ in the 3d gauge theory.

 \subsubsection{}

From the 3d gauge theory perspective, the interpretation of various factors in ${\bf g}$ is as follows. Let $\hbar = q/t$. Then, \eqref{gtwist} becomes

$$
{\bf g}  =  q^{S} \times \hbar^{-S_H} \times t^{-S_V}. 
$$
$S$ generates rotation of $C$, the copy of ${\mathbb C}$ that supports the defect. This is a geometric action from both the bulk and the defect perspective.
$S_V$ acts as a rotation of a complex plane transverse to the defect. It becomes an $R$ symmetry generator in the gauge theory. It corresponds to the $U(1)$ subgroup of $SU(2)_V$ $R$-symmetry that acts on scalars in vector multiplets. (A complex scalar in the vector multiplet is the position of the D-branes on ${\mathbb C}_t^{-1}$ plane.) $S_H$ generates the $U(1)$ subgroup of $SU(2)_H$ $R$-symmetry group acting on hypermultiplet scalars;  it generates an $R-$symmetry both in the bulk and on the defect.

There are factors in ${\bf g}$ we have refrained from writing out explicitly, to keep the formulas simpler.
The remaining part of parameters come from global $U(1)$ symmetries of the 3d ${\cal N}=4$ gauge theory. They enter ${\bf g}$ as  the (complexified) holonomies of the corresponding gauge fields around the $S^1_{R'}$. They are associated with the 
$$ 
\mathsf{T}\times \mathsf{A}^{\vee} \times{\mathbb C}^\times_{q}, \qquad \mathsf{T} = \mathsf{A}\times {\mathbb C}^{\times}_{\hbar}
$$ 
symmetry of the theory. The symmetries in $\mathsf{T}$ are associated to real mass parameters; $\mathsf{A}^{\vee}$ are associated to the real FI parameters. (The parameters in $\mathsf{A}$ preserve ${\cal N}=4$ supersymmetry, those in $\mathsf{T}$ but not in $\mathsf{A}$ break it to ${\cal N}=2$.).

\subsection{Index for non-simply laced ${\fg}$}
In non-simply laced cases, there is an ${H}$-twist around  the complex plane that supports the defect. The trace in \eqref{Trace} is the trace over states invariant under $H$ (they correspond to fields obeying \eqref{eqf}). The generator $S$ of rotations of the plane supporting the defects now has eigenvalues that are integer, and half integer, multiples of $1/m$, where $m$ is the order of $H$. This is because some of the modes come back to themselves only upon going around the circle $m$ or $2m$ times, see Section \ref{sec:mtimes}. 
We prefer that only integer and half integer powers of $q$ appear in the partition function; to achieve this we will replace $q$ by $q^m$, and define ${\bf g}$ in \eqref{gtwist} as:

\beq\label{gtwistm}
{\bf g}  = q^{mS} \times \hbar^{-S_H} \times t^{-S_V}, 
\eeq
where now
\beq\label{relation}
\hbar  = q^m/t
\eeq
for the index to preserve supersymmetry. (The action of $S$, $S_H$ and $S_V$ on the supersymmetry generators is independent of global identifications we make, so it is not sensitive to folding by $H$.) This is the string origin of the identification of parameters  in \eqref{first}.

\subsection{Vertex functions from 3d gauge theory} \label{subsubsec:fp}

The index in \eqref{Trace}, computed in the 3d ${\cal N}=4$ gauge theory on ${C} \times S^1$ based on the quiver ${\scQ}$, is the vertex function ${\bf V}$  of $X$ from \eqref{tVx}.%

\beq\label{equality}
\mathsf{ Index} = {\rm Tr} (-1)^F {\bf g} = {\bf V}
\eeq
The $\mathsf{Index}$ is not a function -- it is a vector instead, because it is defined in the 3d gauge theory on $S^1 \times {C}$, and thus depends on the choice of the vacuum of the gauge theory at infinity in $C = \C$. We will show momentarily that the vector space it takes values in can be identified with $K_{\mathsf{T}}(X)$.

For a non-simply laced Lie algebra $\fg$, the meaning of the $\mathsf{Index}$ is different. It is the vertex function ${\bf V}^H$ of $X_0$, restricted to $H$-invariant modes.

The relation between the partition function of 3d gauge theory on ${\C}\times S^1$ and vertex functions of quantum K-theory of its Higgs branch are well known \cite{Nekrasov, NO2}. The integral representation of vertex functions, which we proved in Section \ref{sec:quasimap-integral} are also known in the physics literature, see for example \cite{beem}. We will briefly review the physics perspective on these.
\subsubsection{}
The 6d little string Hilbert space effectively localizes to the Hilbert space of the 3d gauge theory on $C=\C$, but even that is much larger than the space of configurations that end up contributing to \eqref{Trace}. The index receives contributions only from configurations that are annihilated by the pair of supersymmetry generators ${\overline Q}$, ${\overline Q}^{\dagger}$, which anti-commute with $(-1)^F {\bf g}$; all others come in pairs related by actions of these generators, and cancel out from the index. The field configurations which preserve the supersymmetries are ``quasi-maps'' from $C$ to $X$. The quasi-maps are simply the solutions to vortex equations on $C$ \cite{Wp, Nekrasov}. In the adiabatic approximation, the supersymmetric path integral of the three dimensional theory on $R\times C$ (with $R$ viewed as time direction) localizes to the supersymmetric quantum mechanics on the moduli space ${\cal M}=\QM_{nonsing}(X)$ of quasimaps to $X$, see \cite{Nekrasov, NO2}. The quasi-maps are non-singular at infinity of $C$: this corresponds to working with boundary conditions which require the gauge field strength to vanish there. In addition,  finite energy configurations require one to restrict the matter fields to a approach a vacuum at infinity. In a theory deformed by masses, i.e. working equivariantly with respect to $\mathsf{T}$, the latter corresponds to fixed point of $\mathsf{T}$-action on $X$. 

In supersymmetric quantum mechanics with a pair of supercharges, the partition function ${\rm Tr} (-1)^F$ computes the index of the Dirac operator on ${\cal M}$. In the present case, the supersymmetric quantum mechanics has twice as many supersymmetries: there are in fact two more supercharges $Q, Q^{\dagger}$ that annihilate the solutions in ${\cal M}$, they just fail to commute with ${\bf g}$ for generic $\hbar$. The supercharges ${Q}$, ${Q}^{\dagger}$ and ${\overline Q}$, ${\overline Q}^{\dagger}$ are identified with Dolbeault operators $\partial,  \partial^{\dagger},  {\overline \partial}, {\overline \partial}^{\dagger}$ acting on differential forms on ${\cal M}$. The index of ${\slashed {D}} =    {\overline \partial}+ {\overline \partial}^{\dagger}$  operator on ${\cal M}$ is the holomorphic Euler characteristic 
of the symmetrized virtual structure sheaf $\hat{\cal O}_{\rm vir}$ of ${\cal M}$ in \eqref{def_tO}, see \cite{OK} and also \cite{NO2}.
The $\hat{\cal O}_{\rm vir}$ bundle is the cohomology of the complex generated by the broken supersymmetries $Q \sim  \partial$ and $Q^{\dagger} \sim  \partial^{\dagger}$ acting on differential forms on ${\cal M}$, obtained by quantizing the collective coordinates of fermions. The Kahler variables of $X$ come from the (complexified) real FI parameters in the 3d gauge theory; they lead to grading of quasi-map moduli space by the degree.

In practice, we like to think about indices as functions of their parameters, so we want to extract a particular component of the vector ${\bf V}$. This corresponds to picking a specific vacuum state at infinity. The vacua lie on the $\mathsf{T}$-fixed locus in $X$; if fixed points $p\in X_\mathsf{T}$ are isolated, it suffices restrict ${\cal M}$ to the moduli space of maps ${\cal M}_p$ approaching $p\in X_\mathsf{T}$ at infinity. In that case, $K_\mathsf{T}(X)$ is spanned by classes of fixed points ${\cal O}_p$. A class in $K_T(X)$ labels the choice of a vacuum state even in more general situations. (More naturally, the supersymmetric vacua are ground states of an effective supersymmetric quantum mechanics which arises in studying the 3d gauge theory on $R\times T^2$, with $T^2$ of complex structure parameter $q$, with equivariant/mass deformations turned on corresponding to parameters in $\mathsf{T}$. In this setting, the ground states should be labeled by elements of $\Ell_\mathsf{T}(X)$, the equivariant elliptic cohomology of $X$.  For Nakajima varieties, the ranks of $\Ell_\mathsf{T}(X)$ and $K_\mathsf{T}(X)$ turn out to be the same, so we will use the latter to label the vacua.)
\subsubsection{}\label{subsubsec:sp}

The second way to compute \eqref{Trace}, which leads to integral formulas, is simpler in many respects.

Since $q^{S}$ factor in ${\fg}$ regularizes the non-compactness of $C$, one can treat the three dimensional gauge theory on $S^1 \times C$ as a (gauged) supersymmetric quantum mechanics on the $S^1$, with discrete spectrum. 
The computation becomes as elementary exercise in quantum mechanics (see \cite{Nekrasov} for more detail): enumerating the fields in the 3d theory, decomposing each field into modes on $C$ of fixed momentum, and evaluating their contribution to the trace.  For non-simply laced Lie algebras one includes in the trace only the $H$-invariant configurations, obeying  \eqref{eqf}. 

It is easiest to start by treating $G_{\scQ}$ a global symmetry; gauging it corresponds to projecting to $G_{\scQ}$ invariant states, which one can do in the end. In addition, it is helpful to abelianize the theory, breaking the gauge group $G_{\scQ}$ to its maximal abelian subgroup. Then, at the outset, the partition function depends on equivariant parameters associates with maximal torus of $G_{\scQ}$.  These we denoted by $x$'s elsewhere (and by $s$ in the appendix and in Section \ref{sec:quasimap-integral}) since they come from positions of compact D3 branes on ${\cal C}$. They are also (part of) the Coulomb branch moduli of the 3d gauge theory, so this computes the partition function from the Coulomb-branch perspective.  The index in \eqref{Trace} depends on Kahler moduli of $X$ via the classical FI terms in the Lagrangian. In the end, since $G_{\scQ}$ is gauged one integrates over the $x$'s. The contour is chosen to project to states which are neutral. This means integrating over
\beq\label{scont}
\int_{|x|=1} \dots \, d_\textup{Haar}x \,
\eeq
as in Sections \ref{s_GIT} and \ref{sec:quasimap-integral}, where $d_\textup{Haar}x =\prod_{a, \alpha} dx_{a, \alpha}/x_{a, \alpha}$ and the contour is chosen to pick out contributions independent of $x$'s. Depending on the values of FI parameters, one gets to deform the contour, picking up the residues in the process. This is the GIT quotient from Section \ref{sec:quasimap-integral}.

The contribution to \eqref{Trace} of vector multiplets from the $a$-th node of the Dynkin diagram of ${\bf g}$, is
\beq\label{gaugsl}
\prod_{\alpha \neq \alpha'}{ \varphi_{q_a}(  x_{\alpha ,a}/x_{\alpha',a})\over \varphi_{q_a}(   t \,x_{\alpha, a}/x_{\alpha',a})} \prod_{\alpha < \alpha'} { \theta_{q_a}( t x_{\alpha ,a}/x_{\alpha',a})\over \theta_{q_a}(    \,x_{\alpha, a}/x_{\alpha',a})}.
\eeq
see \cite{Nekrasov} for derivation. 
The one new aspect is the dependence, in non-simply laced cases, on whether $"a"$ labels a short, or a long root.  
Recall that a node corresponding to a short root of ${\fg}$ collects contributions $m$ nodes of ${\fg}_0$ which are in a single orbit of $H$.  The corresponding field configurations come back to themselves only after going around the origin of the ${\mathbb C}$-plane $m$ times. By contrast, a node corresponding to a long root of ${\fg}$ comes from node of ${\fg}_0$ which comes back to itself going around once. Since $q$ keeps track of the minimum momentum on the disk, so that only (half-)integer powers of $q$ enter the partition function, then for $"a"$ a short root $q_a=q$, and for a long root, $q_a = q^m$. This coincides with the ${\cal W}_{q,t}({\fg})$ algebra contributions from screening currents associated to a single node in \eqref{sctdasl1}.
Similarly, hypermultiplets connecting a pair of distinct nodes $a,b$ in the Dynkin diagram of ${\fg}$ contribute:
\beq\label{lowsl}
\prod_{\alpha, \beta} {{ \varphi_{q_{ab}}( t v_{ab} \,x_{\alpha,a}/x_{\beta,b})\over \varphi_{q_{ab}}(v _{ab}\,{x_{\alpha ,a} /x_{\beta,b}})}}
\eeq
where  $v_{ab} =\sqrt{q_{ab}/t}$.
%$\alpha, \beta$ run from $1$ to $d_a, d_b$, respectively.  
If either of the nodes $a,b$ is short, $q_{ab} = q$  since then the fields that contribute are single valued only on the $m$-fold cover of the disk. If both of the nodes are long then $q_{ab}=q^m$. This coincides with the two-point functions of screening currents associated to the distinct pair of nodes $a, b$, in \eqref{sctdasl2}. 
Finally, for each node of the ${\fg}$ Dynkin diagram, the charged fields in fundamental representation contribute
\beq\label{framsl}
\prod_{i, \alpha}  { \varphi_{q_a}( t v_a \,{a_{i, a} /x_{\alpha,a}})\over \varphi_{q_a}( v_a \,{a_{i,a} /x_{\alpha,a}})}
\eeq
where $i$ runs from $1$ to ${\rm rk}(W_a)$, and $q_{a} =q$ for short roots, and $q_a=q^m$ for the long roots, and $v_a = \sqrt{q_a/t}$.  This coincides with the two-point function, from \eqref{sctdasl3}, of screening currents and vertex operators associated to this node.

We have yet to pick a specific vacuum at infinity. In simple cases, this can be done by changing the contour of integration (to replace the contour in \eqref{scont} by an inequivalent one, that approaches thimble integrals in $q\rightarrow 1$ limit). This is not the most convenient way to do that, since in general cases construction of such contours becomes difficult.
Instead, it is better keep the contour of integration fixed to \eqref{scont} and instead realize the choice of vacua as additional insertions in the integral. They arise as follows.

We treat $C=\C$ as a finite disk (since nothing in the computation depends on the area of $C$), with boundary. To reproduce the vertex function ${\bf V}$, we need to impose Dirichlet boundary conditions on the gauge fields, and place the conditions on the matter fields to localize them to a component of $X_T$ at the boundary.  Instead of imposing boundary conditions by hand, we couple the 3d theory to a 2d theory at the boundary, and integrate over all the fields with no restrictions (for examples, see \cite{BDGH}). Due to couplings in \eqref{Trace} the boundary theory has only $(0,2)$ supersymmetry. The contribution of the elliptic genus of the boundary theory to the partition function leads to an additional insertion of \eqref{autcF}
$${\cF(x)}/{\theta(T^{1/2})} ={\cF'(x)},$$
in the integral, see \eqref{VI4}. The condition, from  \eqref{autcF}, that ${\cF'(x)}$ is invariant under $x\mapsto q^\bd x$ says that the boundary theory has no gauge anomalies. While there are many different theories that can be coupled consistently (any anomaly free $(0,2)$ theory would do), there is a finite dimensional space of distinct nontrivial contributions they could give rise to, parameterized by classes in $\Ell_\textup{T}(X)$.

\subsubsection{}

The vertex functions ${\bf V}$ lead solutions of qKZ which are holomorphic in $z$, per construction. We get a second basis of solutions to the same equation, which we denoted ${\bf V}_{{\fC}}$, which are holomorphic in a chamber $\fC$ of mass/equivariant parameter space, corresponding the choice of ordering of defects on ${\cal C}$. 
The vertex functions
 ${\bf V}_{\fC}$ and ${\bf V}$ solve the same set of difference equations in equivariant (and Kahler variables) since they originate from the same 3d gauge theory on $C\times S^1$. 
Correspondingly, the matrix ${\fP}_{\fC}$ relating them
\beq\label{mirror}
{\bf V}_{\fC} =  {\bf V} \; {\fP}_{\fC}
\eeq
 is a matrix of pseudo-constants. Theorem 4 of \cite{ese}, gives the matrix elements of ${\fP}_{\fC}$ in terms of elliptic stable envelopes of $X$.

The change of basis in \eqref{mirror} corresponds to imposing different conditions on the fields of the 3d theory at the $\partial({C} \times S^1) = T^2$ boundary. We can in principle impose boundary conditions leading to ${\bf V}_{\fC}$ in the same way as we did for ${\bf V}$, by coupling the 3d theory to a 2d theory on the boundary. 
This time the coupling, among other things, has an effect of imposing Neumann boundary conditions on the gauge fields. Having picked the chamber ${\fC}$, the stable basis leading to $a$-solutions of qKZ in this chamber is unique \cite{ese}. Here, we will only sketch some salient features of its construction.

To obtain a component of the covector ${\bf V}_{\fC}$, one starts by picking a component of the $A$-fixed point set $X_A \subset X$. The boundary conditions on matter fields parameterizing directions transverse to the fixed locus are either Neumann or Dirichlet boundary conditions depending on whether they correspond to attracting or repelling directions; this depends on ${\fC}$. The rest of boundary theory is determined by cancelation of gauge anomalies. More precisely, the choices left to make are parameterized by equivariant elliptic cohomology classes of the corresponding fixed point locus.
The elliptic genus of the boundary theory leads to a contribution to the integral which now takes the form 
\begin{equation}
  \label{defbS}
 {\cF(x)}/{\theta(T^{1/2})}  \qquad  \rightarrow  \qquad  \Stab_{\fC}^{\textup{\it ell}} (x,z) \,  \be(z)^{-1}/{\theta(T^{1/2})}   \, 
\end{equation}
 (See Sec. 6.3. of \cite{ese} for a more precise statement).
Here $ {\Stab}^{\textup{\it ell}}$ are elements of the elliptic stable basis, which assign, to every class in $\Ell_{\mathsf{T}}(X_\mathsf{A})$
a class in $\Ell_{\mathsf{T}}(X)$,
$$
\Stab_{\fC}^{\textup{\it ell}}(X): \qquad  \Ell_{\mathsf{T}}(X_\mathsf{A}) \;  \longrightarrow  \; \Ell_{\mathsf{T}}(X)
$$
 Per definition, 
the right hand side is invariant under  $x\mapsto q^\bd x$: the automorphy of the elliptic genus of the boundary theory cancels the bulk contribution coming from $ \be(z)= \exp\left(\frac{\bla(z,x)}{\ln q}\right)$.
This reflects the contribution of boundary degrees of freedom to the anomaly which cancels the anomaly the bulk theory has, in presence of $T^2$ boundary. Since the right hand side is constant under  $x\mapsto q^\bd x$ in computing the integral by residues, 
\beq\label{pcx}
{\fP}_{\fC} (x)= \Stab_{\fC}^{ell}(x) \, \be(z,x)^{-1} /{\theta(T^{1/2})}   \,  
\eeq
acts like a matrix of constants, so
${\bf V}_{\fC}$ is related to ${\bf V}$ by a linear operator ${\fP}_{\fC}$ in \eqref{mirror}, obtained by evaluating \eqref{pcx} on classical vacua.

The matrix ${\fP}_{\fC}$ in  \eqref{mirror} itself has a gauge theory interpretation. It the partition function of the 3d gauge theory on $I \times T^2$ with Neumann-type boundary conditions that lead to $a$-solutions are imposed on one end of $I$, and those for Dirichlet-type $z$-solutions on the other. The supersymmetric partition function does not depend on the size of the interval, and shrinking it to zero, one gets an effective two dimensional gauge theory on $T^2$ with $(0,2)$ supersymmetry. The entries of the matrix ${\fP}_{\fC}$ are elliptic genera of the resulting theories.

\subsubsection{}

Vertex functions with descendants correspond to placing line operators at $0\subset {\mathbb C}$, winding around the $S^1$. The line operators one needs can be constructed geometrically as well, see \cite{Bethe}, in terms of the K-theoretic stable basis ${\Stab}^K$. The later is a $q\rightarrow 0$ limit of the elliptic stable basis. One can make use of this fact to obtain its gauge theory construction: first cut open a neighborhood of $0\subset {\mathbb C}$, and impose the boundary conditions corresponding to elliptic stable basis. Then, shrinking the boundary back to a point has the same effect as taking $q$ to zero. The elliptic genus of the boundary theory becomes a line operator insertion -- this is the supersymmetric partition function on $S^1$ of the resulting quantum mechanics problem.
Inserting the line operator, in the integral \eqref{scont} takes ${\bf V}$ and ${\bf V}_{\fC}$ to fundamental $z$- and $a$-solutions of qKZ.

\subsection{Conformal limit}\label{sec:defse}

The variables $q,t, \hbar$ are related to the parameters of the $\Omega$-background as 
\beq\label{epsilons}
q=\exp(R' \epsilon_q), \qquad t=\exp(R' \epsilon_t),\qquad \hbar=\exp(R' \epsilon_\hbar),
\eeq
In the conformal limit, point particle limit, we send $
m_s \rightarrow \infty$
and we keep $\epsilon$'s fixed, since they are part of the definition of the background the $(2,0)$ theory is compactified on. For the same reason, we keep the Riemann surface ${\cal C}$ fixed.  This means the radius ${R}$ of the circle in ${\cal C}$ must stay fixed, and hence the $T$-dual radius 
$$
R' = 1/(m_s^2 {R}) \qquad \rightarrow \qquad  0  
$$ 
goes to zero in the limit. Since $R'$ goes to zero, with $\epsilon$'s fixed, we recover \eqref{limitwzw} and \eqref{limitw}. 

The positions of the points on the Riemann surface are fixed as well, but the
$z$'s have to scale to $1$ to keep the moduli of the $(2,0)$ theory fixed in the limit. Namely, 
\beq\label{deflambda}
z = \exp(R' \zeta) = q^{\mu}
\eeq
where $\zeta$ is the 3d FI parameter complexified by the holonomy of the corresponding background gauge field around the $S^1$. It follows from its string theory origin that ${\rm Re}(\zeta_a)$ is $R$ times the modulus of the $(2,0)$ theory, and both of these we need to fix in the limit. This implies that  ${\rm Re}(R' \zeta_a)$ goes to zero in the conformal limit, and hence $z$ goes to $1$. The rate at which $z$ goes to one is fixed, however, so $\mu$ defined by \eqref{deflambda} remains fixed.

%
%\beq\label{again trace}
%{\rm Tr}(-1)^F {\bf g}
%\eeq

\section{Langlands Correspondence From Little Strings}\label{sec:nine}

It has been known for a long time that geometric Langlands correspondence should be a consequence of $S$-duality of the maximally supersymmetric ${\cal N}=4$ Yang-Mills theory \cite{KW, Wittenl1, Wittenl2}. While some aspects of $S$ duality can be understood within the gauge theory, and many more from theory ${\scX}$, to derive ${S}$-duality one needs string theory.
This was shown in \cite{Vafa}, and reviewed recently in \cite{Seiberg}.

In this section we will recall the derivation of $S$-duality from little string theory, as well as the expected relation between $S$-duality of the ${\cal N}=4$ theory and the geometric Langlands. The fact that one is able to derive $S$-duality from little string theory offers an explanation why one can derive (quantum) geometric Langlands from it.

\subsection{$S$ duality of 4d Yang-Mills theory} 
$S$-duality relates Yang-Mills theories with ${\cal N}=4$ supersymmetry and gauge groups based on Lie algebras ${\fg}$ and $^L{\fg}$, 
\beq\label{Sdual}
S: \qquad (^L{\fg}, {^L\tau}) \qquad \longleftrightarrow \qquad  ({\fg}, \tau).
\eeq
The gauge coupling parameters are related by
\beq\label{dualt}
m
\, \tau \,^L\tau =-1.
\eeq
where $\tau$ is given by $\tau ={ \theta \over {2\pi}} + i {4\pi \over g^2_{YM}}$ in terms of the Yang-Mills coupling constant $g_{YM}$ and the $\theta$ angle.
The theory with Lie algebra ${\fg}$ has in addition a symmetry $T$ corresponding to the $2\pi$ shift of the theta angle\footnote{This assumes the normalization of the invariant metric on ${\fg}$ we have chosen -- the one in which the short coroots of ${\fg}$ have length squared equal to $2$, see \cite{D, Seiberg}.}
$$T: \qquad ({\fg}, \tau)\qquad  \rightarrow  \qquad ({\fg}, \tau+1),
$$ 
which maps the theory to itself for any ${\fg}$. The action of the $S$ and $T$ on the particles of the theory is always non-trivial.
\subsection{Derivation of $S$-duality from little string theory}\label{sec:lsduality}
Start with IIB string on 
\beq\label{IIBb}
(Y_0 \times S^1_{q} \times S^1_t)/H \times M_4,
\eeq
where $Y_0$ is an ADE surface corresponding to ${\fg}_0$ Lie algebra, as in section  \ref{sec:six}, and $S^1_q$, $S^1_t$ are a pair of circles (the subscripts are there to distinguish them). The $H$-twist acts by folding the Dynkin diagram of ${\fg}_0$, in going once around $S^1_q$. Nothing in what follows depends on the value of the string coupling, so we can take $g_s$ to zero to get the $\fg_0$ little string theory on $S^1_{q} \times S^1_t \times M_4$, with the $H$-twist.

We would like to understand which four dimensional theory we get when we send to zero the characteristic size of the string $1/m_s$ and the area of the two torus $T^2 = S^1_q \times S^1_t$. The resulting theory can be derived using $T$-duality symmetry of string theory \cite{Vafa}.

$T$-duality on the $S^1_t$ circle leads to the description based on ${\cal N}=4$ SYM theory with Lie algebra ${\fg}$. The description based on $^L{\fg}$ follows from $T$-duality on the $S^1_q$ circle, instead. The more weakly coupled description comes from $T$-duality on the smaller of the two circles. 

\subsubsection{}\label{sec:gc}

$T$-duality on $S^1_t$ circle relates IIB string on \eqref{IIBb}
to IIA  string on 
\beq\label{IIAa}
(Y_0 \times S^1_q \times S^1_{t'})/H \times M_4.
\eeq
The two string theories are equivalent once we exchange the momentum and the winding modes around the $S^1_t$. 
$Y_0$ is the same ADE surface as in \eqref{IIBb}, and the  $H$-twist is on $S^1_q$. If $R_q$ and $R_t$ are the radii of $S^1_q$ and $S^1_t$ circles, respectively, then ${S}^1_{t'}$ is a circle of radius $R_{t'}= 1/(R_t \,m_s^2 )$. At the singularity in IIA theory we get the $(1,1)$ little string theory, as described in Section \ref{sec:six}.  This theory has a six dimensional gauge symmetry, with gauge group based of ${\fg}_0$ Lie algebra; its gauge coupling parameter is $m_s^2$ \cite{SeibergN}. Presently, the $(1,1)$ little string is compactified on the two-torus $(S^1_q \times S^1_{t'})/H$ times $M_4$, where 
the $H$-twist around $S^1_q$ permutes the fields of the gauge theory according to the action of $H$ on the Dynkin diagram of ${\fg}_0$. 
The theta angle originates from the $NS$ B-field on the two torus; its periodicity results in the symmetry $T$ we had before.

Starting with the gauge theory in six dimensions with $(1,1)$ supersymmetry, in the limit we take the area of the two-torus to zero, and $m_s$ to infinity, we get a four dimensional gauge theory on $M_4$ with ${\cal N}=4$ supersymmetry, gauge group based on ${\fg}$, and coupling $\tau = i m_s^2 R_q R_{t'} =  i R_q/R_t$. This follows by restricting the fields of the 6d gauge theory to constant (zero momentum) modes around the $T^2$; these are the only excitations of the whose energy remains finite as the size of the torus goes to zero.

\subsubsection{}\label{sec:gcL}

It was shown in \cite{Vafa} that $T$-duality on $S^1_q$ circle relates IIB string on \eqref{IIBb} to
IIA  string on 
\beq\label{IIAb}
{(^LY_0 \times S^1_{q'} \times S^1_t)/ {^LH} }\times M_4,
\eeq 
with twist by $^L H$, going around $S^1_{q'}$ circle once. Here $^LY_0$ is the ADE singularity based on a simply laced Lie algebra $^L {\fg}_0$. The Lie algebra $^L{\fg}_0$ has outer automorphism group $^L H$, such that by projecting to its $^L H$ invariant part we get $^L {\fg}$, the Lie algebra which is the Langlands dual of ${\fg}$. 
 The radius of $S^1_{q'}$ is $R_{q'} = {1/(m\,  R_q  \,m_s^2)}$. The factor of $m$ comes about due to the $H$-twist on the original circle in IIB:  the momentum on $S^1_q$ is quantized in units of $1/(mR_q)$ since all modes come back to themselves only after going $m$ times around $S^1_q$. Hence, the strings wound on the $T$-dual circle $S^1_{q'}$ must have masses quantized in units of $m_s^2 \, R_{q'} = 1/(m R_q)$.

The $(1,1)$ little string theory one gets by decoupling the modes far from the singularity in IIA theory now has the low energy description as a six dimensional maximally supersymmetric gauge theory based on the $^L{\fg}_0$ Lie algebra. 
%The little string theory is now compactified on the two-torus $(S^1_{q'} \times S^1_{t})/^LH$, with $H$-twist around $S^1_{q'}.$ 
The four dimensional theory on $M_4$, which we get in the limit the area of the two-torus goes to zero, has ${\cal N}=4$ supersymmetry, gauge group based on $^L{\fg}$, and coupling $^L\tau = i m_s^2 R_t R_{q'} = i R_t/(m R_q)$. In particular, $\tau$ and $^L\tau$ are related by \eqref{dualt}.  \subsubsection{}

In principle, in addition to the Lie algebra, one should specify the global form of the gauge group on each side in \eqref{Sdual}. This corresponds to specifying the allowed representations of electrically charged fields, a character sublattice of the weight lattice of $^L{\fg}$; its dual lattice is the character lattice of ${\fg}$, see \cite{GuW} for review. In this paper, we will allow for the most general choice of electric charges for $^L{\fg}$, choosing the character and weight lattices to coincide. This implicitly sets $^LG$ to be the simply connected group with Lie algebra $^L{\fg}$. The dual group $G$ than is of adjoint type, as its weight lattice equal to its root lattice. 

\subsection{Gauge theory partition function from little string}

We showed that partition functions of $(2,0)$ 6d theory on $M_6^{\times}$ and $M_6$ compute the conformal blocks of $\Lfgh_{^L k}$ and ${\cal W}_{\beta}(\fg)$ (in the conformal limit of our results).  The relation of $(2,0)$ 6d theory to the pair of ${\cal N}=4$ gauge theories with gauge groups based on $^L{\fg}$ and ${\fg}$ then implies that conformal blocks are the partition functions of these gauge theories, in the background induced from their six dimensional origin. 

This leads to an explicit relation between the $S$-duality of ${\cal N}=4$ gauge theories in four dimensions and the conformal field theory approach to geometric Langlands, which we reviewed in Sec. \ref{glc}. We will describe some essential aspects (see also Sec. 8 of \cite{GW}), leaving a more detailed analysis for future work.

\subsubsection{}
Consider ${\fg}_0$-type $(2,0)$ little string theory compactified on a six manifold 
$$
 M_6^\times = {\cal C} \times ({\mathbb C} \times {\mathbb C}^\times)/H,
$$ 
with an $H$ twist around ${\mathbb C}^{\times}$. $M_6^{\times}$ differs from $M_6$ in \eqref{6man} by having the origin of one of the complex planes deleted. This is merely a convenient choice made for ease of discussion. 
Working with $M_6^\times$ leads to partition functions which compute
vector valued $q$-conformal blocks, instead of scalar ones we get from
$M_6$. The converse is that closing up the puncture, and thereby
replacing $M_6^\times$ with $M_6$, corresponds to contraction of the
vector valued partition function with the Whittaker type vector in
\eqref{linear}.

The six-manifold  $M_6^{\times}$ is a $T^2 =S^1_t \times S^1_q
$
fibration 
\beq\label{surface}
T^2 \;\; \rightarrow\;\; M_6^{\times} \;\; \rightarrow \;\; M_4={\cal C} \times B.
\eeq
As we go once around the $S^1_q$ circle, viewed as the circle fiber of ${\mathbb C}^{\times}$, we twist by $H$; the fiber of ${\mathbb C}$ is $S^1_t$. $T$-duality of little string theory on the $S^1_t$ or on the $S^1_q$ circle fiber leads to two distinct descriptions of the four dimensional theory on $M_4$, as we reviewed in section \ref{sec:lsduality}: the first leads to the ${\cal N}=4$ SYM theory based on gauge group ${\fg}$, the second based on$\,^L{\fg}$. 

The base $M_4$ is a manifold with a boundary, since $B={\mathbb R} \times {\mathbb R}^{+}.$  The boundary conditions \cite{KW, GW} for the two 4d gauge theories we need are defined by recalling their origin from the 6d $(2,0)$ theory on $M_6^{\times}$, which is a six manifold without boundaries \cite{WF}.
\subsubsection{}\label{boundaries}

We studied the supersymmetric partition function of the $(2,0)$ little string theory on $M_6^{\times}$ in section \ref{sec:seven}.
In the conformal limit, the partition function we defined in section \ref{sec:seven} becomes the partition function of the $(2,0)$ 6d conformal field theory on ${\cal C}$ times the 4d $\Omega$-background on ${\mathbb C}\times {\mathbb C}^{\times}$, or on ${\mathbb C}\times {\mathbb C}$ if we replace $M_6^{\times}$ by $M_6$. 

It was argued in \cite{NW} that placing the $(2,0)$ 6d CFT theory on ${\cal C}$ times a 4d $\Omega$-background leads to partition functions of the two $S$-dual  ${\cal N}=4$ theories on $M_4$ with topological twist of geometric Langlands type, studied in \cite{KW}. 

Further, \cite{NW} explained how to relate the parameters of ${\Omega}$-background to the effective coupling constant of the gauge theory. One uses the fact that, asymptotically and locally, far away from the fixed points of rotations by ${\mathbb C}_{t}^{\times}$ and ${\mathbb C}_q^{\times}$, $M_6$ is a flat manifold. There, all effects of topological twisting go away and the ${\Omega}$-background parameters of the twisted theory are identified\footnote{Our conventions for $\epsilon$'s are set in Sec. \ref{sec:seven} and in \ref{sec:defse} and they differ from those in \cite{NW} by factors of $2\pi$.} with the inverse radii of the two $S^1$'s in $T^2 = S^1_q \times S^1_t$ in the undeformed gauge theory: $\epsilon_t = 2\pi /R_t$ and $\epsilon_q = 2\pi/(i m R_q)$. 

Putting our results together with those of \cite{NW}, we find the following.

%The partition function of the little string theory on $M_6$ (or on $M_6^\times$) was computed in Sections 4-6. This requires an ${\Omega}$-background which rotates ${\mathbb C}_{q}$ and ${\mathbb C}_t$ planes by $q$ and $t^{-1}$; in addition, to preserve supersymmetry, we imposed an R-symmetry twist, leading to the index in \eqref{Trace}. In the conformal limit, we get the partition function of theory ${\scX}$ on $M_6^{\times}$, and the partition functions of the two ${\cal N}=4$ theories on $M_4$.
%

%The two partition functions, of theories based on ${\fg}$ and $^L{\fg}$, depend each on a single complex parameter $\Psi$ and $^L\Psi$, related by
%
%$$
%m \Psi \, ^L\Psi =-1.
%$$
%They are related to the parameters $q$ and $t$ of the ${\Omega}$-background\footnote{
%The parameters  $\Psi$ and $^L\Psi$ are each combination of the respective Yang-Mills coupling constants $g_{YM}$, a theta angle $\theta$, and an additional parameter, called $t$ in \cite{KW, WF}, associated to the choice of the supersymmetry preserved. Ordinarily, the modular parameter $\tau$ of the torus is $\tau = {\theta\over 2\pi} + {2\pi i \over g_{YM}^2}.$ $\Psi$ ends up transforming in the same way as $\tau$ does under modular transformations, even though it depends on $t$ in addition.} in the way described in \cite{NW}
%\beq\label{Pe}
%\,^L\Psi =  \, {\epsilon_q / \epsilon_t}, \qquad  \,\Psi =  \, {\epsilon_t/(m \epsilon_q)}.
%\eeq
%\end{comment}

\subsubsection{}
The chiral conformal block of ${\cal W}_{\beta} (\fg)$,
 corresponds to the partition function of 4d SYM theory with gauge group based on $\fg$ and coupling (see Sec. \ref{sec:gc})
\beq\label{Peb}
\tau =  {\epsilon_t / m \epsilon_q} = \beta/m,
\eeq
which one gets from $(2,0)$ theory on $M_6$. One wants to work with $M_6$ rather than $M_6^{\times}$ here since the ${\cal W}_{q,t}(\fg)$ algebra blocks from \eqref{magnetic} are naturally scalar.

This relation  follows from AGT correspondence \cite{AGT}, and was used in \cite{NW}.
It also follows from our results (and from \cite{AH}) by taking the conformal limit.

\subsubsection{}
The chiral conformal blocks of $\Lfgh_{^L k}$ 
%\beq\label{pg}
%^L \beta = ^L(k+h) 
%\eeq
correspond to YM theory with gauge group $^L{\fg}$ and coupling parameter $^L\tau$ given by (see Sec. \ref{sec:gcL})
\beq\label{Pe}
\,^L\tau ={\epsilon_q / \epsilon_\hbar} = -{^L(k+h)}.
\eeq
It follows when we place the theory on $M_6^{\times}$.

This relation to WZW models was predicted in \cite{OV} nearly 20 years ago. The $(2,0)$ conformal filed theory on any three manifold $M_3$ times ${\mathbb C} \times S^1_q$ is expected \cite{OV} to compute the partition function of Chern-Simons theory based on  $^L{\fg}$ Lie algebra, on $M_3$.  In the present case, this applies with $M_3 = {\cal C}\times {\mathbb R}$. For non-simply laced $^L{\fg}$, one starts with the ${\fg}_0$ type $(2,0)$ theory, and introduces the twist by $H$ \cite{WF},  just as we did. 

In the construction of \cite{OV}, the level $^L(k+h^{\vee})$ of Chern-Simons theory is determined by the parameter $q'$ arising geometrically from the $\Omega$-background on ${\mathbb C}\times S^1_q$. To get Chern-Simons theory at level $^Lk$ from $(2,0)$ theory on $M_3\times {\mathbb C}\times S^1_q$,  we rotate ${\mathbb C}$  by $q' =\exp({2\pi i \over ^L(k+h^{\vee})})$ as we go around $S^1_q$, and accompany the rotation with an $R$-symmetry twist. 

We can apply this here, with one subtle point. Namely, we need $\epsilon_{\hbar}$ not $-\epsilon_t$ to be the ${\Omega}$ background that rotates the complex plane (the two are related by $ \epsilon_{\hbar} = \epsilon_q-\epsilon_t$). This is related to the fact that since ${\mathbb C}^{\times}$ is a cylinder, the topological twist on it is trivial. We only need the $R$-symmetry twist with parameter $\epsilon_{\hbar}$, from section \ref{sec:seven}, which gets compensated by a twist of the ${\mathbb C}$-plane in $M_6^{\times}$ by $\epsilon_{\hbar}$. Altogether, we find $q' = \exp(- m R_q \cdot \epsilon_\hbar)$, and since $ i m R_q = 2\pi /\epsilon_q$, \eqref{Pe} follows.

\subsubsection{}

Since $m \epsilon_q = \epsilon_t+\epsilon_{\hbar}$, we have that the two theories are related by
$$ 
\tau - 1=-1/(m ^L \tau),
$$
as in \eqref{firstl}.
This is the action of $S$-duality, together with the shift of the $\theta$ angle in the ${\fg}$-theory. 
\subsection{Little string defects and line operators in gauge theory}

In the $(2,0)$ theory on $M_6^{\times}$, we have defects supported on ${\mathbb C}_q^{\times}$, labeled by collection of weights of $^L{\fg}.$ These defects, which originate as D3 branes of IIB wrapping 2-cycles of $Y_0$, are self-dual strings of the $(2,0)$ theory. (The self-dual strings are strings present both in the $(2,0)$ little string theory, and in theory ${\scX}$.  
%charged under the Lie algebra valued 2-form field of the $(2,0)$ theory, and their tension depends on the moduli of $Y_0$. 
They are distinct from fundamental strings of little string theory, which are not present in theory ${\scX}$.) 
 
Reducing the 6d $(2,0)$ theory on $T^2$ to ${\cal N}=4$ theory on ${M}_4  ={\cal C} \times {\mathbb R}_q \times {\mathbb R}_t^+$, the self dual strings supported on $S^1 \subset T^2$ become particles on ${M}_4$. We have particles are supported at a collection of points on ${\cal C}$, with coordinates $\{a_{i}\}$, and charges which are labeled by weights of $^L {\fg}$. They are located at the tip of  ${\mathbb R}_t^+$, and their world lines are along the ``time'' direction ${\mathbb R}_q$. Presence of such particles affects the partition function of the 4d theory by insertion of line operators. Which line operator we get, depends on the ${\cal N}=4$ gauge theory description one uses.

\subsubsection{}
In the ${\cal N}=4$ theory description based on $^L{\fg}$, the self-dual string of the 6d theory supported on ${\mathbb C}_q^{\times}$ become the Wilson line operators.  Namely, when we view the $T^2$ compactification of theory ${\scX}$ as a two step reduction, reducing on $S^1_q\in {\mathbb C}_q^{\times}$ first, the self-dual string defects become particles already in five dimensions, electrically charged under $^L{\fg}$ valued gauge field. Reducing further on $S^1_q$, we get the Wilson lines of ${\cal N}=4$ theory on $M_4$.
This is as expected from our description, in \ref{sec:two}, of the conformal blocks we study. Namely, the Wilson line operators of the Yang-Mills theory become Wilson line operators of $\Lfgh$ Chern-Simons theory and the corresponding WZW model. 
\subsubsection{}
In the ${\cal N}=4$ theory based on ${\fg}$, the same strings give rise to 't Hooft line operators. Compactifying theory ${\scX}$ on $S^1_t$, we get a five dimensional gauge theory with strings; the strings are charged under the magnetic dual of the ${\fg}$ valued 5d gauge field, a two-form. After further compactifying on $S^1_q$ they become magnetically charged particles; their world-lines introduce t'Hooft line operators in $M_4$. The 't Hooft line operators are labeled by coweights of ${\fg}$ and hence by weights of $^L{\fg}$. The effect the line operator in the ${\fg}$ gauge theory is described by how they affect the boundary conditions at the tip of ${\mathbb R}_t^+$.
 In the limit when the gauge theory become classical, \cite{GW} argued that the boundary conditions one gets are described in terms of ${\fg}$ {opers}. 
 %Opers are flat ${\fg}$-bundles over ${\cal C}$ with singularities at points where the monopoles are inserted on ${\cal C}$, but no monodromy around them. 
This agrees with what we find, since{ opers} describe the classical limit of ${\cal W}_{\beta}({\fg})$ algebra conformal blocks of our paper.

\appendix

\section{Integral formulas in K-theory of 
GIT quotients}\label{s_GIT}

Let a reductive group $G$ act on an variety $\tX$ and let $\cL$ be a very 
ample $G$-linearized line bundle on $\tX$. The GIT quotient 
\begin{equation}
X = \tX \rd_\cL G  =  \Proj 
\left( \bigoplus_n H^0(\tX,\cL^{\otimes n})^G \right) \label{XProj}
\end{equation}
is the categorical quotient of the set of semistable points 
$$
\tX_\textup{ss} = \{x \in \tX, \exists s\in H^0(\tX,\cL^{\otimes n})^G, 
s(x) \ne 0\}
$$
by the action of $G$. We denote by 
$$
\pi_{\textup{\tiny GIT}}: \tX_\textup{ss} \to X
$$
the canonical affine morphism.

\subsubsection{}

Let $\tcF$ be a $G$-equivariant 
coherent sheaf on $\tX$, it induces a coherent sheaf $\cF$ on $X$ 
by 
$$
\Gamma(U,\cF)  = \Gamma(\pi_{\textup{\tiny GIT}}^{-1}(U),\tcF)^G \,.
$$
In particular, $\cL$ itself induces the canonical line bundle $\cO(1)$ 
on $X$.

Our interest is in the computation of $\chi(X,\cF)$ in terms involving 
the prequotient $\tX$. This is the value of the quasipolynomial 
$\chi(X,\cF(m))$ at $m=0$, where $\cF(m) = \cF \otimes \cO(m)$. 
By definition, a quasipolynomial in $m$ is an element 
of a ring of the form 
$$
\Q[m,a_1^{\pm m}, a_2^{\pm m}, \dots] \,,
$$
where the parameters $a_i$ may be roots of unity 
 or weights of a group of automorphisms 
of $X$. 

\subsubsection{}

There is the following basic 

\begin{Lemma}
 For $m\gg 0$, 
 \begin{equation}
   \label{mgg0}
    \chi(X,\cF(m)) = \chi(\tX, \tcF 
\otimes \cL^m)^G \,. 
 \end{equation}
\end{Lemma}

A more general formula, valid without the $m\gg 0$ assumption, 
follows from the results of Teleman \cite{Tel}, see also 
\cite{Wood} and \cite{DHL}. 

\begin{proof}
Since $\cL$ is ample, we have 
$$
\chi(X,\cF(m)) = \Gamma(X,\cF(m)) = \Gamma(\tX_\textup{ss}, \tcF 
\otimes \cL^m)^G \,, 
$$
for $m\gg 0$. Therefore, it suffices to see that 
the natural restriction map 
\begin{equation}
\Gamma(\tX, \tcF 
\otimes \cL^m)^G  \to \Gamma(\tX_\textup{ss}, \tcF 
\otimes \cL^m)^G \label{resG}
\end{equation}
is an isomorphism for $m\gg 0$. The spaces in the 
source and the target in \eqref{resG} form a module over 
the graded algebra in \eqref{XProj}. The sheaf $\tcF$ is 
coherent and the line bundle $\cL$ is ample, hence for 
sufficiently large $r$ and $d$ there is a map 
$$
\left(\cL^{-d}\right)^{\oplus r} \to \tcF
$$
inducing a surjection 
$$
\Gamma(X,\cL^{m-d})^{\oplus r} \to 
\Gamma(X,\tcF \otimes \cL^m) \to 0 
$$
for $m\gg 0$. Because $G$ is reductive, we get a surjectivity 
for $G$-invariant sections and so the modules in question 
are finitely generated. Therefore, 
$$
\bigoplus_{d \le D} 
\Gamma(\tX, 
\cL^{m-d})^G  
\otimes \Gamma(\tX, \tcF 
\otimes \cL^{d})^G  \to \Gamma(\tX, 
\tcF  \otimes \cL^{m})^G \to 0 
$$
where $D$ is the maximal degree of a generator, 
 and similarly for $\tX_\textup{ss}$ in place of $X$. 
Since all sections in $\Gamma(\tX_\textup{ss}
\otimes \cL^{m})^G$ extend by zero to $\tX$, the isomorphism in 
\eqref{resG} follows.
\end{proof}

\subsubsection{}

{}From now on we assume there is a torus
$$
\bT \subset \Aut_G(\tX)
$$
acting on $\cL$ and $\cF$ that contracts $\tX$ to a proper
$G$-invariant set as $t\to 0_\bT$, where $0_\bT$ is a point in a toric 
compactification of $\bT$. This is the case, 
for example, when $\tX$ is a linear representation of $G$, or the 
zero locus of a moment map in a linear symplectic representation of 
$G$. The additional $\bT$-grading makes the trace
\begin{equation}
\tr_{\Gamma(\tX, \tcF 
\otimes \cL^m)} (g,t) \in \Q(G\times \bT) \label{trgt}
\end{equation}
converge for $|t|\gg 1$ to a rational function. Here $|t|\gg 1$ means 
that $t^{-1}$ lies in a certain neighborhood of $0_\bT$ and in that 
region the poles of \eqref{trgt} are disjoint from any fixed 
maximal compact subgroup $G_\textup{compact}\subset G$. 
Therefore 
\begin{align}
  \tr_{\Gamma(\tX, \tcF 
\otimes \cL^m)^G} t &= 
\int_{G_\textup{compact}} \tr_{\Gamma(\tX, \tcF 
\otimes \cL^m)} (g,t) \, d_\textup{Haar}g \notag \\
&= 
\frac{1}{|W|} \int_{|s|=1} \Delta_\textup{Weyl}(s) \tr_{\Gamma(\tX, \tcF 
\otimes \cL^m)} (s,t) \, d_\textup{Haar}s 
\label{intW} 
\end{align}
for $|t|\gg 1$, where 
\begin{align}
  W &\quad \textup{is the Weyl group of $G$,}\notag \\
  \{|s|=1\} \subset G_\textup{compact} &\quad \textup{is a maximal
                                         torus,} \label{torusS} \\
 \Delta_\textup{Weyl}(s)
                                       &\quad \textup{is the Weyl
                                         denominator}\,,
\notag 
\end{align}
and the Haar measures are normalized to have total mass $1$. 

\subsubsection{}

We denote by $\bS\subset G$ the complexification of the 
torus in \eqref{torusS}. 
By localization, 
\begin{equation}
\tr_{\Gamma(\tX, \tcF 
\otimes \cL^m)} (s,t) = 
\sum_k  \frac{p_k(s,t)}{\prod (1-w_{k,i}^{-1})} \, \nu_k^m \label{loc_sum}
\end{equation}
where the sum in \eqref{loc_sum} is over the components $
\widetilde{F}_k
$ of the 
fixed locus 
$$
\tX^{\bS\times \bT} = \bigsqcup \widetilde{F}_k \,,
$$ 
$p_k$ are certain Laurent polynomials in $s$ and $t$, 
the characters $\nu_k$ are the $\bS\times \bT$-weights of 
$\cL$ restricted to the components of the fixed locus, and 
$w_{k,i}$ are the weights in the denominators of the localization 
formula (i.e.\ normal weights to the fixed locus in some ambient 
smooth equivariant embedding of $\tX$).

\subsubsection{}\label{ss_residues}

The integral in \eqref{intW} may be computed by residues as follows. 
By linearity, it suffices to deal with each term in \eqref{loc_sum} 
separately. If $\nu_k\big|_\bS$ is a trivial character then 
$$
\int_{|s|=1} \dots \, \nu_k^m \, d_\textup{Haar}s  = 
\nu_k^m \,  \int_{|s|=1} \dots \, d_\textup{Haar}s \,,
$$
which is tautologically a quasipolynomial in $m$. 

If $\nu_k\big|_\bS$ is a nontrivial character, then we 
deform the integration contour $\{|s|=1\}$ to the region 
$|\nu_k| \ll 1$ in $\bS$ while picking up residues in the 
process. These residues are integrals of the same form over 
translates of codimension $1$ subtori is $\bS$, and so we
can deal with them inductively. 

The resulting quasipolynomial in $m$ computes the 
quasipolynomial $\chi(X,\cF(m))$. 

\subsubsection{}\label{s_ex_Pn}
For the most basic example, we can take $G=GL(1)$ acting with 
weight one on $\C^n$ and $\cL = \cO_{\C^n}(k)$, where the twist
is by the $k$th power of the defining representation. Then 
$$
\Gamma(\cL^{m})^G = \textup{polynomials of degree $km$}\,, 
$$
and so 
$$
X = \begin{cases}
\bP^{n-1}\, \quad &k> 0 \,, \\
\pt\, \quad &k= 0 \,, \\
\varnothing \,, \quad &k<0\,, 
\end{cases}
$$
or, taking polarization into the account, $X$ is the $k$th Veronese embedding of $\bP^{n-1}$
for $k>0$. 
We can take $\bT=\bA$, where 
\begin{equation}
\bA = \left\{ 
  \begin{pmatrix}
    a_1 \\
 & \ddots \\
 && a_ n 
  \end{pmatrix} \right\} \subset GL(n)\,,\label{bAGL}
\end{equation}
which gives normal weights $w_i = s a_i$, $i=1,\dots,n$ at 
the unique fixed point $0\in \C^n$. Therefore, we get the 
integral 
$$
\chi(X,\cO(m)) = 
\frac{1}{2\pi i} 
\int_{|s|=1} \frac{s^{km}}{\prod(1-a_i^{-1} s^{-1})} \, \frac{ds}{s}
\,,
\quad |a_i| > 1 \,, 
$$
which may be computed by 
deforming the contour to $|s|=\varepsilon^{\pm 1}$, depending 
on the sign of $k$. 

\subsubsection{}\label{s_large_chi}

Of importance to us will be the special case when $\cL$ is twisted by 
a large power of a nontrivial $G$-characters $\chi$. In this case, we 
can start the analysis of \eqref{intW} with deforming the 
contour 
into the region $|\chi| \ll 1$.  

We denote by 
$$
\bS^\circ = \left\{s\, \Big| \, \forall w_{k,i}, w_{k,i} \ne 1
\right\}
\subset \bS 
$$
the regular locus of the integrand. The homology groups 
of $\bS^\circ$ have been studied in detail, see 
\cite{ConProc}. In particular, noncanonically, 
\begin{equation}
H_*(\bS^\circ,\C) = \bigoplus_{\bS'}
H_*(\bS',\C) \otimes H_*(N_{\bS/\bS'} \setminus 
\{\textup{hyperplanes}\}, \C) \label{homS0_}
\end{equation}
where $\bS'$ ranges over the components of all possible 
intersections of $\{w_{k,i}=1\}$, and hyperplanes in the 
(trivial) normal bundle $N_{\bS/\bS'}$ to $\bS'$ are cut out by 
the differentials of the characters $w_{k,i}$ trivial on $\bS'$.  Since all 
homology groups vanish above the middle dimension, we have
\begin{equation}
H_\textup{mid}(\bS^\circ,\C) = \bigoplus_{\bS'}
H_\textup{mid} (\bS',\C) \otimes H_\textup{mid} (N_{\bS/\bS'} \setminus 
\{\textup{hyperplanes}\}, \C) \label{homS0} \,. 
\end{equation}

For the computation of the integral, we are interested in homology
relative the subset $|\chi| \ll 1$, and for those we conclude 
\begin{equation}
H_\textup{mid}(\bS^\circ,\{|\chi| \ll 1\}, \C) = \bigoplus_{\bS', \chi\big|_{\bS'}=\textup{const}}
\textup{same as in \eqref{homS0}} \,, \label{relhom}
\end{equation}
which parallels the computation by residues discussed in Section 
\ref{ss_residues}. We set
\begin{equation}
\gamma_\chi = \textup{image of $\{|s|=1\}$ in LHS of \eqref{relhom}}
\,.\label{def_gamma_chi}
\end{equation}
As a middle-dimensional cycle, it is represented by products of a 
maximal compact torus in 
$\bS'$ with a middle-dimensional cycle in a certain 
hyperplane arrangement. We conclude the following 

\begin{Proposition}
If $\cL$ is twisted by a sufficiently large power of a character
$\chi$ then 
\begin{equation}
  \label{intchi}
  \chi(X, \cF(m)) = \frac{1}{|W|}  
\int_{\gamma_\chi}  \Delta_\textup{Weyl}(s) \tr_{\chi(\tX, \tcF 
\otimes \cL^m)} (s,t) \, d_\textup{Haar}s 
\end{equation}
for $m\gg 0$, and for all $m$ if $\dim \bS'=0$ for all $\bS'$ in 
\eqref{relhom}. 
\end{Proposition}

If $\dim \bS' >0$ for a certain $\bS'$ in \eqref{relhom} then 
the corresponding integral needs to be treated as in 
Section \ref{ss_residues} to pick the right quasipolynomial in $m$. 

An important special case when one can be sure that $\dim \bS' = 0$
for all $\bS'$ in \eqref{ss_residues} is the case of Nakajima quiver 
varieties. More generally, we have the following simple 

\begin{Lemma}\label{l_dim0}
Suppose $G=\prod GL(V_i)$ and 
$$
\{w_k\} = \textup{weights of $V_i$, $V_i^*$, and $V_i \otimes
  V_j^*$}\,,
$$
where $i,j=1,\dots,n$. Then a generic character $\chi$ of $G$ is 
nontrivial on every components $\bS'$ of $\{w_{k_1}=\dots=w_{k_l}=1\}$ 
of positive dimension. 
\end{Lemma}

\begin{proof}
This is equivalent to the differential 
$d\chi$ being in the span of $dw_{k_1},\dots,dw_{k_l}$ if and only if
this span is the whole space. The generic character is not 
in the span of weights of $V_i \otimes V_j^*$ and so at least one
fundamental or the dual fundamental weight has to appear among 
$w_{k_i}$. We can then argue modulo this weight and induct on 
$\sum \dim V_i$. 
\end{proof}
\pagebreak


\begin{thebibliography}
%\normalsize
%\normalsize\vskip-1.5cm
\raggedright

\bibitem{IH} 
  M.~Aganagic, R.~Dijkgraaf, A.~Klemm, M.~Marino and C.~Vafa,
  ``Topological strings and integrable hierarchies,''
  Commun.\ Math.\ Phys.\  {\bf 261} (2006) 451
  %doi:10.1007/s00220-005-1448-9
  [arXiv:hep-th/0312085].
  %%CITATION = doi:10.1007/s00220-005-1448-9;%%
  %296 citations counted in INSPIRE as of 28 Dec 2016

\bibitem{AH} 
  M.~Aganagic and N.~Haouzi,
  ``ADE Little String Theory on a Riemann Surface (and Triality),''
  arXiv:1506.04183 [hep-th].
  %%CITATION = ARXIV:1506.04183;%%
  %1 citations counted in INSPIRE as of 08 Oct 2015 
\bibitem{AHKS} 
  M.~Aganagic, N.~Haouzi, C.~Kozcaz and S.~Shakirov,
  ``Gauge/Liouville Triality,''
  arXiv:1309.1687 [hep-th].
  %%CITATION = ARXIV:1309.1687;%%
  %50 citations counted in INSPIRE as of 19 Apr 2018
\bibitem{An} M.~Aganagic, N.~Haouzi and S.~Shakirov,
  ``$A_n$-Triality,''
  arXiv:1403.3657 [hep-th].
  %%CITATION = ARXIV:1403.3657;%%
  %29 citations counted in INSPIRE as of 08 Oct 2015


\bibitem{ese} 
  M.~Aganagic and A.~Okounkov,
  ``Elliptic stable envelope,''
  arXiv:1604.00423 [math.AG].
  %%CITATION = ARXIV:1604.00423;%%
  %1 citations counted in INSPIRE as of 28 Apr 2016

 \bibitem{Bethe}{M.~Aganagic and A.~ Okounkov, ``Quasimap counts and
     Bethe eigenfunctions,''  
Mosc.\ Math. J.\ \textbf{17} (2017), no.~4, 565¡V-600. 
}
 \bibitem{AS} 
  M.~Aganagic and S.~Shakirov,
  ``Gauge/Vortex duality and AGT,''
  %doi:10.1007/978-3-319-18769-3_13
  arXiv:1412.7132 [hep-th].
  %%CITATION = doi:10.1007/978-3-319-18769-3_13;%%
  %14 citations counted in INSPIRE as of 19 Apr 2018
 \bibitem{AGT}
  L.~F.~Alday, D.~Gaiotto and Y.~Tachikawa,
 ``Liouville Correlation Functions from Four-dimensional Gauge Theories,''
  Lett.\ Math.\ Phys.\  {\bf 91} (2010) 167
  [arXiv:0906.3219 [hep-th]].
  %%CITATION = ARXIV:0906.3219;%%
  %596 citations counted in INSPIRE as of 08 Oct 2015

\bibitem{Arakawa} T. Arakawa, ``Representation Theory of
  Superconformal Algebras and the Kac-Roan-Wakimoto Conjecture,'' Duke
  Math. J. {\bf 130} (2005) 435-478.

\bibitem{Arakawa1} T. Arakawa, ``Representation theory of
  $W$-algebras,'' Invent. Math. {\bf 169} (2007) 219-320.

\bibitem{Seiberg} 
  P.C.~Argyres, A.~Kapustin and N.~Seiberg,
  ``On S-duality for non-simply-laced gauge groups,''
  JHEP ${\bf 0606}$ (2006) 043
  %doi:10.1088/1126-6708/2006/06/043
  [arXiv:hep-th/0603048].
  %%CITATION = doi:10.1088/1126-6708/2006/06/043;%%
  %25 citations counted in INSPIRE as of 24 May 2016

\bibitem{AG} D. Arinkin and D. Gaitsgory, ``Singular support of
  coherent sheaves and the geometric Langlands conjecture,'' Selecta
  Math. {\bf 21} (2015) 1-199.

\bibitem{AKOS} H. Awata, H. Kubo, S. Odake, and J. Shiraishi,
  ``Quantum ${\mc W}_N$ algebras and Macdonald polynomials,'' Comm. Math. Phys.
179 (1996) 401-416.

\bibitem{ATY} H. Awata, A. Tsuchiya and Y. Yamada, ``Integral formulas
  for the WZNW correlation functions,'' Nuclear Physics {\bf B365} (1991)
  680-696.

\bibitem{beem} 
  C.~Beem, T.~Dimofte and S.~Pasquetti,
 ``Holomorphic Blocks in Three Dimensions,''
  JHEP {\bf 1412} (2014) 177.
  %doi:10.1007/JHEP12(2014)177
  [arXiv:1211.1986 [hep-th]].
  %%CITATION = doi:10.1007/JHEP12(2014)177;%%
  %91 citations counted in INSPIRE as of 01 Dec 2016

\bibitem{BF}
Behrend, K., Fantechi, B., ``The intrinsic normal cone,'' Invent. Math., {\bf 128}, {1997},{45-88}.

\bibitem{BD}
A.~Beilinson and V.~Drinfeld, ``Quantization of Hitchin's Integrable
System and Hecke Eigensheaves,'' preprint available at
www.math.uchicago.edu/$\sim$arinkin/langlands
 
 \bibitem{Bershadsky} 
  M.~Bershadsky, K.A.~Intriligator, S.~Kachru, D.R.~Morrison, V.~Sadov and C.~Vafa,
  ``Geometric singularities and enhanced gauge symmetries,''
  Nucl.\ Phys.\ {\bf B481} (1996) 215.
  %doi:10.1016/S0550-3213(96)90131-5
 % [arXiv:hep-th/9605200].
  %%CITATION = doi:10.1016/S0550-3213(96)90131-5;%%
  %360 citations counted in INSPIRE as of 30 May 2016 

\bibitem{BO} M. Bershadsky and H. Ooguri, ``Hidden $SL(n)$ symmetry in
  conformal field theory,'' Commun. Math. Phys. {\bf 126} (1989) 49-83.
  
\bibitem{BFS} R. Bezrukavnikov, M. Finkelberg and V. Schechtman,
  Factorizable Sheaves and Quantum Groups, Lect. Notes in Math. {\bf 1691},
  Springer 1998.

  \bibitem{BFN} 
  A.~Braverman, M.~Finkelberg and H.~Nakajima,
  ``Instanton moduli spaces and ${\cal W}$-algebras,''
  arXiv:1406.2381 [math.QA].
  %%CITATION = ARXIV:1406.2381;%%
  %8 citations counted in INSPIRE as of 31 Oct 2015

\bibitem{Nak13a} 
  A.~Braverman, M.~Finkelberg and H.~Nakajima,
  ``Coulomb branches of $3d$ ${\cal N}=4$ quiver gauge theories and
  slices in the affine Grassmannian'' (with appendices by A.
  Braverman, M. Finkelberg, J. Kamnitzer, R. Kodera,
  H. Nakajima, B. Webster, and A. Weekes),
  arXiv:1604.03625 [math.RT].
  %%CITATION = ARXIV:1604.03625;%%

\bibitem{Nak13c} 
   A.~Braverman, M.~Finkelberg and H.~Nakajima,
  ``Towards a mathematical definition of Coulomb branches of $3$-dimensional ${\cal N}=4$ gauge theories, II,''
  arXiv:1601.03586 [math.RT].
  %%CITATION = ARXIV:1601.03586;%%
  %6 citations counted in INSPIRE as of 23 May 2016


\bibitem{BDGH} 
  M.~Bullimore, T.~Dimofte, D.~Gaiotto and J.~Hilburn,
  ``Boundaries, Mirror Symmetry, and Symplectic Duality in 3d $\mathcal{N}=4$ Gauge Theory,''
  arXiv:1603.08382 [hep-th].
  %%CITATION = ARXIV:1603.08382;%%
  %1 citations counted in INSPIRE as of 13 May 2016


\bibitem{CV}
S.~Cecotti and C.~Vafa,
  ``Topological antitopological fusion,''
  Nucl.\ Phys.\ {\bf B367} (1991) 359.
  %doi:10.1016/0550-3213(91)90021-O
  %%CITATION = doi:10.1016/0550-3213(91)90021-O;%%
  %267 citations counted in INSPIRE as of 08 Jun 2016 ;  
  S.~Cecotti, D.~Gaiotto and C.~Vafa,
  ``$tt^*$ geometry in 3 and 4 dimensions,''
  JHEP {\bf 1405} (2014).
 % doi:10.1007/JHEP05(2014)055
 % [arXiv:1312.1008 [hep-th]];
  %%CITATION = doi:10.1007/JHEP05(2014)055;%%
  %18 citations counted in INSPIRE as of 08 Jun 2016
  S.~Cecotti, A.~Neitzke and C.~Vafa,
  ``Twistorial Topological Strings and a tt* Geometry for N=2 Theories in 4d,''
  arXiv:1412.4793 [hep-th];
  %%CITATION = ARXIV:1412.4793;%%
  %3 citations counted in INSPIRE as of 08 Jun 2016,
\bibitem{CF3} Ciocan-Fontanine, I., Kapranov, M., ``{Virtual fundamental classes via dg-manifolds},'' {Geom. Topol.}, {\bf 13}, {(2009)}, {1779-1804}


\bibitem{CKM} {Ciocan-Fontanine, I.}, {Kim, B.}, {Maulik, D.}, ``{Stable quasimaps to GIT quotients},'' {J. Geom. Phys.}, {\bf 75} {(2014)}, {17-47}.


\bibitem{ConProc}
C.~De Concini and C.~Procesi, 
\emph{On the geometry of toric arrangements},  
Transform.\ Groups \textbf{10} (2005) 387-422. 


\bibitem{D} 
  N.~Dorey, C.~Fraser, T.~J.~Hollowood and M.~A.~C.~Kneipp,
  ``S duality in N=4 supersymmetric gauge theories with arbitrary gauge group,''
  Phys.\ Lett.\ B {\bf 383} (1996) 422.
  %[arXiv:hep-th/9605069].
  %%CITATION = doi:10.1016/0370-2693(96)00773-3;%%
  %44 citations counted in INSPIRE as of 29 May 2016


\bibitem{DM}
  M.~R.~Douglas and G.~W.~Moore,
  ``D-branes, quivers, and ALE instantons,''
  arXiv:hep-th/9603167.
  %%CITATION = HEP-TH/9603167;%%
  %1019 citations counted in INSPIRE as of 28 Sep 2015

\bibitem{Drinfeld} V. Drinfeld, ``Hopf algebras and the quantum
  Yang-Baxter equation,'' Dokl. Soviet Math. {\bf 32} (1985) 254-258.

\bibitem{EFK} P.I. Etingof, I.B. Frenkel, and A.A. Kirillov, Jr.,
  Lectures on Representation Theory and Knizhnik-Zamolodchikov
  Equations, Math. Surveys and Monographs {\bf 58}, Amer. Math. Soc.,
  Providence, RI, 1998.

\bibitem{EV} Etingof, P., {Varchenko, A.},
``{Dynamical Weyl groups and applications},''
{Adv. Math.},
{\bf 167},
{(2002)},
{74-127}.

\bibitem{FG}{Fantechi, B.},
{G{\"o}ttsche, L.},
``{Riemann-Roch theorems and elliptic genus for virtually smooth
   schemes},''
{Geom. Topol.},{\bf 14},
 {(2010)},
{83-115}.
   
   \bibitem{FF:ds} B. Feigin and E. Frenkel, ``Quantization of the
  Drinfeld-Sokolov reduction,'' Phys. Lett. {\bf B246} (1990) 75-81.

\bibitem{FF} B. Feigin and E. Frenkel, ``Affine Kac-Moody algebras at
  the critical level and Gelfand-Dikii algebras,''
  Int. J. Mod. Phys. {\bf A7}, Suppl. 1A (1992) 197-215 (reprinted in
  ``Infinite analysis'' (Kyoto, 1991), Adv. Ser. Math. Phys. 16, World
  Scientific, 1992).

\bibitem{FF:laws} B. Feigin and E. Frenkel, ``Integrals of motion and
  quantum groups,'' Proceedings of the C.I.M.E. School Integrable
    Systems and Quantum Groups, Italy, June 1993, Lect. Notes in
  Math. {\bf 1620}, Springer, Berlin, 1995 [arXiv:hep-th/9310022].

\bibitem{FF:qw} B. Feigin and E. Frenkel, ``Quantum ${\mathcal
    W}$--algebras and elliptic algebras,'' Comm. Math. Phys. {\bf 178}
  (1996) 653-678 [arXiv:q-alg/9508009].

\bibitem{FFR} B. Feigin, E. Frenkel and N. Reshetikhin, ``Gaudin
  model, Bethe ansatz and critical level,'' Comm. Math. Phys.
    {\bf 166} (1994) 27-62.

\bibitem{Felder} G.~Felder, ``Conformal field theory and integrable
  systems associated to elliptic curves,'' arXiv:hep-th/9407154.
 
\bibitem{F} 
  E.~Frenkel,
  ``${\mc W}$-algebras and Langlands-Drinfeld correspondence,'' in
New Symmetries in Quantum Field Theory,
eds. J. Fr\"{o}hlich, e.a., pp. 433-447, New York: Plenum Press, 1992.
  
\bibitem{Frenkel1}
E.~Frenkel, ``Affine Algebras, Langlands Duality and Bethe Ansatz,'' in
Proceedings of the  International Congress of Mathematical Physics,
Paris, 1994, ed. D. Iagolnitzer,  pp. 606-642, International Press,
1995 [arXiv:q-alg/9506003].

\bibitem{Frenkel2} E.~Frenkel, ``Wakimoto modules, opers and the center
  at the critical level,'' Adv. Math. {\bf 195} (2005) 297-404
  [arXiv:math/0210029].

\bibitem{Frenkel}
{E.~{Frenkel}}, ``Lectures on the Langlands Program and Conformal
Field Theory,'' in Frontiers in Number Theory, Physics and Geometry
    II, eds. P. Cartier, e.a., pp. 387-536, Springer, 2007
[arXiv:hep-th/0512172].

\bibitem{F:bourbaki} E. Frenkel, ``Gauge theory and Langlands
duality,'' S\'eminaire Bourbaki, Ast\'erisque {\bf 332} (2010) 369-403
[arXiv:0906.2747].

\bibitem{FB} E.~Frenkel and D.~Ben-Zvi, {\em Vertex Algebras and
Algebraic Curves}, Mathematical Surveys and
Monographs {\bf 88}, Second Edition, AMS, 2004.

\bibitem{FGT}
  E.~Frenkel, S.~Gukov and J.~Teschner,
  ``Surface Operators and Separation of Variables,'' JHEP (2016) 179
  [arXiv:1506.07508 [hep-th]].

\bibitem{FKW} E. Frenkel, V. Kac and M. Wakimoto, ``Characters and
  fusion rules of W-algebras via quantized Drinfeld-Sokolov
  reduction,'' Comm. Math. Phys. {\bf 147} (1992) 295-328.

\bibitem{FR:dca} E.~Frenkel, and N.~{Reshetikhin}, ``Towards Deformed
  Chiral Algebras,'' in Proceedings of the
  Quantum Group Symposium at the XXIth International Colloquium on
  Group Theoretical Methods in Physics (Goslar, 1996), eds. H.-D. Doebner
  and V.K. Dobrev, pp. 27-42, Heron Press, Sofia, 1997
  [arXiv:q-alg/9706023].

\bibitem{FR} E.~Frenkel, and N.~{Reshetikhin}, ``Deformations of
    W-algebras associated to simple Lie algebras,''
    Comm. Math. Phys. {\bf 197}
(1998) 1-32 [arXiv:q-alg/9708006].

\bibitem{FIR}
I.B.~Frenkel, N.Yu.~Reshetikhin, ``Quantum affine algebras and
holonomic difference equations,'' Comm. Math. Phys. {\bf 146} (1992) 1-60.

\bibitem{FGPP} P. Furlan, A. Ch. Ganchev, R. Paunov, and
    V. B. Petkova, ``Solutions of the Knizhnik - Zamolodchikov
      Equation with Rational Isospins and the Reduction to the Minimal
      Models,'' Nucl. Phys. {\bf B394} (1993) 665-706.

 \bibitem{GK} 
  D.~Gaiotto and P.~Koroteev,
  ``On Three Dimensional Quiver Gauge Theories and Integrability,''
  JHEP {\bf 1305} (2013) 126
  %doi:10.1007/JHEP05(2013)126
  [arXiv:1304.0779 [hep-th]].
  %%CITATION = doi:10.1007/JHEP05(2013)126;%%
  %40 citations counted in INSPIRE as of 28 Nov 2016

 \bibitem{GW}
  D.~Gaiotto and E.~Witten,
  ``Knot Invariants from Four-Dimensional Gauge Theory,''
  Adv.\ Theor.\ Math.\ Phys.\  {\bf 16} (2012) 935
  [arXiv:1106.4789 [hep-th]].
  %%CITATION = ARXIV:1106.4789;%%
  %46 citations counted in INSPIRE as of 28 Sep 2015
  

\bibitem{gaitsW} D.~Gaitgory, ``Twisted Whittaker model and
  factorizable sheaves,'' Sel. Math., New ser. {\bf 13} (2008) 617
  [arXiv:0705.4571].

\bibitem{gaitsQ} D.~Gaitsgory, ``Quantum Langlands Correspondence,''
  arXiv:1601.05279.


\bibitem{GinzNak}
 {Ginzburg, V.},
 ``{Lectures on Nakajima's quiver varieties}",
{Geometric methods in representation theory. I},
{S\'emin. Congr.}, {\bf 24}, {Soc. Math. France, Paris}, {(2012)},
   pages={145-219},
   %review={\MR{3202703}},



\bibitem{Giv}
{Givental, A.}, ``On the WDVV equation in quantum $K$-theory",''
Dedicated to William Fulton on the occasion of his 60th birthday,
{Michigan Math. J.},
{\bf 48},
{(2000)},
{295-304}.

\bibitem{GP}
{Graber, T.},
{Pandharipande, R.},
``Localization of virtual classes,''
{Invent. Math.},
{\bf 135},
{(1999)},
%  number={2},
{487-518}.
 
 
  \bibitem{GuW} 
  S.~Gukov and E.~Witten,
  ``Gauge Theory, Ramification, And The Geometric Langlands Program,''
Current Developments in Mathematics {\bf 2006} (2008) 35-180.
%[arXiv:hep-th/0612073].
  %%CITATION = HEP-TH/0612073;%%
  %218 citations counted in INSPIRE as of 19 May 2016
  

  
  
\bibitem{DHL}
{Halpern-Leistner, D.},
``The derived category of a GIT quotient,''
{J. Amer. Math. Soc.},
{\bf 28}, {(2015)},
%  number={3},
{871-912}.

\bibitem{HL} Y.-Z. Huang and J. Lepowsky, ``Tensor categories and the
  mathematics of rational and logarithmic conformal field theory,''
  Journal of Physics {\bf A46} (2013) 494009 [arXiv:1304.7556
  [hep-th]].

\bibitem{Jimbo} M. Jimbo, ``A $q$-difference analogue of
  $\mathcal{U}({\mathfrak{g}})$ and Yang-Baxter equation,''
  Lett. Math. Phys. {\bf 10} (1985) 63-69.

\bibitem{Kapustin} 
  A.~Kapustin,
  ``A Note on Quantum Geometric Langlands Duality, Gauge Theory, and
  Quantization of the Moduli Space of Flat Connections,''
  arXiv:0811.3264 [hep-th].
  %%CITATION = ARXIV:0811.3264;%%
  %8 citations counted in INSPIRE as of 04 Oct 2016


 \bibitem{KW}
  A.~Kapustin and E.~Witten,
  ``Electric-Magnetic Duality And The Geometric Langlands Program,''
  Commun.\ Num.\ Theor.\ Phys.\  {\bf 1} (2007) 1
  [arXiv:hep-th/0604151].
  %%CITATION = HEP-TH/0604151;%%
  %306 citations counted in INSPIRE as of 20 Aug 2015

\bibitem{KL} D. Kazhdan and G. Lusztig, ``Tensor structures arising
from affine Lie algebras,'' Journal of AMS {\bf 6} (1993) 905-947.

\bibitem{KimPest} 
  T.~Kimura and V.~Pestun,
  ``Quiver W-algebras,''
  arXiv:1512.08533 [hep-th].
  %%CITATION = ARXIV:1512.08533;%%
  %7 citations counted in INSPIRE as of 13 May 2016
   
\bibitem{KZ} 
  V.~G.~Knizhnik and A.~B.~Zamolodchikov,
 ``Current Algebra and Wess-Zumino Model in Two-Dimensions,''
  Nucl.\ Phys.\ {\bf B247} (1984) 83.
 % doi:10.1016/0550-3213(84)90374-2
  %%CITATION = doi:10.1016/0550-3213(84)90374-2;%%
  %1348 citations counted in INSPIRE as of 02 May 2016
  
  

\bibitem{KN}
 K. Kuroki, A. Nakayashiki, 
``Free Field Approach to Solutions of the Quantum
Knizhnik--Zamolodchikov Equations,'' SIGMA {\bf 4} (2008) 049.

\bibitem{Lee} {Lee, Y.-P.}, ``Quantum $K$-theory. I. Foundations,''
{Duke Math. J.},
{\bf 121}
{(2004)},
%  number={3},
{389-424}.

\bibitem{Laumon} G. Laumon, ``Transformation de Fourier
g\'{e}n\'{e}ralis\'{e}e,'' Preprint alg-geom/9603004.

\bibitem{Li1}
   {Li, J.},
   ``Stable morphisms to singular schemes and relative stable
   morphisms,''
   {J. Differential Geom.},
   {\bf 57}
   {(2001)},
%D   number={3},
   {509-578}.
   
\bibitem{Li2}
  {Li, J.},
   ``A degeneration formula of GW-invariants,''
   {J. Differential Geom.},
   {\bf 60} {(2002)},
%  number={2},
  {199-293}.
  
\bibitem{LiWu} {Li, J.}, {Wu, B.},
   ``Good degeneration of Quot-schemes and coherent systems,''
 {Comm. Anal. Geom.},
{\bf 23}  {(2015)},
%  number={4}, 
{841-921}.

 \bibitem{LMN} 
  A.~S.~Losev, A.~Marshakov and N.~A.~Nekrasov,
  ``Small instantons, little strings and free fermions,''
  in From Fields to Strings, eds. Shifman, M., e.a., vol. 1, pp. 581-621
  [arXiv:hep-th/0302191].
  %%CITATION = HEP-TH/0302191;%%
  %138 citations counted in INSPIRE as of 31 May 2016 
 
 \bibitem{Mm} 
  A.~Losev, G.~W.~Moore and S.~L.~Shatashvili,
  ``M \& m's,''
  Nucl.\ Phys.\ {\bf B522} (1998) 105
  %doi:10.1016/S0550-3213(98)00262-4
  [arXiv:hep-th/9707250].
  %%CITATION = doi:10.1016/S0550-3213(98)00262-4;%%
  %92 citations counted in INSPIRE as of 22 Aug 2016
   
  \bibitem{Matsuo}
  A. Matsuo, ``Jackson Integrals of Jordan-Pochhammer Type and
  Quantum Knizhnik-Zamolodchikov Equations,'' Comm. Math. Phys. {\bf
    151} (1993) 263-273

  \bibitem{Matsuo3}
  A. Matsuo, ``Quantum algebra structure of certain Jackson
  integrals,'' Comm. Math. Phys. {\bf 157} (1993) 479-498. 
    
\bibitem{MO} 
  D.~Maulik and A.~Okounkov,
  ``Quantum Groups and Quantum Cohomology,''
  arXiv:1211.1287, to appear in Ast\'erisque. 
  %%CITATION = ARXIV:1211.1287;%%
  %49 citations counted in INSPIRE as of 08 Oct 2015

\bibitem{mn}
K.~McGerty and T.~Nevins, ``Kirwan surjectivity for quiver varieties,''
arXiv:1610.08121. 

      
\bibitem{Nakajima}
H.~Nakajima, ``Quiver varieties and finite-dimensional representations of quantum
   affine algebras,''
{J. Amer. Math. Soc.}, {\bf 14}, (2001) 145--238.

\bibitem{Nak13b}  H.~Nakajima,
  ``Towards a mathematical definition of Coulomb branches of $3$-dimensional ${\cal N}=4$ gauge theories, I,''
  arXiv:1503.03676 [math-ph].
  %%CITATION = ARXIV:1503.03676;%%
  %11 citations counted in INSPIRE as of 23 May 2016

\bibitem{BPS} N. Nekrasov, ``On the BPS/CFT correspondence,'' seminar
  at the University of Amsterdam, Feb 2004.

\bibitem{Nekrasov} 
  N.~A.~Nekrasov,
  ``Instanton partition functions and M-theory,''  Proceedings of the
15th International Seminar on High Energy
                        Physics (Quarks 2008).                   
                          %%CITATION = INSPIRE-1372940;%%
\bibitem{Nekrasov:2015wsu} 
  N.~Nekrasov,
  ``BPS/CFT correspondence: non-perturbative Dyson-Schwinger equations and qq-characters,''
  JHEP {\bf 1603} (2016) 181
  %doi:10.1007/JHEP03(2016)181
  [arXiv:1512.05388 [hep-th]].
  %%CITATION = doi:10.1007/JHEP03(2016)181;%%
  %22 citations counted in INSPIRE as of 28 Dec 2016

\bibitem{Nekrasov:2016qym} 
  N.~Nekrasov,
  ``BPS/CFT correspondence II: Instantons at crossroads, Moduli and Compactness Theorem,''
  arXiv:1608.07272 [hep-th].
  %%CITATION = ARXIV:1608.07272;%%
  %4 citations counted in INSPIRE as of 28 Dec 2016

\bibitem{Nekrasov:2016ydq} 
  N.~Nekrasov,
  ``BPS/CFT Correspondence III: Gauge Origami partition function and qq-characters,''
  arXiv:1701.00189 [hep-th].
  %%CITATION = ARXIV:1701.00189;%%

\bibitem{NO1} 
  N.~Nekrasov and A.~Okounkov,
  ``Seiberg-Witten theory and random partitions,''
  Prog.\ Math.\  {\bf 244} (2006) 525
%  doi:10.1007/0-8176-4467-9_15
  [arXiv:hep-th/0306238].
  %%CITATION = doi:10.1007/0-8176-4467-9_15;%%
  %482 citations counted in INSPIRE as of 31 May 2016
 
 \bibitem{NO2} 
  N.~Nekrasov and A.~Okounkov,
  ``Membranes and Sheaves,''
  arXiv:1404.2323 [math.AG].
  %%CITATION = ARXIV:1404.2323;%%
  %17 citations counted in INSPIRE as of 08 Jan 2016
  
  \bibitem{NP} 
  N.~Nekrasov and V.~Pestun,
  ``Seiberg-Witten geometry of four dimensional N=2 quiver gauge theories,''
  arXiv:1211.2240 [hep-th].
  %%CITATION = ARXIV:1211.2240;%%
  %63 citations counted in INSPIRE as of 30 May 2016

\bibitem{Rosly} 
  N.~Nekrasov, A.~Rosly and S.~Shatashvili,
  ``Darboux coordinates, Yang-Yang functional, and gauge theory,''
  Nucl.\ Phys.\ Proc.\ Suppl.\  {\bf 216} (2011) 69
 % doi:10.1016/j.nuclphysbps.2011.04.150
  [arXiv:1103.3919 [hep-th]].
  %%CITATION = doi:10.1016/j.nuclphysbps.2011.04.150;%%
  %80 citations counted in INSPIRE as of 01 Jun 2016
  
\bibitem{NW} 
  N.~Nekrasov and E.~Witten,
  ``The Omega Deformation, Branes, Integrability, and Liouville Theory,''
  JHEP {\bf 1009} (2010) 092
  [arXiv:1002.0888 [hep-th]].
  %%CITATION = ARXIV:1002.0888;%%
  %113 citations counted in INSPIRE as of 20 Aug 2015

 \bibitem{OK}
A.~Okounkov, ``Lectures on K-theoretic computations in enumerative
geometry,'' Geometry of moduli spaces and representation theory, 251-¡V380, 
IAS/Park City Math.\ Ser., 24, AMS 2017. 

\bibitem{OS}
A.~Okounkov and A.~Smirnov, ``Quantum difference equations for Nakajima varieties,''
arXiv:1602.09007.

   \bibitem{OV} 
  H.~Ooguri and C.~Vafa,
  ``Knot invariants and topological strings,''
  Nucl.\ Phys.\ {\bf B577} (2000) 419
  [arXiv:hep-th/9912123].
  %%CITATION = HEP-TH/9912123;%%
  %310 citations counted in INSPIRE as of 08 Oct 2015

\bibitem{PRY} J.L. Petersen, J. Rasmussen, and M. Yu,
     ``Conformal Blocks for admissible representations in $SL(2)$
     current algebra,'' Nucl.Phys. {\bf B457} (1995) 309-342.

\bibitem{Tasi} 
  J.~Polchinski,
  ``Tasi lectures on D-branes,''
  arXiv:hep-th/9611050.
  %%CITATION = HEP-TH/9611050;%%
  %1192 citations counted in INSPIRE as of 01 Dec 2016

\bibitem{PR}
A. Polishchuk and M. Rothstein, ``Fourier transform for D-algebras,''
Duke Math. J. {\bf 109} (2001) 123-146.

\bibitem{Raskin} S. Raskin, ``${\mc W}$-algebras and Whittaker
categories,'' arXiv:1611.04937.

\bibitem{Raskin:ch} S. Raskin, ``Chiral categories,'' available at
  http://math.mit.edu/$\sim$sraskin/chiralcats.pdf



\bibitem{Reshetikhin} N.~Reshetikhin, ``Quantized Universal Enveloping
  Algebras, The Yang-Baxter Equation and Invariants of Links, I and
  II,'' LOMI Preprints, E-4-87 and E-17-87, USSR Academy of Sciences.

\bibitem{ReshqKZ} N.~Reshetikhin, 
``Jackson-type integrals, Bethe vectors, and solutions to a difference
analog of the Knizhnik-Zamolodchikov system'', 
 Lett.\ Math.\ Phys.\ 26 (1992), no.~3, 153¡V-165. 

\bibitem{Rothstein} M. Rothstein, ``Connections on the total Picard
    sheaf and the KP hierarchy,'' Acta Applicandae Mathematicae {\bf
    42} (1996) 297-308.

\bibitem{Sch} V. Schechtman, ``Dualit\'e de Langlands quantique,''
  Ann. Fac. Sci. Toulouse {\bf 23} (2014) 129--158.

\bibitem{SV} V. Schechtman and A. Varchenko, ``Arrangements of
  hyperplanes and Lie algebra homology,'' Invent. Math. {\bf 106}
  (1991) 139-194.

\bibitem{SV1} V. Schechtman and A. Varchenko, ``Quantum groups and
  homology of local systems,'' in Algebraic Geometry and Analytic
  Geometry, ICM-90 Satellite Conference Proceedings,
pp. 182-197, Springer, 1991.

\bibitem{SeibergN} 
  N.~Seiberg,
  ``New theories in six-dimensions and matrix description of M theory on $T^5$ and $T^5/Z_2$,''
  Phys.\ Lett.\ {\bf B408} (1997) 98
 % doi:10.1016/S0370-2693(97)00805-8
  [arXiv:hep-th/9705221].
  %%CITATION = doi:10.1016/S0370-2693(97)00805-8;%%
  %366 citations counted in INSPIRE as of 01 Jun 2016
  
\bibitem{SKAO} J. Shiraishi, H. Kubo, H. Awata, S.  Odake, ``A
  quantum deformation of the Virasoro algebra and the Macdonald
  symmetric functions,'' Lett. Math. Phys. {\bf 38} (1996) 33-51.

\bibitem{Sm_cap}
A.~Smirnov, ``Rationality of capped descendent vertex in K-theory,''
arXiv:1612.01048, and in preparation.

\bibitem{FFS} 
A.~{Stoyanovsky}, ``On Quantization of the Geometric Langlands Correspondence I,''
arXiv:math/9911108.

\bibitem{Stoyan} A. Stoyanovsky, {\em A Relation Between the
  Knizhnik-Zamolodchikov and Belavin-Polyakov-Zamolodchikov Systems of
  Partial Differential Equations}, arXiv:math-ph/0012013v3.

\bibitem{Stoyanovsky} A.~Stoyanovsky, ``Quantum Langlands duality and
  conformal field theory,'' arXiv:math/0610974.

\bibitem{TV1}
 {Tarasov, V.},
 {Varchenko, A.},
  ``Geometry of $q$-hypergeometric functions as a bridge between
   Yangians and quantum affine algebras,''
   {Invent. Math.},
{\bf 128},
{(1997)},
%  number={3},
{501-588}.


\bibitem{TV} {Tarasov, V.}, {Varchenko, A.},
   ``Geometry of $q$-hypergeometric functions, quantum affine algebras
   and elliptic quantum groups,''
   %language={English, with English and French summaries},
{Ast\'erisque},
 {\bf 246}
  {(1997)}, {vi+135}.


\bibitem{TV4}
{Tarasov, V.},
{Varchenko, A.},
``Combinatorial formulae for nested Bethe vectors,''
{SIGMA Symmetry Integrability Geom. Methods Appl.},
{\bf 9}, {(2013)}.

\bibitem{TV5}
{Tarasov, V.},
{Varchenko, A.},
  ``Jackson integral representations for solutions of the
   Knizhnik-Zamolodchikov quantum equation,''
{Algebra i Analiz}, {\bf 6} {(1994)},
 {90-137}; {St. Petersburg Math. J.}, {\bf 6} {(1995)},{2}, {275--313}.


\bibitem{Tel} C.~Teleman, ``The quantization conjecture revisited,''
 Ann.\ of Math., {\bf 152}, (2000) 1-43.


\bibitem{JT} 
  J.~Teschner,
  ``Quantization of the Hitchin moduli spaces, Liouville theory, and
  the geometric Langlands correspondence I,''
  Adv.\ Theor.\ Math.\ Phys.\  {\bf 15} (2011) 471.
%  [arXiv:1005.2846 [hep-th]].
  %%CITATION = ARXIV:1005.2846;%%
  %74 citations counted in INSPIRE as of 22 Oct 2015

  \bibitem{Vafa} 
  C.~Vafa,
  ``Geometric origin of Montonen-Olive duality,''
  Adv.\ Theor.\ Math.\ Phys.\  {\bf 1} (1998) 158
  [arXiv:hep-th/9707131].
  %%CITATION = HEP-TH/9707131;%%
  %35 citations counted in INSPIRE as of 24 May 2016

\bibitem{Varchenko}
A. Varchenko, ``Quantized Knizhnik-Zamolodchikov equations,
quantum Yang-Baxter equation, and difference equations for
$q$-hypergeometric functions,'' Comm. Math. Phys. {\bf 162} (1994) 499-528.  
  
\bibitem{Varchenko1} A. Varchenko, Multidimensional Hypergeometric
  Functions and Representation Theory of Lie Algebras and Quantum
  Groups, World Scientific, 1995.

\bibitem{Wittengrass} E. Witten, ``Quantum field theory, Grassmannians,
and algebraic curves,'' Comm. Math. Phys. {113} (1988) 529-600.

\bibitem{wm} 
  E.~Witten,
  ``Mirror manifolds and topological field theory,''
  in Mirror symmetry I, ed. S.-T. Yau, pp. 121-160
  [arXiv:hep-th/9112056].
  %%CITATION = HEP-TH/9112056;%%
  %338 citations counted in INSPIRE as of 01 Dec 2016
  
  
  
\bibitem{Wp} 
  E.~Witten,
  ``Phases of N=2 theories in two-dimensions,''
  Nucl.\ Phys.\ B {\bf 403} (1993) 159
 % doi:10.1016/0550-3213(93)90033-L
  [arXiv:hep-th/9301042].
  %%CITATION = doi:10.1016/0550-3213(93)90033-L;%%



\bibitem{Wittenl1} 
  E.~Witten,
  ``Geometric Langlands From Six Dimensions,''
  arXiv:0905.2720 [hep-th].
  %%CITATION = ARXIV:0905.2720;%%
  %69 citations counted in INSPIRE as of 20 ao¡Pt 2015

\bibitem{Wittenl3} 
  E.~Witten,
  ``Geometric Langlands And The Equations Of Nahm And Bogomolny,''
  arXiv:0905.4795 [hep-th].
  %%CITATION = ARXIV:0905.4795;%%
  %18 citations counted in INSPIRE as of 01 Jun 2016

\bibitem{WF} 
  E.~Witten,
  ``Fivebranes and Knots,''
  arXiv:1101.3216 [hep-th].
  %%CITATION = ARXIV:1101.3216;%%
  %96 citations counted in INSPIRE as of 28 Sep 2015
 
\bibitem{WpI}
  E.~Witten,
  ``A New Look At The Path Integral Of Quantum Mechanics,''
  arXiv:1009.6032 [hep-th].
  %%CITATION = ARXIV:1009.6032;%%
  %67 citations counted in INSPIRE as of 23 May 2016
  
 
\bibitem{Wittenl2} 
  E.~Witten,
  ``More On Gauge Theory And Geometric Langlands,''
  arXiv:1506.04293 [hep-th].
  %%CITATION = ARXIV:1506.04293;%%
 
  
\bibitem{Wood} C.~Woodward, ``Moment maps and geometric invariant theory,''
arXiv:0912.1132. 

\end{thebibliography}
\end{document}